\documentclass[format=acmsmall, review=false, screen=true, draft, nonacm]{acmart}

\usepackage{enumitem}
\usepackage{mathrsfs}
\usepackage{xspace}
\usepackage[utf8]{inputenc}
\usepackage{cleveref}
\usepackage{stackengine}
\usepackage{subcaption}
\usepackage[final]{listings}
\usepackage{tikz}
\usetikzlibrary{shapes,calc,arrows,automata,decorations.pathmorphing}
\usepackage{algorithm}
\usepackage{xifthen}
\usepackage[shortcuts]{extdash}
\usepackage{centernot}
\usepackage{multirow}
\usepackage{pbox}
\usepackage{hhline}
\usepackage{adjustbox}

\DeclareMathAlphabet{\mathpzc}{OT1}{pzc}{m}{it}

\makeatletter
\DeclareFontEncoding{LS1}{}{}
\DeclareFontSubstitution{LS1}{stix}{m}{n}
\DeclareMathAlphabet{\mathscr}{LS1}{stixscr}{m}{n}
\makeatother

\newcommand{\dom}{\fun{dom}}
\newcommand{\img}{\fun{img}}
\newcommand{\proc}{\fun{proc}}
\newcommand{\bound}{b}
\newcommand{\TV}{\mathcal{TV}}
\newcommand{\update}{\fun{update}}
\newcommand{\dest}{\fun{target}}

\newcommand{\size}[1]{\left|#1\right|}
\newcommand{\dht}{\fun{dh}}
\newcommand{\rc}{\fun{rc}}
\newcommand{\code}[1]{\lstinline{#1}}
\newcommand{\tv}{\mathit{tv}}
\newcommand{\itr}{\mathrm{it}}
\newcommand{\tool}[1]{\textsf{#1}\xspace}

\newcommand{\loat}{\tool{LoAT}}
\newcommand{\koat}{\tool{KoAT}}
\newcommand{\pl}[1]{\textsf{#1}}
\renewcommand{\phi}{\varphi}
\renewcommand{\infty}{\omega}
\newcommand{\constr}[1]{\ \left[#1\right]}
\newcommand{\VV}{\mathcal{V}}
\newcommand{\PP}{\mathcal{P}}
\newcommand{\TT}{\mathcal{T}}
\newcommand{\OO}{\mathcal{O}}
\newcommand{\CC}{\mathcal{C}}
\newcommand{\TTT}{T}
\newcommand{\SSS}{S}
\newcommand{\QQQ}{Q}
\newcommand{\ZZ}{\mathbb{Z}}
\newcommand{\NN}{\mathbb{N}}
\newcommand{\RR}{\mathbb{R}}
\newcommand{\vect}[1]{\overline{#1}}
\newcommand{\cc}{c}
\newcommand{\concretecosts}{k}
\newcommand{\head}{\fun{root}}

\newcommand{\fun}[1]{\mathrm{#1}}
\newcommand{\lhs}{\fun{lhs}}
\newcommand{\rhs}{\fun{rhs}}
\newcommand{\guard}{\fun{guard}}
\newcommand{\cost}{\fun{cost}}
\newcommand{\smt}{\fun{smt}}
\newcommand{\fs}[1]{\mathsf{#1}}
\newcommand{\Ff}{\fs{f}}
\newcommand{\Fg}{\fs{g}}
\newcommand{\charfun}[1]{\llbracket#1\rrbracket}

\newcommand{\true}{\mathsf{true}}
\newcommand{\false}{\mathsf{false}}

\newcommand{\tox}[4][-2pt]{
  \xxrightarrow{\raisebox{#1}[0pt][0pt]{$\scriptstyle #2$}}^{\hspace{-1pt}#3}_{\hspace{-1pt}#4}
}
\newcommand{\toxx}[4][-2pt]{
  \xxrightarrow{\raisebox{#1}[0pt][3pt]{$\scriptstyle #2$}}^{\hspace{-1pt}#3}_{\hspace{-1pt}#4}
}

\newcommand\xleadsto[1]{%
  \mathrel{%
    \begin{tikzpicture}[%
      baseline={([yshift=-2pt]current bounding box.south)}
      ]
      \node[%
      ,inner sep=.44ex
      ,align=center
      ] (tmp) {$\scriptstyle #1$};
      \path[%
      ,draw,<-
      ,semithick
      ,decorate,decoration={%
        ,snake
        ,amplitude=1pt
        ,segment length=2.5mm
        ,pre length=0.7mm
      }
      ]
      (tmp.south east) -- (tmp.south west);
    \end{tikzpicture}
  }
}

\newcommand\leadstox[1]{\mathrel{\smash{%
      \setbox2=\hbox{$\scriptstyle #1$}%
      \xleadsto{\makebox[\dimexpr\wd2\relax]{$\scriptstyle #1$}}%
    }}}

\renewcommand{\leadsto}{\leadstox{\phantom{\rightarrow}}}

\stackMath
\newcommand\xxrightarrow[2][]{\mathrel{\smash{%
    \setbox2=\hbox{\stackon{\scriptstyle#1}{#2}}%
    \stackunder[0pt]{%
      \xrightarrow{\makebox[\dimexpr\wd2\relax]{$#2$}}%
    }{%
      \scriptstyle#1\,%
    }%
  }}}
\newenvironment{claims}{\begin{enumerate}[label*=\textit{Claim \arabic*.}, wide, leftmargin=1em, ref=\textit{Claim \arabic*}]}{\end{enumerate}}

\newcommand{\claim}[1]{\item #1\\}

\newcounter{sectionctr}
\newcounter{lemmactr}
\newcounter{auxctr}
\newcounter{eq:finalFirstLeadingExCtr}
\newcounter{eq:unboundedAcceleratedCtr}
\newcounter{eq:FibonacciAcceleratedCtr}

\makeatletter
\renewenvironment{cases}[1][l]{\matrix@check\cases\env@cases{#1}}{\endarray\right.}
\def\env@cases#1{%
  \let\@ifnextchar\new@ifnextchar
  \left\lbrace\def\arraystretch{1.2}%
  \array{@{}#1@{\quad}l@{}}}
\makeatother

\acmJournal{TOPLAS}

\begin{document}
\title{Inferring Lower Runtime Bounds for Integer Programs}

\author{Florian Frohn}
\affiliation{%
  \institution{Max Planck Institute for Informatics}
  \streetaddress{Saarland Informatics Campus,
    Campus E1 4}
  \city{66123 Saarbrücken}
  \country{Germany}}
\email{florian.frohn@mpi-inf.mpg.de}
\author{Matthias Naaf}
\affiliation{%
  \institution{RWTH Aachen University}
  \city{Aachen}
  \country{Germany}
}
\email{naaf@logic.rwth-aachen.de}
\author{Marc Brockschmidt}
\affiliation{%
  \institution{Microsoft Research}
  \city{Cambridge}
  \country{UK}}
\email{mabrocks@microsoft.com}
\author{Jürgen Giesl}
\affiliation{%
  \institution{LuFG Informatik 2, RWTH Aachen University}
  \streetaddress{Ahornstr.\ 55}
  \city{52074 Aachen}
  \country{Germany}}
\email{giesl@informatik.rwth-aachen.de}

\begin{abstract}
  We present a technique to
  infer \emph{lower} bounds on the worst-case runtime complexity of integer
  programs, where in contrast to earlier work, our approach is not restricted to
  tail-recursion.
   Our technique constructs symbolic representations of program
  executions using a framework for iterative, under-approximating program
  simplification.
  The core of this simplification is a method for
  (under-approximating) program acceleration based on recurrence solving and a
  variation of ranking functions.
    Afterwards,
  we deduce \emph{asymptotic} lower
  bounds from the resulting simplified programs using a special-purpose calculus
  and an SMT encoding. We implemented our
  technique in our tool \loat and show
  that it infers non-trivial lower bounds for a large class of examples.
\end{abstract}



\maketitle

\section{Introduction}
\label{sec:intro}

Recent advances in program analysis yield efficient methods to find \emph{upper}
bounds on the complexity of sequential integer programs.  Here, one usually
considers ``worst-case complexity'', i.e., for any variable valuation, one
analyzes the length of the longest execution starting from that valuation.  But
in many cases, in addition to upper bounds, it is also important to find
\emph{lower} bounds for this notion of complexity.  Together with an analysis
for upper bounds, this can be used to infer \emph{tight} complexity bounds.
Lower bounds also have important applications in security analysis.
If one can infer that there exists a family of inputs which lead to unacceptably large
runtime of the program (i.e., a family for which there is  a non-linear or probably even exponential
lower bound on the runtime), then this family of inputs
represents a possible denial-of-service attack.
Thus, techniques for the computation of lower bounds on the worst-case complexity can be used to detect
such attacks.\footnote{In a joint project \emph{CAGE}
  \cite{STAC,CAGE} with Draper Inc.\ (\url{https://www.draper.com}) and the University of Innsbruck, we used our tool \loat (that implements
  the techniques described in this paper) together with our tool \koat \cite{koat} (that
infers upper runtime bounds)
to analyze the complexity of large
\pl{Java} programs in order to detect vulnerabilities.}

While worst-case lower bounds are useful to \emph{detect} attacks or performance
bugs, \emph{worst-case upper bounds} can \emph{prove the absence} of such problems. \emph{Best-case} lower bounds, which
have also been investigated in the literature (see \Cref{sec:related}), are bounds on
 all program runs, whereas
worst-case lower bounds hold for (usually infinite) families of (expensive) program runs. Thus, best-case lower
bounds can, e.g., be used to decide whether a certain task is expensive enough to compensate the overhead of executing
it remotely. So in general, the use cases of worst-case and best-case bounds are orthogonal.

We introduce the first technique to infer worst-case lower bounds for
integer programs automatically. Besides \emph{concrete} bounds, our technique
can also deduce \emph{asymptotic} bounds. In general, concrete lower bounds that
hold for arbitrary variable valuations are difficult to express concisely.  In
contrast, asymptotic bounds are easily understood by humans and witness possible
attacks in a convenient way.

We first introduce our program model in \Cref{sec:preliminaries}. Afterwards,
\Cref{sec:simplification} shows how to transform arbitrary
tail-recursive integer programs into so-called \emph{simplified programs}
without loops.  To this end, in \Cref{sec:metering} we introduce
a variation of classical ranking
functions which we call \emph{metering functions}. These metering functions are used
to under-estimate the number
of iterations of simple loops, i.e.,  loops consisting of a single transition
without nested loops or branching.  Based on this concept, we present a framework to simplify programs
iteratively in Sections \ref{subsec:its-contraction} and \ref{subsec:its-simplification}.
It transforms tail-recursive programs
(with branching and sequences of possibly nested loops) into programs without
loops.  To this end, it \emph{accelerates} and then
eliminates simple loops by \mbox{(under-)}approximating
their effect using a combination of metering functions and recurrence solving.
In
\Cref{sec:non-linear}
we extend our technique to
an automatic approach which also transforms non-tail-recursive integer programs into simplified
programs.

\Cref{sec:asymptotic}
introduces techniques which allow us to infer \emph{asymptotic} lower
bounds from simplified programs.
In Sections \ref{subsec:its-asymptotic-bounds-and-limit-problems} and
\ref{subsec:its-transforming-limit-problems} we present a calculus to compute asymptotic
bounds by
repeatedly simplifying a \emph{limit problem}, which is an abstraction of the path
condition $\phi$ of a simplified program. This abstraction allows us to focus
on $\phi$'s behavior in the limit, i.e.,
we search for an infinite family of
  inputs that satisfy $\varphi$.
  In addition, \Cref{sec:asymptotic-smt} shows how a limit problem can
be encoded into a quantifier-free first-order formula with integer arithmetic.
Then off-the-shelf SMT solvers can be used to find a model for this formula.
This model in turn gives rise to
the desired family of inputs that satisfy $\phi$.
Thus, in many cases we can benefit from the power of SMT
solvers  instead of applying the rules of our calculus from \Cref{subsec:its-transforming-limit-problems}
heuristically.  Note that the calculus from \Cref{subsec:its-transforming-limit-problems} can simplify
limit problems such that our SMT encoding from
 \Cref{sec:asymptotic-smt}
becomes applicable and on the other hand, our SMT
encoding can be integrated into the calculus from
\Cref{subsec:its-transforming-limit-problems}, i.e., both techniques complement each
other.

Finally, we evaluate our implementation in the tool \loat in \Cref{sec:experiments}, discuss
related work in \Cref{sec:related}, and conclude in \Cref{sec:conclusion}.

A preliminary version of this
paper was published in \cite{ijcar16}. The current paper extends
\cite{ijcar16} significantly with the following novel contributions:
\begin{enumerate}
\item The new \Cref{lem:irrelevant-constraints}
integrates an optimization that we
proposed in \cite{ijcar16} into our
notion of metering functions. The novel insight is that in this way we can
 infer more
  expressive ``conditional'' metering functions of the form $\charfun{\psi} \cdot \bound$ where
$\bound$ is an ordinary arithmetic expression and
  $\charfun{\psi}$ is
  the characteristic function of the arithmetic condition $\psi$ (i.e., $\charfun{\psi}$ yields $1$ if
  $\psi$ is satisfied and $0$, otherwise).
Such metering functions are also useful to treat terminating and non-terminating rules in
a uniform way in our approach.
To ease the use of such metering functions,
in \Cref{thm:its-acceleration}
we extend the technique
for accelerating loops
from \cite[Theorem 7]{ijcar16} accordingly.
  Thus,
 conditional metering functions are now seamlessly integrated into our
 framework, resulting in a more streamlined formalization and presentation than
 in \cite{ijcar16}.
\item In \Cref{thm:its-instantiation}, we present a
  technique to eliminate variables from the
  program. In order to apply it automatically, the new \Cref{lem:zz-polys} clarifies how
  to check a crucial side condition which  requires that certain arithmetic expressions
  evaluate to integers whenever one
  instantiates their variables with integers. A similar side condition is also needed for the automation
  of our calculus for limit problems
  in \Cref{sec:asymptotic}. Moreover,
this
check could also be used to ensure that the initial program is ``well formed'' before
starting the analysis.
  In \cite{ijcar16}, the automation of
  this check has not been discussed.
 \item We lift our approach to non-tail-recursive programs in the new
  \Cref{sec:non-linear}.
  In contrast, \cite{ijcar16} was restricted to
  tail-recursive integer programs. While \cite{ijcar16} used a
  graph-based program model,
 in the current paper we propose a different
  (rule-based) representation of integer programs. This representation allows an easy
  formulation of non-tail-recursion, by using rules whose right-hand
  sides are multisets of terms.
\item The SMT encoding from \Cref{sec:asymptotic-smt} is completely new and improves the
  performance of our approach considerably.
  In particular,
  in our experiments of \Cref{sec:experiments}
  it outperformed
 the calculus from \Cref{subsec:its-transforming-limit-problems}
on examples with polynomial limit problems.
\item We provide formal proofs for all lemmas and theorems, which were missing
  in \cite{ijcar16}.
  \item Throughout the paper, we added many more examples, discussions, and explanations.
\end{enumerate}


\section{Program Model}
\label{sec:preliminaries}

We consider sequential imperative integer programs, allowing non-linear
arithmetic and non-determinism, whereas heap usage and concurrency are not supported.
In \cite{koat} we used an equivalent
program model, and showed how to deduce \emph{upper}
runtime bounds for integer programs.

Most existing abstractions that transform heap programs to integer programs are
``over-app\-rox\-i\-ma\-tions''.
However, in order to apply our approach to heap programs, we would need an under-approximating abstraction to ensure
that the inference of worst-case \emph{lower} bounds is sound.
As in most related work, we treat numbers as mathematical integers $\ZZ$.
However, one can use suitable transformations \cite{kittel-bitvectors,JLAMP18}  to handle
machine integers correctly, e.g.,  by inserting explicit normalization steps at possible
overflows.

In our program model, we use a
rule-based representation of integer
programs
where \emph{integer program rules} are of
the form $f(\vect{x}) \tox{\cc}{}{} \TTT \constr{\phi}$.  The \emph{left-hand
  side} $f(\vect{x})$ consists of a \emph{function symbol} $f$ and a vector of
pairwise different variables $\vect{x}$. The sets of all function symbols and
all variables are $\Sigma$ and $\VV$, respectively.
While $\Sigma$ is finite, we assume $\VV$ to be countably infinite, as we rely on fresh
variables to model non-determinism.  The \emph{arithmetic
  expression} $\cc$ represents the \emph{cost}
  of the rule, where arithmetic expressions
are composed of variables from $\VV$, numbers, and pre-defined operations like $+$, $-$,
$*$, etc.\
Annotating rules with costs enables a modular analysis, as it allows us to
summarize a sub-program $\PP$ into a single rule whose cost is a lower bound on
$\PP$'s complexity.
To ease readability, we sometimes omit the costs of rules.
The \emph{guard}
$\phi$ is a \emph{constraint} over $\VV$, i.e., a finite conjunction\footnote{Note that negations can be expressed by negating
  inequations directly, and disjunctions in programs can be expressed using
  several rules. We write ``$s = t$'' as syntactic sugar for ``$s \geq t \land s \leq t$''.}
of
inequations (built with ${<}$, ${\leq}$, ${>}$, or ${\geq}$)
over arithmetic expressions, which we omit if it is empty (i.e.,
we write $f(\vect{x}) \tox{\cc}{}{} \TTT$ instead of $f(\vect{x}) \tox{\cc}{}{}
\TTT \constr{\true}$).
So $c$ is an expression like $x \cdot y + 2^y$ and $\varphi$ is a formula like  $x \cdot y
\leq 2^y \land y > 0$, for example.
The \emph{right-hand side} $\TTT$ is a multiset of \emph{terms}
of the form $g(\vect{t})$ where $g \in \Sigma$ and $\vect{t}$ is a vector of
arithmetic expressions.
In the following, the notion of ``term'' always refers to
terms of this specific form. The set of all terms is $\TT$. We use $\VV(\cdot)$ to denote all variables occurring in the argument expression
  (e.g., $\VV(t)$ consists of all variables occurring in the term $t$).

Note that as in \cite{koat},
we do not allow nested calls of function symbols from $\Sigma$ in right-hand sides
of rules. For that reason,
  our program model also does not support
  return values. Instead of a rule $f(\vect{x}) \tox{\cc}{}{} f(g(\vect{t}))
  \constr{\varphi}$ with $f,g\in \Sigma$, one has
  to represent the result of the inner call $g(\vect{t})$ by a fresh \emph{temporary
    variable} $\tv$. So one uses a rule
  $f(\vect{x}) \tox{\cc}{}{} \{g(\vect{t}),f(\tv)\} \constr{\varphi \land \psi}$ instead,
where $\psi$ may restrict the possible values of $\tv$ by suitable inequations.

 We say that a function symbol $f$ has an \emph{incoming}
rule $\alpha$ if $f$ occurs in $\alpha$'s right-hand side and $f$ has an \emph{outgoing} rule
$\alpha$ if $f$ occurs on (the root position of)  $\alpha$'s left-hand side.
Given a rule $\alpha$,  $\head(\alpha)$ denotes the root symbol of $\alpha$'s left-hand
side $\lhs(\alpha)$. Furthermore, $\cost(\alpha)$, $\rhs(\alpha)$, and
$\guard(\alpha)$ denote the cost, the right-hand side, and the guard of
$\alpha$. The number of elements of the multiset in $\alpha$'s right-hand side is called the
\emph{degree} of $\alpha$.
A rule is \emph{tail-recursive} if its degree is at
most $1$.\footnote{\label{sink footnote}Note that rules with
    right-hand side $\emptyset$ can equivalently be transformed to rules with right-hand side
   ``$\fs{sink}$'' where $\fs{sink}$ is a fresh function symbol of arity $0$. Thus w.l.o.g., for tail-recursive
    programs we assume that all rules have degree $1$.}
For a rule $\alpha$ with degree 1, let $\dest(\alpha)$ be the root symbol of the term in
$\rhs(\alpha)$.

An \emph{integer program} is a finite set of integer program rules. It is
tail-recursive if all of its rules are tail-recursive.

\begin{example}[Fibonacci]
  \label{ex:its-fib}
  Consider the following imperative
  program, which computes the \lstinline{x}-th Fibonacci number and returns $1$ if
  $\lstinline{x}$ is negative.
\begin{lstlisting}
int fib(int x) {
  if (x <= 1) return 1;
  else return fib(x - 1) + fib(x - 2);
}
\end{lstlisting}
A suitable abstraction of this program
  would yield  the following integer program which represents its recursion pattern.
   For
simplicity, we sometimes write $t$ instead of $\{ t \}$ for singleton
(multi)sets of terms, i.e., in the first rule  we write
$\Ff_0(x)  \tox{0}{}{}  \fs{fib}(x)$ instead of $\Ff_0(x)  \tox{0}{}{}  \{\fs{fib}(x)\}$.
Here, $\Ff_0 \in \Sigma$ is the \emph{canonical start
symbol}, i.e., the entry point of the program.
\begin{alignat}{2}
       \Ff_0(x) & \tox{0}{}{}  \fs{fib}(x) && \label{eq:fib1}\\
      \fs{fib}(x) &\tox{1}{}{} \{ \fs{fib}(x-1), \fs{fib}(x-2) \} && \constr{x > 1} \label{eq:fib2}\\
      \fs{fib}(x) &\tox{1}{}{} \emptyset && \constr{x \leq 1}\label{eq:fib3}
     \end{alignat}
As the right-hand side of
Rule \eqref{eq:fib2} consists of two terms, it is not
  tail-recursive.  Note that the result of the fib program is not represented in the
  abstraction.

Using multisets as right-hand sides allows us to express that a function $f$
calls several other functions $f_1,\ldots,f_n$. Note that a rule of the
form $f(\ldots) \to \{f_1(\ldots),\ldots,f_n(\ldots)\} \constr{\phi}$ is
not equivalent to the $n$ rules $f(\ldots) \to f_i(\ldots) \constr{\phi}$ with $1 \leq i
\leq n$: While
the former rule expresses that $f$ invokes \emph{all} functions
$f_1,\ldots,f_n$, the latter rules mean that $f$ non-deterministically invokes
\emph{some} function $f_i$.
Thus, the recursive $\fs{fib}$-rule \eqref{eq:fib2}
cannot be replaced by the rules
  \begin{alignat}{2}
      \fs{fib}(x) &\tox{1}{}{} \fs{fib}(x-1) && \constr{x > 1} \quad \text{and} \label{eq:fib1-wrong}\\
      \fs{fib}(x) &\tox{1}{}{} \fs{fib}(x-2) && \constr{x > 1}. \label{eq:fib2-wrong}
   \end{alignat}
  These rules would mean that $\fs{fib}(x)$ \emph{either} evaluates to
  $\fs{fib}(x-1)$ \emph{or} to $\fs{fib}(x-2)$. In contrast, the recursive rule
  \eqref{eq:fib2} expresses that $\fs{fib}(x)$ evaluates to \emph{both}
  $\fs{fib}(x-1)$ and $\fs{fib}(x-2)$. Thus, the integer program $\{\eqref{eq:fib1}, \eqref{eq:fib2},
  \eqref{eq:fib3} \}$ has
  exponential complexity, but replacing its recursive rule \eqref{eq:fib2} with
  \eqref{eq:fib1-wrong} and \eqref{eq:fib2-wrong}  would result in a program with only linear
  complexity.
  So the recursion pattern of
  a non-tail-recursive program like fib cannot be modeled with rules
  whose right-hand sides are singleton sets.
\end{example}

Note that our notion of tail-recursion is a special case of the standard notion, where a procedure is considered to be
tail-recursive if recursive calls are only performed as its last action. The reason is that in our program model the
right-hand side of a rule is considered as a multiset and thus, there is no order imposed on the evaluation of its
elements. Thus, the non-tail-recursive rule
\[
  \Ff(\ldots) \tox{}{}{} \{\Ff(\ldots),\Fg(\ldots)\} \; \constr{\phi}
\]
could correspond to either of the following procedures, where the first one is tail-recursive, but the second one is not.

\begin{center}
  \begin{minipage}{0.15\textwidth}
\begin{lstlisting}
f(...) {
  g(...);
  f(...);
}
\end{lstlisting}
  \end{minipage}
  \hspace{10em}
  \begin{minipage}{0.15\textwidth}
\begin{lstlisting}
f(...) {
  f(...);
  g(...);
}
\end{lstlisting} 
  \end{minipage}
\end{center}
On the other hand, programs which are tail-recursive w.r.t.\ our notion of
tail-recursion are clearly also
tail-recursive in the usual sense.

The guards of rules restrict the control flow of the program, i.e., a rule
$f(\vect{x}) \tox{\cc}{}{} \TTT \constr{\phi}$ can only be applied if the
current valuation of the variables is a model of $\phi$. As we are concerned
with integer programs, all variables range over $\ZZ$, such that we only
consider models of $\phi$ which map variables to integers.

\begin{definition}[Substitutions]
  A substitution $\sigma$ instantiates variables by arithmetic expressions.
 We sometimes denote
  substitutions $\sigma$ by finite sets of key-value pairs
  $\{y_1/t_1, \ldots, y_k/t_k \}$ or $\{\vect{y} / \vect{t}\}$
for short, where $\vect{y}$ is the vector  $(y_1,\ldots,y_k)$ and $\vect{t}$ is
the vector $(t_1,\ldots,t_k)$. This substitution instantiates every variable $y_i \in \VV$
by the arithmetic expression $t_i$. Furthermore, it maps each\footnote{Slightly abusing notation,
    we sometimes use vectors as sets. So for a vector $\vect{y} = (y_1,\ldots,y_k)$,
    $\VV\setminus \vect{y}$ denotes $\VV \setminus \{ y_1,\ldots,y_k \}$.} $x \in \VV \setminus \vect{y}$  to $x$, i.e.,
the \emph{domain} of $\sigma$ is $\dom(\sigma) = \{ y_i \mid 1 \leq i \leq k, y_i \neq t_i\}$ and its
\emph{range} is $\{ x \sigma \mid x \in \dom(\sigma) \}$. Substitutions are
homomorphically extended to terms (i.e., $\sigma(t)$ instantiates all variables $x$ in the
term $t$ by $\sigma(x)$) and we usually write $t\sigma$ instead of $\sigma(t)$. For two
substitutions $\theta$ and $\sigma$, their \emph{composition} is denoted by $\theta \circ
\sigma$ where $t \, (\theta \circ \sigma) = t\theta\sigma$ (i.e., $\theta$ is applied first).

An \emph{integer substitution} is a substitution
that maps every
variable $x \in \dom(\sigma)$ to an integer number.
 We write $\sigma \models \phi$ for an integer substitution $\sigma$ if
  $\VV(\phi) \subseteq \dom(\sigma)$ and $\sigma$
  is a model of $\phi$.
\end{definition}

We always assume that $\Sigma$ contains the canonical start
symbol $\Ff_0$ and
we are only
interested in program runs that start with terms of the form $\Ff_0(\vect{n})$
where $\vect{n} \subset \ZZ$.  Note that this is not a restriction, as we can
simulate several start symbols $\Ff_1,\ldots,\Ff_k$ by adding corresponding
rules from $\Ff_0$ to $\Ff_1,\ldots,\Ff_k$.  W.l.o.g., we assume that $\Ff_0$
does not occur on right-hand sides of rules.  Otherwise, one could rename $\Ff_0$ to
$\Ff'_0$ and add a rule $\Ff_0(\vect{x}) \tox{0}{}{} \Ff'_0(\vect{x})$.

\begin{figure}[t]
 \hspace*{-.2cm}\begin{subfigure}[b]{0.37\textwidth}
    \begin{lstlisting}[basicstyle=\small]
&$\Ff_0$:\hspace*{-.1cm}& y = 0;
&$\Ff_1$:\hspace*{-.1cm}& while (x > 0) {
&\phantom{$\Ff_0$:}\hspace*{-.1cm}&   y = y + x;
&\phantom{$\Ff_0$:}\hspace*{-.1cm}&   x = x - 1;
&\phantom{$\Ff_0$:}\hspace*{-.1cm}& }
&\phantom{$\Ff_0$:}\hspace*{-.1cm}& z = y;
&$\Ff_2$:\hspace*{-.1cm}& while (z > 0) {
&\phantom{$\Ff_0$:}\hspace*{-.1cm}&   u = z - 1;
&$\Ff_3$:\hspace*{-.1cm}&   while (u > 0) {
&\phantom{$\Ff_0$:}\hspace*{-.1cm}&     u = u - random$(0,\infty)$;
&\phantom{$\Ff_0$:}\hspace*{-.1cm}&   }
&\phantom{$\Ff_0$:}\hspace*{-.1cm}&   z = z - 1;
&\phantom{$\Ff_0$:}\hspace*{-.1cm}& }
    \end{lstlisting}
    \caption{Example Integer Program}
    \label{fig:its-leading-ex-ip}
  \end{subfigure}
  \begin{subfigure}[b]{0.57\textwidth}
     \[
      \begin{array}{l@{\;}l@{\;\;}c@{\;}l@{\,}l@{\hspace*{-.5cm}}}
        \alpha_0\!:&\Ff_0(x,y,z,u) & \tox{1}{}{}{} & \Ff_1(x,0,z,u)\\
        \alpha_1\!:&\Ff_1(x,y,z,u) & \tox{1}{}{}{} & \Ff_1(x-1,y+x,z,u) & \constr{x > 0}\\
        \alpha_2\!:&\Ff_1(x,y,z,u) & \tox{1}{}{}{} & \Ff_2(x,y,y,u) & \constr{x \leq 0}\\
        \alpha_3\!:&\Ff_2(x,y,z,u) & \tox{1}{}{}{} & \Ff_3(x,y,z,z-1) & \constr{z > 0}\\
        \alpha_4\!:&\Ff_3(x,y,z,u) & \tox{1}{}{}{} & \Ff_3(x,y,z,u - \tv) & \constr{u > 0 \land \tv > 0}\\
        \alpha_5\!:&\Ff_3(x,y,z,u) & \tox{1}{}{}{} & \Ff_2(x,y,z-1,u) & \constr{u \leq 0}\\
      \end{array}
      \]
      \vspace*{2.5cm}
    \caption{Example Integer Program -- Rule Representation}
    \label{fig:its-leading-ex}
  \end{subfigure}
  \caption{Different Representations of Integer Programs}
  \label{fig:leading_ex_initial}
\end{figure}

\Cref{fig:its-leading-ex}
  shows an example of a tail-recursive integer program,
  i.e., here every right-hand side just consists of a single term.
\Cref{fig:its-leading-ex} corresponds to
the pseudo-code in \Cref{fig:its-leading-ex-ip}, where \(\code{random}(x,y)\)
returns a random integer \(\tv\) with $x < \tv < y$ and  $\omega$ is the smallest
infinite ordinal, i.e., we have $-\omega < n < \omega$ for all
numbers $n \in \ZZ$.  The following definition
clarifies how to evaluate integer programs.

\begin{definition}[Integer Transition Relation]
  \label{def:its-relation}
  A \emph{configuration}
is a multiset of terms of the form $f(\vect{n})$ where $f \in \Sigma$ and $\vect{n} \subset \ZZ$. The set of all
  configurations is denoted by $\CC$.

  Let $\PP$ be an integer program.  For configurations $\SSS,\TTT \in \CC$ and
  $\concretecosts \in \RR$, \emph{$\SSS$ evaluates to $\TTT$ with cost
    $\concretecosts$} ($\SSS \tox{\concretecosts}{}{\PP} \TTT$) if there is an
  $s \in \SSS$, a rule $\alpha$ of the form $f(\vect{x}) \tox{\cc}{}{} \QQQ \constr{\phi} \in \PP$,
  and an integer substitution $\sigma$ with $\VV(\alpha) \subseteq \dom(\sigma)$ such that
  $f(\vect{x})\sigma = s$,\footnote{Throughout this paper, ``${=}$'' means equality modulo arithmetic and
we assume that ground arithmetic expressions and comparisons are
evaluated exhaustively, so we have, e.g.,
    $f(2) = f(3-1)$ and $3 - 1 \in \ZZ$.}
  $\TTT
  = (\SSS \setminus \{s\}) \cup \QQQ\sigma$,\footnote{As usual, $\SSS \setminus \{s \}$
    means that the number of occurrences of $s$ in the multiset $\SSS$ (if any) is reduced
    by 1.
    We lift substitutions to
    (multi)sets of terms in the obvious way, i.e., $\QQQ\sigma = \{ q\sigma \mid
    q \in \QQQ \}$.} $\sigma \models \phi$, and $\cc\sigma = \concretecosts$.

  For any integer program rule
  $\alpha$ and any integer substitution $\sigma$, let $\sigma \models \alpha$ denote that
  $\VV(\alpha) \subseteq \dom(\sigma)$ and $\sigma \models \guard(\alpha)$.

  We write $\SSS \tox{\concretecosts}{}{\alpha} \TTT$ instead of $\SSS
  \tox{\concretecosts}{}{\PP} \TTT$ if $\PP$ is the singleton set containing
  $\alpha$.
  Moreover,
\[ \SSS_0 \tox{\concretecosts}{m}{\PP} \SSS_m \;\; \text{ denotes that } \;\;
  \SSS_0 \tox{\concretecosts_1}{}{\PP} \ldots \tox{\concretecosts_m}{}{\PP}
  \SSS_m \;\; \text{ for } \;\; \concretecosts = \sum_{i=1}^m \concretecosts_i.\]
  If $m$ is
  irrelevant, we write $\SSS_0 \tox{\concretecosts}{*}{\PP} \SSS_m$ if $m \geq 0$ and
  $\SSS_0 \tox{\concretecosts}{+}{\PP} \SSS_m$ if $m > 0$. Finally, we sometimes
  omit the costs (of both rules and evaluations) if they are not important.

  We say that a program is \emph{simplified} if $\head(\alpha) = \Ff_0$ for all rules
  $\alpha$, i.e., if all left-hand sides of rules are constructed with the canonical start
  symbol $\Ff_0$ which does
  not occur on right-hand sides. So any run of a simplified program  that
  starts with a term of the form $\Ff_0(\vect{n})$ has at most length 1.
\end{definition}

\begin{example}[Evaluation of Integer Programs]
Using the rules from \Cref{fig:its-leading-ex}, we have, e.g.,
 \[
    \Ff_0(3,2,1,0) \tox{1}{}{\alpha_0} \Ff_1(3,0,1,0) \tox{1}{}{\alpha_1} \Ff_1(2,3,1,0) \tox{1}{}{\alpha_1} \Ff_1(1,5,1,0) \tox{1}{}{\alpha_1} \ldots
    \]
    
    For the rules $\PP_{\fs{fib}} = \{\eqref{eq:fib1}, \eqref{eq:fib2},
  \eqref{eq:fib3} \}$ from \Cref{ex:its-fib}, we have
 \[
    \Ff_0(3) \tox{0}{}{{\PP_{\fs{fib}}}} \fs{fib}(3)
    \tox{1}{}{{\PP_{\fs{fib}}}}
    \{ \fs{fib}(2), \fs{fib}(1) \}
   \tox{1}{}{{\PP_{\fs{fib}}}} \{ \fs{fib}(1), \fs{fib}(0), \fs{fib}(1) \}
 \tox{3}{3}{{\PP_{\fs{fib}}}} \emptyset,
     \]
where $\{ \fs{fib}(1), \fs{fib}(0), \fs{fib}(1) \}
\tox{3}{3}{{\PP_{\fs{fib}}}} \emptyset$ abbreviates the three steps
\[ \{ \fs{fib}(1), \fs{fib}(0), \fs{fib}(1) \} \tox{1}{}{{\PP_{\fs{fib}}}}
 \{ \fs{fib}(0), \fs{fib}(1) \} \tox{1}{}{{\PP_{\fs{fib}}}}
 \{  \fs{fib}(1) \} \tox{1}{}{{\PP_{\fs{fib}}}}
 \emptyset,\]
 whose combined cost is $1 + 1 + 1 = 3$.
 \end{example}
 
    According to our definition, integer programs may also contain rules
like $\Ff(x) \tox{}{}{} \Ff(\frac{x}{2})$.  While
evaluations cannot yield
non-integer values (e.g., we cannot evaluate
$\Ff(1)$ to $\Ff(\frac{1}{2})$, as $\Ff(\frac{1}{2})$ is not a
configuration),
our technique assumes that arithmetic expressions on right-hand sides of rules always map integers to
integers.
Hence, throughout this paper we
restrict ourselves to \emph{well-formed} integer programs.

\begin{definition}[Well-Formed Integer Program]
  An integer program rule $\alpha$ is \emph{well formed} if $t_i\sigma \in \ZZ$
  for each $f(t_1,\ldots,t_k) \in \rhs(\alpha)$, for each $1
  \leq i \leq k$, and for each integer substitution $\sigma$ with  $\sigma \models
 \alpha$.
  An integer program is well
  formed if each of its rules is well formed.
\end{definition}
Note that for a well-formed rule $\alpha$ we do not require $\cost(\alpha)\sigma \in \ZZ$
 for integer substitutions $\sigma$ with  $\sigma \models
 \alpha$.

To ensure that the analyzed program $\PP_0$ is initially well formed, we just
allow integer numbers, addition, subtraction, and multiplication in
$\PP_0$.\footnote{\label{footnote:non-integer numbers}One
  could also allow expressions with non-integer numbers like $\frac{1}{2} x^2 +
  \frac{1}{2}x$ in the initial program,
  as long as every arithmetic
  expression in the program evaluates to an integer when instantiating its variables by
  integers.
We will present a criterion to detect such expressions in \Cref{lem:zz-polys}.}
  Our approach
relies on several program transformations, i.e., the initial program $\PP_0$ is
transformed into other programs $\PP_1, \PP_2, \ldots$ which may contain further
operations like division and exponentiation. However, all our transformations preserve
well-formedness.

To simplify the presentation, throughout
this paper, we assume that all function symbols in $\Sigma$ have the
same arity.  Otherwise, one can construct a variant of $\PP$ where additional
unused arguments are added to each function symbol whose arity is not maximal.
Moreover, we assume that the left-hand sides of $\PP$ only differ in their root
symbols, i.e., the argument lists are equal (e.g., in \Cref{fig:its-leading-ex},
the variables on the left-hand sides are consistently named $x,y,z,u$).
Otherwise, one can rename variables accordingly without affecting the relation
${\tox{}{}{\PP}}$.  The variables $\vect{x}$ on the left-hand sides are called
\emph{program variables} and for any rule $\alpha$,
all other variables $\TV(\alpha) = \VV(\alpha) \setminus \vect{x}$ are called
\emph{temporary}. These temporary variables are used to model
non-deterministic program data. 
So in \Cref{fig:its-leading-ex}, we have $\vect{x}=
(x,y,z,u)$ and $\tv \in \TV(\alpha_4)$.

In  \Cref{fig:leading_ex_initial}, the loop at $\Ff_1$ computes a value for $y$
that is quadratic in the original value of $x$.
Thus, the loop at $\Ff_2$ is executed quadratically often
where in each iteration, the inner loop at $\Ff_3$ may also be repeated
quadratically often.  Thus, the program's (\emph{worst-case}) runtime is a polynomial of
degree 4 in $x$.
In contrast, the \emph{best-case} runtime of the program is only quadratic in the original
value of $x$, because then the inner loop at $\Ff_3$ would always set $u$ to a
non-positive value immediately.
The goal of our paper is to infer lower bounds for worst-case runtimes automatically.

To formalize the (worst-case) runtime complexity of an integer program, we define the
\emph{derivation height} of a configuration $\SSS$ to be the cost of the most
expensive evaluation starting with $\SSS$.
Here, for any non-empty set $M
\subseteq \RR \cup \{\omega \}$, $\sup M$ is the least upper bound of $M$.
In the following, let  $\RR_{\geq 0} = \{k \in \RR \mid k \geq 0\}$.

\begin{definition}[Derivation Height \cite{dh}]
  \label{def:dh}
  Let $\PP$ be an integer program. Its \emph{derivation height} function
  $\dht_\PP: \CC \to \RR_{\geq 0} \cup \{\omega\}$ is defined as $\dht_\PP(\SSS) =
  \sup\{\concretecosts \in \RR \mid  \SSS
  \tox{\concretecosts}{*}{\PP} \TTT \text{ for some } \TTT \in \CC \}$.
\end{definition}

Clearly, we always have $\dht_\PP(\SSS) \geq 0$, since
    $\tox{\concretecosts}{*}{\PP}$  also permits evaluations with 0 steps.
For the integer program $\PP$ in \Cref{fig:its-leading-ex},
we obtain \(\dht_\PP(\Ff_0(0,y,z,u)) = 2\) for all $y,z,u \in \ZZ$,
since then we can only apply the transitions \(\alpha_0\) and  \(\alpha_2\) once.
For all terms
$\Ff_0(x,y,z,u)$ with \(x > 1\),
$\alpha_0$ is executed once, then $\alpha_1$ is executed $x$ times. Afterwards, $y$ has
the value $\tfrac{(x+1) \cdot x}{2}$. Now $\alpha_2$ is executed once and sets $z$ to the
value $\tfrac{(x+1) \cdot x}{2}$. The outer loop at $\Ff_2$ is executed $z$ times, where
in each iteration, the inner loop at $\Ff_3$ is executed $z-1$ time (in the worst case)
and $z$ is decreased by 1 in $\alpha_5$. So overall, $\alpha_3$ and $\alpha_5$ are both
executed $z$ times and $\alpha_4$ is executed $(z-1) + (z-2) + \ldots + 1 = \tfrac{z \cdot
  (z-1)}{2}$ times. Hence, the worst-case runtime is $1 + x + 1 + z + \tfrac{z \cdot
  (z-1)}{2} + z$, where $z = \tfrac{(x+1) \cdot x}{2}$, i.e.,
$\dht_\PP(\Ff_0(x,y,z,u)) = \frac{1}{8}  x^4 + \frac{1}{4}  x^3 + \frac{7}{8}  x^2 + \frac{7}{4}
x + 2$.
Our method will detect that the derivation
height of $\Ff_0(x,y,z,u)$  is at least
$\frac{1}{8}  x^4 + \frac{1}{4}  x^3 + \frac{7}{8}  x^2 + \frac{7}{4}
x$.  From this
concrete lower bound, our approach will infer that the asymptotic runtime
complexity of the program is in $\Omega(n^4)$ where $n$ is the size of the input, i.e., $n
= |x| + |y| + |z| + |u|$.
So the \emph{size} of the input is measured by the sum of the absolute values of all program variables.

Note that the derivation height can be unbounded even if the integer program
terminates.

\begin{example}[Unbounded Costs without Non-Termination]\label{ex:Unbounded Costs without Non-Termination}
  Consider the integer program $\PP$ with the rule
  $\Ff_0 \tox{\tv}{}{} \Ff$
  where $\tv$ is a temporary variable of the rule.  By \Cref{def:its-relation}, \emph{every} integer substitution
  $\sigma$ with
$\tv \in \dom(\sigma)$ can be used to evaluate $\Ff_0$ to
  $\Ff$.  Thus, for any $n \in \NN$ we have $\Ff_0 \tox{n}{}{}{}
  \Ff$ using an integer substitution $\sigma_n$ with $\tv\,\sigma_n =
  n$.
  Therefore we obtain $\dht_\PP(\Ff_0) = \omega$ although the rule has no recursive call.

  Similarly, for the program $\PP'$ with the rules $\Ff_0(x) \tox{0}{}{}{} \Ff(\tv)$ and
  $\Ff(x) \tox{1}{}{}{} \Ff(x-1) \, [x > 0]$, we also have  $\dht_{\PP'}(\Ff_0(1)) = \omega$,
  although every evaluation of the program is finite.
\end{example}

While
$\dht_\PP$ is defined on configurations,
the complexity of a program is
often defined as a function on $\NN$, in particular when considering asymptotic complexity
bounds. To bridge this gap, we use the common definition of
complexity as a function of the size of the input.
So the
\emph{runtime
  complexity} function $\rc_\PP(n)$
maps a natural
number $n$ to the cost of the most expensive program run where the size of the
input is bounded by $n$.

\begin{definition}[Runtime Complexity]\label{def:runtime complexity}
  Let $\PP$ be an integer program and let $k$ be the arity of $\Ff_0$.  The
  \emph{runtime complexity} function $\rc_\PP: \NN \to \RR_{\geq 0} \cup \{\omega\}$ of
  $\PP$ is defined as
  \[
    \rc_\PP(n) = \sup\{\dht_\PP(\Ff_0(\vect{n})) \mid \vect{n} \in \ZZ^k, \size{\vect{n}} \leq n\},
  \]
  where for the vector $\vect{n} =(n_1,\ldots,n_k)$, we have $\size{\vect{n}} = \sum_{i=1}^k|n_i|$.
\end{definition}

For the program $\PP$ from \Cref{fig:its-leading-ex}, recall that the derivation height
$\dht_\PP(\Ff_0(x,y,z,u)) = \tfrac{1}{8}x^4 + \tfrac{1}{4}x^3 + \tfrac{7}{8}x^2 + \tfrac{7}{4}x +
2$ solely depends on the value of the first argument $x$ of $\Ff_0$. 
As $n$ in \Cref{def:runtime complexity}  is a bound on the sum of the absolute values of
all arguments (i.e., $|x| + |y| + |z| + |u| \leq n$), setting $x=n$ and $y=z=u=0$ maximizes $\dht_\PP(\Ff_0(x,y,z,u))$.
Hence, we have $\rc_\PP(n) = \dht_\PP(\Ff_0(n,0,0,0)) = \tfrac{1}{8}n^4 + \tfrac{1}{4}n^3 + \tfrac{7}{8}n^2 + \tfrac{7}{4}n + 2$.

Obviously, $\dht_\PP$ and $\rc_\PP$ are not computable in general.
Thus, our goal is to find a lower bound on the runtime complexity of a program $\PP$ automatically
which is as precise as possible (i.e., a lower bound which is, e.g., unbounded,
exponential, or a polynomial of a degree as high as possible). So for the program in
\Cref{fig:its-leading-ex}, we would like to
derive
$\rc_\PP(n) \in
\Omega(n^4)$,
i.e., that the runtime complexity is asymptotically bounded from below by $n^4$.
 As usual,  \(f(n) \in \Omega(g(n))\)
means that there is an $m  > 0$ and an $n_0 \in \mathbb{N}$ such that
\(f(n) \geq m \cdot g(n)\) holds for all \(n \geq n_0\).
In our example, we also have $\rc_\PP(n) \in
\OO(n^4)$, i.e., $n^4$ is both an asymptotic \emph{lower} and \emph{upper} bound on the
(worst-case) runtime complexity.

Note that according to \Cref{def:runtime complexity}, $\rc_\PP(n)$ takes all runs into account that start with 
$\Ff_0(\vect{n})$ where the size $\size{\vect{n}}$ is $n$ \emph{or smaller} (i.e., in the
definition of $\rc_\PP(n)$ we use ``$\size{\vect{n}} \leq n$'' instead of
``$\size{\vect{n}} = n$''). This corresponds to the notion of ``runtime complexity'' used
e.g., for complexity analysis of term rewriting \cite{tct-dp-1} or for complexity analysis of
integer transition systems in the \emph{International
  Termination and Complexity Competition} \cite{termcompTACAS}.
To see the difference between $\rc_\PP$ and the alternative definition
 \[
    \rc_\PP'(n) = \sup\{\dht_\PP(\Ff_0(\vect{n})) \mid \vect{n} \in \ZZ^k, \size{\vect{n}} = n\},
  \]
  consider a program $\PP$ with the rule $\Ff_0(x) \tox{x}{}{}{} \Ff(x) \constr{x \geq 0 \land x = 2 \cdot
    \tv}$. For  non-negative numbers $n$ we have $\rc_\PP'(n) = \dht_\PP(\Ff_0(n)) = n$
  and $\rc_\PP(n) = n$ if $n$ is even,
  but $\rc_\PP'(n) = \dht_\PP(\Ff_0(n)) = 0$
and $\rc_\PP(n) =  \dht_\PP(\Ff_0(n-1)) = n-1$
  if $n$ is odd. As long as one is only interested in (asymptotic) upper bounds, the 
 difference between $\rc_\PP$ and $\rc_\PP'$ is negligible, since  we have both
 $\rc_\PP(n) \in \OO(n)$ and $\rc_\PP'(n) \in \OO(n)$. (More precisely, for any program
 $\PP$ we have $\rc_\PP(n) \in \OO(g(n))$ iff $\rc_\PP'(n) \in \OO(g(n))$ if $g$ is weakly
 monotonically increasing for large enough $n$.) But for (asymptotic) lower bounds,
 $\rc_\PP$ and $\rc_\PP'$ differ. For our example
 program we have $\rc_\PP(n) \in \Omega(n)$, but $\rc_\PP'(n) \notin \Omega(n)$ (we only have
 $\rc_\PP'(n) \in \Omega(1)$).
 Recall that two of our main motivations for the inference of worst-case lower bounds are
 \begin{enumerate}[label=(\Alph*), ref=(\Alph*)]
 \item \label{it:tight} to deduce \emph{tight} bounds in combination with existing techniques for the inference of worst-case upper bounds and
 \item \label{it:attack} to find denial-of-service vulnerabilities (or, more generally, performance bugs),
 \end{enumerate}
 see \Cref{sec:intro}.
The example above shows that in order to achieve \ref{it:tight}, one should use our definition of $\rc_\PP$
instead of $\rc_\PP'$.\footnote{Nevertheless, almost all of our techniques would also work
  in order to infer a lower bound on $\rc_\PP'$ instead of $\rc_\PP$. The only problem is
  in \Cref{sec:asymptotic} where we search for an infinite family of inputs that satisfy the guard
 of the program. Here, it is not required that this family can represent
  inputs of size $n$ for \emph{all} large enough $n$. So in our example, the technique of
\Cref{sec:asymptotic}
 would infer 
 $\rc_\PP(2 \cdot n) \in \Omega(n)$ (and indeed, we also have
$\rc_\PP'(2 \cdot n) \in \Omega(n)$), but the family of inputs ``$2 \cdot n$'' for all $n
  \geq 0$ does not represent \emph{all} possible large enough numbers. For weakly
  monotonically increasing functions like $\rc_\PP$, we present a technique in \Cref{sec:asymptotic}
  (viz.\ \Cref{lem:its-to-irc}) to transform a lower bound on  $\rc_\PP(2 \cdot n)$ into a lower bound
  on $\rc_\PP(n)$, i.e.,  we show that $\rc_\PP(2 \cdot n) \in \Omega(n)$ implies
  $\rc_\PP(n) \in \Omega(n)$. But the technique of \Cref{lem:its-to-irc} is not applicable to $\rc_\PP'$,
  because $\rc_\PP'$ is not weakly monotonically increasing.}
Regarding \ref{it:attack}, note that $\rc_\PP(n) \in \Omega(g(n))$ means that for large
enough $n$ one can always find inputs whose size is not greater than $n$, which lead to a
runtime of at least $m \cdot g(n)$. Hence,
our notion $\rc_\PP(n)$ can indeed be used to find families of
program inputs that lead to runtimes of at least length $m \cdot g(n)$, \emph{for all
  large enough $n$}.
So if $g$ is unacceptably large (e.g., exponential or a high-degree polynomial),
then such a family of program inputs witnesses a performance bug (which might, e.g., be exploited
for denial-of-service attacks).


\section{Simplifying Tail-Recursive Integer Programs}
\label{sec:simplification}

We now show how to transform any tail-recursive integer program $\PP$ into a simplified program $\PP'$
such that the runtime complexity of $\PP'$ is smaller or equal to the runtime complexity of $\PP$. Thus, any
lower bound for $\rc_{\PP'}$ is also a lower bound for  $\rc_{\PP}$.
In \Cref{sec:non-linear} we will extend our transformation to non-tail-recursive integer programs,
before inferring asymptotic lower bounds for the runtime
complexity of simplified programs in \Cref{sec:asymptotic}.

We first show in \Cref{sec:metering} how to under-estimate the number of possible loop iterations for
\emph{simple loops} $\alpha$ of the
form $f(\vect{x}) \tox{}{}{} f(\vect{x})\mu \constr{\phi}$, where we define
$\update(\alpha)= \mu$ and require $\dom(\mu) \subseteq \vect{x}$.
So for instance, the rule
\[ \alpha_1\!: \; \Ff_1(x,y,z,u) \; \tox{}{}{}{} \; \Ff_1(x-1,y+x,z,u) \;\; \constr{x >
  0}\]
from \Cref{fig:its-leading-ex} is a simple loop where
$\update(\alpha_1)$ is the substitution
$\mu = \{x/x-1, \, y/y+x \}$.
Based on the under-estimation of possible iterations, \Cref{subsec:its-contraction}
presents our technique to accelerate simple loops.
We introduce a technique to transform more complex loops into
simple loops in
\Cref{subsec:its-simplification}.

\subsection{Under-Estimating the Number of Iterations}
\label{sec:metering}

For a simple loop
$\alpha$ of the form $f(\vect{x}) \tox{}{}{} f(\vect{x})\mu \constr{\phi}$,
our goal is to infer an
arithmetic expression \(\bound\) such that for all integer substitutions
\(\sigma\) with $\sigma \models \alpha$,
the rule $\alpha$ can be executed at least $b\sigma$ times, i.e.,
there is an integer substitution
\(\sigma'\) with $f(\vect{x})\sigma \tox{}{\lceil \bound\sigma \rceil}{\alpha}
f(\vect{x})\sigma'$.  Here, as usual, $\lceil x \rceil$ is the smallest integer
$n$ with $n \geq x$.

To find such estimates, we use an adaptation of ranking
functions~\cite{podryb,bradley05,rank,costa-rf}
which we call \emph{metering functions}. In the following, we say that a
quantifier-free formula $\varphi$ is \emph{valid} if
we have $\sigma \models \varphi$
for every integer substitution $\sigma$ with $\VV(\varphi) \subseteq \dom(\sigma)$.

\begin{definition}[Ranking Function]
  \label{def:ranking}
  An arithmetic expression $\bound$ is a \emph{ranking function} for
  a simple loop $\alpha$ with $\update(\alpha) = \mu$ and $\TV(\alpha) = \emptyset$ if the following
  conditions are valid:
  \begin{eqnarray}
    \label{eq:its-rank1}
    \guard(\alpha) &\implies& \bound > 0\\
    \label{eq:its-rank2}
    \guard(\alpha) &\implies& \bound\mu \leq \bound-1
  \end{eqnarray}
\end{definition}

So for example, \(x\) is a ranking function for the rule \(\alpha_1\) in
\Cref{fig:its-leading-ex}, since both
$x > 0 \implies x > 0$ and $x > 0 \implies x-1 \leq x-1$
are clearly valid.
If $b$ is a ranking function for a rule $\alpha$, then for any integer substitution
$\sigma$ with $\VV(\alpha) \subseteq \dom(\sigma)$,
\(\bound\sigma\) \emph{over-estimates} the number of possible iterations of the loop $\alpha$:
\eqref{eq:its-rank2} ensures that \(\bound\sigma\) decreases at least
by $1$ in each loop iteration (i.e., $b \mu \sigma \leq b\sigma - 1$ holds whenever
$\guard(\alpha)\sigma$ is $\true$), and \eqref{eq:its-rank1} requires that
\(\bound\sigma\) is positive whenever the loop can be executed.

Note that \Cref{def:ranking} would be incorrect for the case $\TV(\alpha) \neq \emptyset$.
For example, consider the rule $\alpha: \; \Ff(x) \to \Ff(x+1) \constr{x < \tv}$.
If we omitted the requirement
 $\TV(\alpha) = \emptyset$,
then $\tv - x$ would be a ranking function for $\alpha$ since $\guard(\alpha)$ implies both
$\tv - x > 0$ and $\tv - (x+1) \leq \tv - x - 1$.
However, there are non-terminating evaluations like
$\Ff(0) \tox{}{}{\alpha} \Ff(1) \tox{}{}{\alpha} \Ff(2) \tox{}{}{\alpha} \ldots$, since
$\tv$ can be instantiated differently in each evaluation step. Thus,
$\tv - x$ is not a correct over-estimation for the number of loop iterations.

To cover the case $\TV(\alpha) \neq \emptyset$, \Cref{def:ranking} would need to reflect that the values of
temporary variables may change non-deterministically in every iteration.
We chose the simple definition above as it nicely exposes the analogy to our following novel concept of \emph{metering functions}.
In contrast to ranking functions, metering functions are \emph{under-estimates} for the maximal
number of iterations of a simple loop.

\begin{definition}[Metering Function]
  \label{def:its-meter}
  We call an arithmetic expression $\bound$ a \emph{metering function} for a
  simple loop \(\alpha\) with $\update(\alpha) = \mu$ if the following
  conditions are valid:
  \begin{eqnarray}
    \label{eq:its-meter1}
    \neg\guard(\alpha) &\implies& \bound \leq 0\\
    \label{eq:its-meter2}
    \guard(\alpha) &\implies& \bound\mu \geq \bound - 1
  \end{eqnarray}
\end{definition}

Here, \eqref{eq:its-meter2} ensures that \(\bound\sigma\) decreases at most by
$1$ in each loop iteration, and \eqref{eq:its-meter1} requires that
\(\bound\sigma\) is non-positive if the loop cannot be executed.  Thus, the loop
can be executed \emph{at least} $\bound\sigma$ times (i.e., $\bound\sigma$ is an
under-estimate).

In contrast to our definition of ranking functions, \Cref{def:its-meter} also covers the case
$\TV(\alpha) \neq \emptyset$. As we will show in \Cref{thm:MeteringFunctions},
the reason is that a metering function $\bound$ for a simple loop $\alpha$ is a witness that
$\alpha$ can be applied at least $\bound$ times for \emph{fixed} values of
$\alpha$'s temporary variables. In particular, metering functions can also 
contain temporary variables to express that the number of loop iterations is
unbounded, see \Cref{ex:its-unbounded}.

As an example, for the loop \(\alpha_1\) in \Cref{fig:its-leading-ex}, \(x\) is also a
metering function.  Condition \eqref{eq:its-meter1} requires the validity of \(\neg (x > 0)
\implies x \leq 0\) and \eqref{eq:its-meter2} requires \(x > 0 \implies x-1 \geq
x-1\).  While \(x\) is a metering \emph{and} a ranking function, \(\frac{1}{2}
x\) is a metering, but not a ranking function for \(\alpha_1\).
Similarly, \(x^2\) is a ranking, but not a metering function for \(\alpha_1\).
\Cref{thm:MeteringFunctions} states
that if $\bound$ is a metering function for a simple loop $\alpha$, then
$\alpha$ can be executed
at least $\lceil \bound\sigma \rceil$ times when starting the evaluation with
$\lhs(\alpha)\sigma$.
Thus, if every rule has the constant cost 1, then
$\dht_{\{\alpha\}}(\lhs(\alpha)\sigma) \geq \bound\sigma$ holds for all integer substitutions
$\sigma$ with $\VV(\alpha) \cup \VV(\bound) \subseteq \dom(\sigma)$.
Recall that we have $\rhs(\alpha) = \lhs(\alpha) \mu$ for
$\mu = \update(\alpha)$, since $\alpha$ is a simple loop. Hence,  the
evaluation has the form
\[ \lhs(\alpha)\,\sigma \; \tox{}{}{\alpha} \;
\rhs(\alpha)\,\sigma \; = \;
\lhs(\alpha)\,\mu\,\sigma \; \tox{}{}{\alpha}\;
\rhs(\alpha)\,\mu\,\sigma \; = \;
\lhs(\alpha)\,\mu^2\,\sigma \; \tox{}{}{\alpha} \; \ldots
\]
Here, for any $k \in \NN$,
$\mu^k$ stands for $k$ applications of $\mu$. So for example,
$\mu^{3}$ stands for $\mu \circ \mu \circ \mu$ and $\mu^0$ is the identity substitution.

\begin{theorem}[Metering Functions Under-Estimate Simple Loops]
  \label{thm:MeteringFunctions}
  Let $\bound$ be a metering function for a well-formed simple loop $\alpha$ with
  $\mu = \update(\alpha)$.  Then
  for all integer substitutions
  $\sigma$ with  $\VV(\alpha) \subseteq \dom(\sigma)$ and $\sigma \models b \geq 0$,  there is the following evaluation of length
    $\lceil \bound\sigma
  \rceil$:
  \[
    \lhs(\alpha)\,\sigma \tox{}{}{\alpha} \lhs(\alpha)\,\mu\,\sigma \tox{}{}{\alpha}
    \lhs(\alpha)\,\mu^2\,\sigma \tox{}{}{\alpha} \ldots  \tox{}{}{\alpha}
    \lhs(\alpha)\,\mu^{\lceil \bound\sigma \rceil}\,\sigma
    \]
   where $\mu^k \circ \sigma \models \guard(\alpha)$ for all $0 \leq k < \bound\sigma$.
\end{theorem}
\begin{proof}
  For any integer substitution $\sigma$ with $\VV(\alpha) \subseteq \dom(\sigma)$, let $m_{\sigma} \in \NN \cup \{ \omega
  \}$ be the length of the longest evaluation of the form $\lhs(\alpha)\,\sigma
  \tox{}{m_{\sigma}}{\alpha} \lhs(\alpha)\,\mu^{m_\sigma}\,\sigma$ where
 $\mu^k \circ \sigma \models \guard(\alpha)$ for all $0 \leq k < m_\sigma$.  So the loop $\alpha$ can be executed
  $m_{\sigma}$ times when starting with
  $\lhs(\alpha)\,\sigma$. We prove that
 $m_\sigma \geq b\sigma$.

  If $m_{\sigma} = \omega$, then the claim is trivial.
  For  $m_{\sigma} \neq \omega$, we use induction on $m_{\sigma}$.  In the
  base case $m_{\sigma} = 0$,
  we have  \(\sigma \not\models \guard(\alpha)\). Thus,
$\eqref{eq:its-meter1}$ implies $\bound\sigma \leq 0 = m_{\sigma}$.

   For the induction step $m_\sigma \geq 1$, we must have
  \(\sigma \models \guard(\alpha)\) which implies:
  \begin{eqnarray}
    \label{eq:MeteringFunctions-2a}
    \bound\mu\sigma  &\geq& \bound\sigma - 1 \qquad \text{by \eqref{eq:its-meter2}}\\
    \label{eq:MeteringFunctions-2}
    \lhs(\alpha)\,\sigma &\tox{}{}{\alpha}& \lhs(\alpha)\,\mu\,\sigma
  \end{eqnarray}

Due to
  \eqref{eq:MeteringFunctions-2},
the longest evaluation  $\lhs(\alpha)\, \sigma
\tox{}{m_{\sigma}}{\alpha} \lhs(\alpha)\,\mu^{m_\sigma}\,\sigma$ has the form
\[
\lhs(\alpha)\,\sigma \tox{}{}{\alpha} \lhs(\alpha)\,\mu\,\sigma
\tox{}{m_{\mu \circ \sigma}}{\alpha} \lhs(\alpha)\,\mu^{m_\sigma}\,\sigma,
\]
i.e., $m_\sigma =
m_{\mu \circ \sigma} + 1$. Since $\VV(\alpha) \subseteq \dom(\mu \circ \sigma)$
and
$\mu \circ \sigma$ is an integer substitution
(as $\alpha$ is well formed), the induction hypothesis implies $m_{\mu \circ \sigma} \geq
\bound  \mu \sigma$. Hence, we have $m_\sigma =
m_{\mu \circ \sigma} + 1 \geq \bound  \mu  \sigma + 1 \geq \bound \sigma$ by
\eqref{eq:MeteringFunctions-2a}.
\end{proof}

Note that if one regards a single simple loop $f(\vect{x}) \tox{}{}{} f(\vect{x})\mu
\constr{\phi}$
without any other $f$-rules that may lead to non-determinism, then the only remaining
possible non-determinism is due to the temporary variables. So then the number of
iterations of the loop in the worst and in the best case only depends on the instantiation
of the temporary variables. Since the requirements \eqref{eq:its-meter1} and
\eqref{eq:its-meter2} for the metering function must hold for all instantiations of the
variables (i.e., also for all instantiations of the temporary variables), then a metering
function is also a lower bound on the number of iterations of the loop in the best
case. To exploit that we only need lower bounds on the worst-case runtime of the loop, in
\Cref{subsec:its-contraction}
we will present a technique which can instantiate temporary variables by
suitable values which (hopefully) lead to long runtimes.

Our implementation builds upon a well-known transformation based on Farkas'
Lemma~\cite{bradley05,podryb} to find \emph{linear} metering functions.  The
basic idea is to search for coefficients of a linear template polynomial
\(\bound\) such that \eqref{eq:its-meter1} and \eqref{eq:its-meter2} hold for
all possible instantiations of the variables \(\VV(\alpha)\).  In addition to
\eqref{eq:its-meter1} and \eqref{eq:its-meter2}, we also require
\eqref{eq:its-rank1} to avoid trivial solutions like \(\bound = 0\).  Here, the
coefficients of $\bound$ are existentially quantified, while the variables from
\(\VV(\alpha)\) are universally quantified.  As in \cite{bradley05,podryb},
eliminating the universal quantifiers using Farkas' Lemma allows us to use
standard SMT solvers to search for $\bound$'s coefficients.\footnote{Since Farkas' Lemma
  is only applicable for linear constraints, for loops with non-linear arithmetic, our implementation uses 
simplifications in order to linearize the constraints that have to be satisfied for
metering functions:
We may substitute a non-linear term by a fresh variable,
provided that the variables of the non-linear term do not appear elsewhere
(the reverse substitution is applied to the metering function afterwards),
and we may omit irrelevant non-linear constraints or updates.
 For example, if the update of a variable $x$ which does not occur in
the guard is non-linear, then we use a linear template polynomial $b$ without the
variable $x$ for the metering function. But if
such simplifications are not possible, then we fail when trying to infer metering
functions for loops with non-linear arithmetic.}

If \(\guard(\alpha)\)
contains constraints that are irrelevant for
\(\alpha\)'s termination (provided that $\guard(\alpha)$ is satisfiable),
then one can improve our approach. More precisely, if $\guard(\alpha) =  \phi \land \psi$,
then $\psi$ is irrelevant for termination of the loop $\alpha$ if it always holds after
executing the loop (given that it holds before the loop), i.e., if
  $\guard(\alpha)$ implies $\psi\mu$.  In this case, one can
infer
``conditional'' metering functions of the form $\charfun{\psi} \cdot b$. Here, $\charfun{\psi}$
is the \emph{characteristic function} of $\psi$, i.e., for any integer substitution
$\sigma$ with $\VV(\psi) \subseteq \dom(\sigma)$ we have
$\charfun{\psi}\,\sigma = 1$ if $\sigma \models \psi$ and $\charfun{\psi}\,\sigma =
0$ otherwise.
So for example, $\charfun{y+z = 1} \, \sigma = 1$ holds for an integer substitution
$\sigma$ iff $y\sigma + z\sigma = 1$.

\begin{theorem}[Inferring Conditional Metering Functions]
  \label{lem:irrelevant-constraints}
  Let
  $\alpha$ be a simple loop such that $\update(\alpha) = \mu$ and
  $\guard(\alpha) = \phi \land \psi$ where $\guard(\alpha) \implies \psi\mu$ is
  valid.  If the conditions
  \begin{eqnarray}
    \label{eq:irrelevant1}
    \neg\phi \land \psi &\implies& \bound \leq 0\\
    \label{eq:irrelevant2}
    \phi \land \psi &\implies& \bound\mu \geq \bound - 1
  \end{eqnarray}
are valid,  then $\charfun{\psi} \cdot \bound$ is a metering function for $\alpha$.
\end{theorem}
\begin{proof}
  If $\sigma \models \guard(\alpha)$, i.e., $\sigma \models \phi \land \psi$, then we
extend $\sigma$ arbitrarily to the variables in $\psi\mu$ that do not occur in $\phi$ or
$\psi$. Then we have $\sigma \models
  \psi\mu$ as $\guard(\alpha)$ implies $\psi\mu$.
  Thus, $\charfun{\psi}\, \mu \sigma = 1$ and $\charfun{\psi}\,  \sigma = 1$. So
 by \eqref{eq:irrelevant2} we obtain
  $(\charfun{\psi} \cdot \bound)\,\mu\sigma =
  \bound\mu\sigma \geq \bound\sigma -1 = (\charfun{\psi} \cdot \bound)\,\sigma - 1$, i.e.,
  then $\charfun{\psi} \cdot \bound$ satisfies \eqref{eq:its-meter2}.

Now we regard the case where $\sigma \models \neg\guard(\alpha)$. If $\sigma \models \neg\psi$,
 then $(\charfun{\psi} \cdot \bound)\,\sigma = 0$,
  i.e., then $\charfun{\psi} \cdot \bound$ satisfies \eqref{eq:its-meter1}.
  Otherwise, we have $\sigma \models \neg\phi$
  and $\sigma \models \psi$. Then
  \eqref{eq:irrelevant1} implies $(\charfun{\psi} \cdot \bound)\, \sigma = \bound\sigma   \leq
  0$, i.e., then $\charfun{\psi} \cdot \bound$ also satisfies \eqref{eq:its-meter1}.
\end{proof}

While our implementation of \Cref{def:its-meter} was restricted to the search for linear
metering functions $\bound$, with  \Cref{lem:irrelevant-constraints} our implementation
can now also be used to obtain conditional metering functions of the form  $\charfun{\psi}
\cdot \bound$ for  linear arithmetic expressions $\bound$.

Compared to \Cref{def:its-meter}, \Cref{lem:irrelevant-constraints} weakens the conditions for metering functions:
If
$\bound$ is a metering function according to \eqref{eq:its-meter1}
and \eqref{eq:its-meter2}, then we can also prove that $\charfun{\psi} \cdot
\bound$ is a metering function via \Cref{lem:irrelevant-constraints} (but the
converse does not hold). The reason is that
 in  \eqref{eq:its-meter1}
we require $b \leq 0$  whenever $\neg(\phi \land \psi)$ holds, whereas
 in \eqref{eq:irrelevant1} we only require  $b \leq 0$  whenever $\neg\phi
\land \psi$ holds.

\begin{example}[Conditional Metering Function]
  \label{ex:Conditional Metering Function}
To illustrate the use of conditional metering functions, consider the following rule $\alpha$:
\[
   \Ff(x,y,z) \tox{}{}{} \Ff(x-y-z,y,z) \constr{x > 0 \land y + z = 1}.
\]
Here, we can choose $\phi$ to be $x > 0$ and $\psi$ to be $y + z = 1$, since $x > 0 \land y + z = 1$ implies $(y + z = 1
)\, \mu$ for $\alpha$'s update $\mu = \{ x/x-y-z\}$, i.e., it implies
$y + z = 1$.
 Hence, to infer the metering
function $\charfun{y + z = 1} \cdot x$, according to \Cref{lem:irrelevant-constraints}
it suffices to show that $\neg (x > 0) \land y + z = 1 \implies x \leq 0$ and $x > 0 \land y + z = 1
\implies x-y-z \geq x - 1$ are valid. Using this metering function, our approach can show
that the rule can be applied at least linearly often. In contrast,
without \Cref{lem:irrelevant-constraints} our implementation would not be able to generate a
useful metering function for this example, since it would only search for
linear
arithmetic expressions
$\bound$ that satisfy \eqref{eq:its-meter1} and \eqref{eq:its-meter2}. However, $x$ is not a metering function, since Condition
\eqref{eq:its-meter1} would not be satisfied (i.e., $\neg (x > 0 \land y + z = 1) \implies
x \leq 0$ is not valid).
\end{example}

In \cite{ijcar16}, we already sketched a related optimization, which is however
weaker than
\Cref{lem:irrelevant-constraints}. There the idea was to omit $\psi$ completely
when searching for metering functions. So with this optimization, one would check
the implications $\neg (x > 0) \implies x \leq 0$ and $x > 0 \implies x-y-z \geq x - 1$
to prove that $\charfun{y + z = 1} \cdot x$
is a metering function for the loop of \Cref{ex:Conditional Metering Function}.
As the second implication is not valid, this approach is not sufficient to handle
\Cref{ex:Conditional Metering Function}. In contrast,
\Cref{lem:irrelevant-constraints} adds $\psi$ to the premise in
\eqref{eq:irrelevant1} and \eqref{eq:irrelevant2}, i.e., the approach of \Cref{lem:irrelevant-constraints} for
inferring conditional metering
functions is strictly more powerful than the optimization from
\cite{ijcar16}.

Conditional metering
    functions are also particularly useful to integrate the handling of non-terminating
    rules in our approach.

\begin{example}[Unbounded Loops]
  \label{ex:its-unbounded}
  Let \(\alpha\) be a simple loop whose update is $\mu$.  If \(\guard(\alpha)
  \implies \guard(\alpha)\mu\) is valid and hence the \emph{whole} guard is irrelevant for
  $\alpha$'s termination,
  then $\alpha$ does not terminate (provided that $\guard(\alpha)$ is satisfiable).  In such cases, we can choose $\psi = \varphi$ in
  \Cref{lem:irrelevant-constraints} and thus,
  \eqref{eq:irrelevant1} and \eqref{eq:irrelevant2} from \Cref{lem:irrelevant-constraints} become
  \[
    \false \implies \bound \leq 0 \quad \text{ and } \quad
    \guard(\alpha) \implies \bound\mu \geq \bound - 1.
  \]
  This is valid for a fresh temporary variable $\bound = \tv$. Thus, for
  \[
    \PP = \{\Ff_0(x,y) \tox{0}{}{} \Ff(x,y), \, \alpha\} \quad \text{where $\alpha$ is the
      rule} \quad   \Ff(x,y) \tox[-1pt]{y}{}{} \Ff(x+1,y) \constr{0 < x},
  \]
  we obtain the metering function $\charfun{0 < x} \cdot \tv$.
  As $\guard(\alpha) = 0 < x$ is satisfiable, this indicates that the runtime of the program is unbounded, i.e.,
 $\dht_\PP(\Ff(x,y)\sigma) \geq \tv\,\sigma$ and thus  $\dht_\PP(\Ff(x,y)\sigma) = \omega$
  for all integer substitutions $\sigma$ with
$\{ x, y, \tv \} \subseteq \dom(\sigma)$ and $0 < x\sigma$.

  Note that in this example, \Cref{lem:irrelevant-constraints} succeeds when choosing
  $\bound$ to be $\tv$ (i.e.,   $\charfun{0 < x} \cdot \tv$ is a metering function). In contrast,
    $\tv$ is not a metering function, since \eqref{eq:its-meter1} does not hold (i.e.,
    $\neg (0 < x)$ does not imply $\tv \leq 0$). Thus, conditional metering
    functions allow us to handle terminating and non-terminating rules in a uniform way.
\end{example}


\subsection{Accelerating Simple Loops}
\label{subsec:its-contraction}

We now define \emph{sound processors} that simplify integer programs. A sound
processor is essentially a program transformation which preserves lower runtime
bounds.

\begin{definition}[Processor]
  \label{def:Processor}
  A \emph{processor} $\proc$ is a partial function which maps integer programs
  to integer programs. It is \emph{sound} if $\rc_\PP(n) \geq
  \rc_{\proc(\PP)}(n)$
  holds for all $n \in \NN$ and all
  $\PP$ where $\proc$ is defined.
\end{definition}

In our framework, processors are applied repeatedly until the extraction of a
concrete lower bound is straightforward.  We first
show how to \emph{accelerate} a simple loop \(\alpha\) to a rule which is
equivalent to applying \(\alpha\) multiple times (according to a metering
function for \(\alpha\)).  In
\Cref{subsec:its-simplification} we will show that
the resulting integer program can be simplified by
\emph{chaining} subsequent rules which may result in new simple loops.  Moreover, we describe a simplification strategy which
alternates these steps repeatedly.  In this way, we eventually obtain a
\emph{simplified} program without loops which directly gives rise to a concrete
lower bound.
\Cref{subsec:its-contraction} only deals with simple loops
and in \Cref{subsec:its-simplification}, we consider
arbitrary tail-re\-cur\-sive rules
  $\alpha$ of the form $f(\vect{x}) \tox{}{}{}{} g(\vect{t}) \constr{\phi}$
with
  $\update(\alpha) = \{\vect{x} / \vect{t}\}$ and
  $\dest(\alpha) = g$. We extend our approach to arbitrary rules in \Cref{sec:non-linear}.

First, consider a simple loop \(\alpha\) with $\update(\alpha) = \mu$ and
$\cost(\alpha) = \cc$.  To accelerate $\alpha$, we compute its \emph{iterated}
update and cost, i.e., a substitution $\mu_\itr$ that is a
closed form
of $\mu^\tv$ and an
arithmetic expression \(\cc_\itr\) that is an
under-approximation of $\sum_{i=0}^{\tv-1}\cc\mu^i$ for a fresh
temporary variable \(\tv\), where
 $\mu^i$ again denotes the $i$-fold composition $\mu \circ \ldots \circ \mu$ of the substitution
  $\mu$. So
$\mu_\itr = \mu^\tv$ and $\cc_\itr \leq \sum_{i=0}^{\tv-1}\cc\mu^i$ must hold for all
$\tv > 0$.
If $\charfun{\psi} \cdot \bound$ is a metering function for $\alpha$, then we add
the \emph{accelerated} rule
\[
  \lhs(\alpha) \tox{\cc_\itr}{}{} \lhs(\alpha)\,\mu_\itr \constr{\guard(\alpha) \land \psi \land 0 < \tv < \bound + 1}
\]
to the program.  It summarizes \(\tv\) iterations of $\alpha$, where \(\tv\) is
positive\footnote{\label{coverage}The accelerated rule does not cover the case that
$\alpha$ is not applied at all, i.e., it does not cover the case where $\tv = 0$.
 We excluded this case in order to ease the inference of the closed
  forms $\mu_\itr$ and $\cc_\itr$. To see this, consider a loop $\Ff(x) \to \Ff(0)$ with $x\mu = 0$. Here,  we
would get
$x\mu_\itr = 0$ if $\tv > 0$ and $x\mu_\itr  = x$ for $\tv = 0$. Hence, even in such simple
examples it would be difficult to express the iterated update in closed form when considering
the case $\tv = 0$ as well.}
and bounded by $\lceil \bound \rceil$.
Note that $\mu_\itr$ and $\cc_\itr$ may also contain operations which are not allowed
in the input program like division and exponentiation (i.e., in this way we can also infer
non-polynomial bounds).

For the program variables $\vect{x} =(x_1,\ldots,x_k)$, the iterated update $\mu_\itr$ is computed by solving the
recurrence equations \(x^{(1)} = x\mu\) and \(x^{(\tv+1)} =
x\mu\,\{x_1/x_1^{(\tv)}, \ldots, x_k/x_k^{(\tv)}\}\) for all \(x \in \vect{x}\).  So
for the rule \(\alpha_1\) from \Cref{fig:its-leading-ex} we get the
recurrence equations \(x^{(1)}=x-1\), \(x^{(\tv_1+1)}=x^{(\tv_1)}-1\),
\(y^{(1)}=y + x\), and \(y^{(\tv_1+1)}=y^{(\tv_1)} + x^{(\tv_1)}\).  Usually,
the resulting equations can easily be solved using state-of-the-art recurrence solvers, e.g.,
\cite{purrs,mathematica,maple}.
In our example, we obtain the closed forms
\[\begin{array}{l}
x\mu_\itr = x^{(\tv_1)} = x - \tv_1 \quad \text{and} \quad y\mu_\itr = y^{(\tv_1)} = y + \tv_1 \cdot x - \frac{1}{2}
\tv_1^2 + \frac{1}{2} \tv_1.
\end{array}
\]  While $y\mu_\itr$ contains rational
coefficients, our approach ensures that $\mu_\itr$ always maps integers to
integers.  Thus, our technique to accelerate loops preserves well-formedness.
We proceed similarly for the iterated cost of a rule, where we may
under-approximate the solution of the recurrence equations \(\cc^{(1)} = \cc\)
and \(\cc^{(\tv+1)} = \cc^{(\tv)} + \cc\,\{x_1/x_1^{(\tv)}, \ldots,
x_k/x_k^{(\tv)}\}\).  For the rule $\alpha_1$ in \Cref{fig:leading_ex_initial}, we get
\(\cc^{(1)} = 1\) and \(\cc^{(\tv_1+1)} = \cc^{(\tv_1)} + 1\) which leads to the
closed form $\cc_\itr = \cc^{(\tv_1)} = \tv_1$.
Hence, when using $\alpha_1$'s metering function $x$, it is accelerated to the following
rule:
\begin{equation}
\label{alpha1-accelerated}
  \Ff_1(x,y,z,u) \; \tox{\tv_1}{}{}{} \; \Ff_1(x-\tv_1,
   y + \tv_1 \cdot x - \tfrac{1}{2}
   \tv_1^2 + \tfrac{1}{2} \tv_1,z,u) \;\; \constr{x > 0 \land
   0 < \tv_1 < x+1}
\end{equation}
Here, the guard can be simplified to $0 < \tv_1 < x+1$. (We will perform such simplifications in all examples to ease readability.)

\begin{theorem}[Loop Acceleration]
  \label{thm:its-acceleration}
  Let $\PP$ be a well-formed integer program with the program variables $\vect{x}$, let \(\alpha \in \PP\) be a simple
  loop with $\update(\alpha) = \mu$ and $\cost(\alpha) = \cc$, let \(\tv\) be a
  fresh temporary variable, and let \(\charfun{\psi} \cdot \bound\) be a metering function for
  \(\alpha\), where $\bound$ is an arithmetic expression.  Moreover, for all $\tv > 0$,
  let \(x\mu_\itr = x\mu^{\tv}\) be valid
    for all \(x
  \in \vect{x}\),
    let $\cc_\itr \leq \sum_{i=0}^{\tv-1}\cc\mu^i$ be valid, and let
  \[
    \PP' = \PP \cup\{ \alpha_\itr \} \quad \text{where $\alpha_\itr$ is the rule} \quad  \lhs(\alpha)
    \tox{\cc_\itr}{}{} \lhs(\alpha)\,\mu_\itr \constr{\guard(\alpha) \land \psi \land 0 <
      \tv < \bound + 1}.
  \]
  Then
    $\PP'$ is well formed and the processor that maps \(\PP\) to \(\PP'\) is
  sound.
\end{theorem}
\begin{proof}
  Let $\sigma$ be an integer substitution  such that $\sigma \models
 \alpha_\itr$, i.e., we have
\begin{equation}
  \label{eq:accelerated-rule-application}
  \lhs(\alpha)\,\sigma
  \tox{\cc_\itr\sigma}{}{\alpha_\itr} \lhs(\alpha)\,\mu_\itr\,\sigma.
\end{equation}
  Note that $\sigma
  \models 0 < \tv < \bound + 1$ implies $\sigma \models 0 < \tv \leq \lceil \bound \rceil$
  and
  we have $\bound \sigma
 = (\charfun{\psi} \cdot \bound) \, \sigma$, because  $\sigma \models \guard(\alpha_\itr)$ implies
 $\sigma \models \psi$
and thus, $\charfun{\psi}\,\sigma = 1$.
  Since $\sigma \models \guard(\alpha_\itr)$ also implies
  $\sigma \models \charfun{\psi} \cdot \bound \geq 0$ and $\charfun{\psi} \cdot
  \bound$ is a metering function for $\alpha$, by \Cref{thm:MeteringFunctions}
  there is the following evaluation of length $\tv \, \sigma$:
  \[
  \lhs(\alpha)\,\sigma \tox{}{}{\alpha} \lhs(\alpha)\,\mu\,\sigma
    \tox{}{}{\alpha} \ldots \tox{}{}{\alpha}
    \lhs(\alpha)\,\mu^{\tv\, \sigma}\,\sigma
    \]
    where $\mu^k \circ \sigma \models \guard(\alpha)$ for all $0 \leq k < \tv \,
    \sigma$. For that reason, the costs of the rule applications are $\cc\sigma$,
    $\cc\mu\sigma$, \ldots, $\cc\mu^{\tv\,\sigma\, - \, 1}\sigma$, i.e.,
    \begin{equation}
    \label{eq:its-acceleration-1}
    \lhs(\alpha)\,\sigma \tox{\cc\sigma}{}{\alpha} \lhs(\alpha)\,\mu\,\sigma
    \tox{\cc\mu\sigma}{}{\alpha} \ldots \tox{\cc\mu^{\tv\,\sigma\, - \, 1}\sigma}{}{\alpha}
    \lhs(\alpha)\,\mu^{\tv\, \sigma}\,\sigma.
    \end{equation}

    For soundness, it suffices to show that every evaluation
    step \eqref{eq:accelerated-rule-application} with $\alpha_\itr$ can be simulated using a sequence of evaluation steps
    with $\alpha$ where the costs are at least the same.
    As $\alpha$ is well formed, this also implies well-formedness of $\alpha_\itr$.
    By definition of $\mu_\itr$ and
    $\cc_\itr$,  $x\mu_\itr = x\mu^\tv$ for all $x \in \vect{x}$ and $\cc_\itr \leq
    \sum_{i=0}^{\tv-1} \cc\mu^i$ are valid.  Thus, the evaluation \eqref{eq:its-acceleration-1}
indeed simulates the evaluation step \eqref{eq:accelerated-rule-application} with $\alpha_\itr$.
\end{proof}

Note that
\Cref{thm:its-acceleration} shows
that when using
conditional metering
functions of the form $\charfun{\psi} \cdot \bound$ (which we infer via
\Cref{lem:irrelevant-constraints}) to accelerate loops, the characteristic
function $\charfun{\psi}$ is not needed in the accelerated rule. Instead, the condition $\psi$ is simply added to its
guard. Thus, whenever the accelerated rule is
applicable, then $\charfun{\psi} \cdot \bound$ is equal to $\bound$ and hence,
$\lceil \bound \rceil$ under-estimates the number of consecutive
iterations of the original loop.
Clearly, \Cref{thm:its-acceleration} is also applicable if the metering function is not
conditional, i.e., if it
is an ordinary arithmetic expression (by choosing $\psi = \true$).

The following example illustrates that
the iterated update and cost may also contain non-poly\-no\-mial arithmetic, which may lead
to exponential bounds.

\begin{example}[Non-Polynomial Arithmetic due to Loop Acceleration]
  Consider  the program with the rule
$\Ff_0(x,y)  \tox{0}{}{} \Ff(x,y)$ and the
  simple loop $\Ff(x,y) \tox[-1pt]{y}{}{} \Ff(x-1,2y)
  \constr{x>0}$. Here, the update is $\mu = \{ x/x-1, \; y/2y \}$ and hence, the resulting iterated update is
  $\mu_\itr = \{ x/x-\tv, \; y/2^{\tv} \cdot y\}$. Moreover, the cost is $\cc = y$ and the iterated
  cost is
  $\cc_\itr = \sum_{i=0}^{\tv-1} 2^i y = (2^\tv -1) \cdot y$. Thus, both
  the iterated update and the iterated cost are exponential. Accelerating the simple loop
 via \Cref{thm:its-acceleration} with the
  metering function $x$
  yields $\Ff(x,y) \tox[-1pt]{(2^\tv - 1) \cdot y}{}{}
  \Ff(x-\tv,2^{\tv} \cdot y) \constr{0 < \tv < x+1}$, where we again simplified
  the guard $x>0 \land 0 < \tv < x+1$
to $0 < \tv < x+1$.
  Using this accelerated rule, our approach can infer an exponential lower bound for the
  program's runtime complexity.
\end{example}

Recall that the fresh variable $\tv$ represents the
number of loop iterations which are summarized by an accelerated rule. While $\tv$  ranges
over the integers, its upper bound $b + 1$ can be rational, as the following
example shows.

\begin{example}[Non-Integer Metering Functions]
  \label{ex:its-rational-metering}
  \Cref{thm:its-acceleration} also allows bounds that do not map to
  the integers.
  Consider the program
  \[
    \PP = \{\Ff_0(x) \tox{0}{}{} \Ff(x), \, \alpha\} \quad \text{where $\alpha$ is the rule} \quad \Ff(x) \tox{1}{}{} \Ff(x-2) \constr{0 < x}.
    \]
    Clearly, $\frac{1}{2} x$ is a metering function for $\alpha$, as $\neg (0 < x)
    \implies \frac{1}{2}  x \leq 0$ and
    $0 < x
    \implies \frac{1}{2} \, (x-2) \leq \frac{1}{2} x - 1$\linebreak are valid.
    For $\mu = \{ x / x-2 \}$ we have $\mu_\itr = \mu^\tv = \{ x / x - 2 \,  \tv \}$ and
    we choose $c_\itr = \sum_{i=0}^{\tv-1} 1 = \tv$.
Hence, accelerating \(\alpha\) with the metering function \(\frac{1}{2}  x\)
  yields
  \begin{equation}
    \label{eq:its-rational-metering}
   \begin{array}{rcll} \Ff(x) &\tox{\tv}{}{}& \Ff(x - 2 \, \tv) & \constr{0 < \tv <
       \frac{1}{2}  x + 1}.
     \end{array}
  \end{equation}
  Note that \(0 < \tv < \frac{1}{2} x+1\) implies \(0 < x\) as $\tv$ ranges over \(\ZZ\).
  Hence, \(0 < x\) can be omitted in the resulting
  guard.
\end{example}

If a (non-terminating) simple loop has the metering function $\charfun{\phi} \cdot \tv$
where $\tv$ is a fresh temporary variable,
then the upper bound $b + 1 = \tv + 1$ on the number of summarized loop
iterations can take arbitrary values.

\begin{example}[Unbounded Loops Continued]
  \label{ex:its-unbounded2}
  In \Cref{ex:its-unbounded}, $\charfun{0 < x} \cdot \tv$ is a metering function for
  \(\alpha: \, \Ff(x,y) \tox[-1pt]{y}{}{} \Ff(x+1,y) \constr{0 < x}\).
  The resulting accelerated rule $\alpha_\itr$ is
  \[
    \Ff(x, y) \tox[-1pt]{\tv_1 \cdot y}{}{} \Ff(x + \tv_1, y) \constr{0 < x \land 0 < \tv_1 < \tv + 1}.
  \]
  Since $\tv$ does not have any upper bound, the value of $\tv_1$ is not bounded
  by the values of the program variables $x$ and $y$. Thus, the condition of the rule
  could be replaced by $0 < x \land 0 < \tv_1$, i.e., we obtain
\setcounter{eq:unboundedAcceleratedCtr}{\value{equation}}
  \begin{equation}
    \label{eq:unboundedAccelerated}
 \Ff(x, y) \tox[-1pt]{\tv_1 \cdot y}{}{} \Ff(x + \tv_1, y) \constr{0 < x \land 0 < \tv_1}.
      \end{equation}
\end{example}

\bigskip

After accelerating a simple loop $\alpha$ according to \Cref{thm:its-acceleration}
using the metering function $\charfun{\psi} \cdot \bound$,
we
eliminate
  the fresh variable $\tv$ by instantiating it with $\bound$,  provided
that $\bound$ maps to $\ZZ$
(i.e., for every integer substitution $\sigma$ with
$\sigma \models \alpha\{\tv /
  \bound\}$
we have $\bound\sigma \in \ZZ$).
The reason is that we want to keep the number of variables small for the sake
of efficiency. However, this is just a heuristic which can also lead to worse results (e.g., if
there is a non-terminating run where the original non-accelerated loop must not be applied more than
$\bound - 1$ times after each other).
If we cannot verify that  $\bound\sigma \in \ZZ$ holds for every 
integer substitution $\sigma$ with
$\sigma \models \alpha\{\tv /
\bound\}$, then we do not eliminate $\tv$, but
keep the inequation $0 < \tv < b+1$ in the accelerated rule.

We apply the following processor for the
instantiation of temporary variables.

\begin{theorem}[Instantiation]
  \label{thm:its-instantiation}
  Let $\PP$ be a well-formed
  integer program, let \(\alpha \in \PP\),
  let \(\tv \in \TV(\alpha)\), let $\bound$ be an arithmetic expression such
  that for every integer substitution $\sigma$ with
$\sigma \models \alpha\{\tv /
  \bound\}$
we have $\bound\sigma \in \ZZ$, and let \(\PP' = \PP \cup \{\alpha\{\tv / \bound\}\}\).
  Then $\PP'$ is well formed and the processor mapping \(\PP\) to
  \(\PP'\) is sound.
\end{theorem}
\begin{proof}
  To
  show the soundness of the processor, let $\sigma$ be an integer substitution with $\sigma \models
\alpha\{\tv /
  \bound\}$ and $\bound\sigma = m \in \ZZ$. Then
 let $\sigma'$ be the integer substitution with
  $\dom(\sigma') = \dom(\sigma) \cup \{ \tv \}$, $\sigma'(\tv) = m$, and $\sigma'(x) = \sigma(x)$ for all
 $x \in \dom(\sigma) \setminus \{ \tv \}$. Clearly,
 $\sigma \models \guard(\alpha\{\tv / \bound\})$ iff
 $\sigma' \models \guard(\alpha)$
  and moreover,
 $\cost(\alpha\{\tv / \bound\})\sigma = \cost(\alpha)\sigma'$.
Thus,
\[ \lhs(\alpha\{\tv / \bound\})\sigma   \tox{k}{}{\alpha\{\tv / \bound\}} \rhs(\alpha\{\tv /
\bound\})\sigma
\quad \text{implies} \quad
\lhs(\alpha)\sigma'   \tox{k}{}{\alpha} \rhs(\alpha)\sigma'.\]
This shows that every step with $\alpha\{\tv / \bound\}$ can
also be done with $\alpha$, because we have
$\lhs(\alpha\{\tv / \bound\})\sigma = \lhs(\alpha)\{\tv / \bound\}\sigma = \lhs(\alpha)\{\tv / m\}\sigma = \lhs(\alpha)\sigma'$
(and $\rhs(\alpha\{\tv / \bound\})\sigma = \rhs(\alpha)\sigma'$ can be derived analogously).

  To show that $\PP'$ is well formed, recall that
$\sigma \models \guard(\alpha\{\tv / \bound\})$ iff
  $\sigma' \models \guard(\alpha)$. If $\rhs(\alpha)$ contains $f(t_1,\ldots,t_k)$, then
 $t_i\sigma' \in \ZZ$ holds for all $1 \leq i \leq k$
  by
well-formedness of $\PP$. As $t_i\sigma' = t_i
\{\tv/b\} \sigma$, this implies well-formedness of $\PP'$.
\end{proof}

\begin{example}[Instantiation of Fresh Temporary Variables]\label{ex:loop_elimination}
For
 our example from \Cref{fig:its-leading-ex}, accelerating
 $\alpha_1$ results in the rule \eqref{alpha1-accelerated}.
By instantiating its temporary variable $\tv_1$ with the metering function $x$,
the above processor yields
\[\alpha_{\overline{1}}\!: \;  \Ff_1(x,y,z,u) \; \tox{x}{}{}{} \; \Ff_1(0,
   y +  \tfrac{1}{2} x^2 + \tfrac{1}{2} x,z,u) \;\; \constr{x > 0}.\]
\end{example}

In  \cite[Theorem 10]{ijcar16}, we presented a processor
which can extend the guard of a rule by arbitrary conjuncts. This processor could
be used  as an
 alternative to \Cref{thm:its-instantiation}, because instead of instantiating
$\tv$, one could add the constraint ``$\tv = \bound$'' to the guard of
$\alpha$.  In practice, however, it is preferable to instantiate $\tv$ in order to
keep the number of variables as small as possible.

If we cannot apply \Cref{thm:its-acceleration}  for a simple loop $\alpha$, because our
implementation fails to solve the recurrence equations needed to compute
the
closed forms \(\mu_\itr\) or \(\cc_\itr\), or because it cannot
find
a useful metering function,  then we
can simplify \(\alpha\) by eliminating temporary variables.  To do so, we fix
their values via \Cref{thm:its-instantiation}.  As we are
interested in witnesses for maximal computations, we use a heuristic that sets
\(\tv\) to \(a\) for temporary variables \(\tv\) where the arithmetic
expression $a$ is a minimal upper or a maximal lower bound on \(\tv\)'s values,
i.e., \(\guard(\alpha)\) implies \(\tv \leq a\) but not $\tv \leq a - 1$, or
\(\guard(\alpha)\) implies \(\tv \geq a\) but not $\tv \geq a + 1$.  This
elimination of temporary variables is repeated until we find constraints which
allow us to apply loop acceleration.

\begin{example}[Instantiation of Other Temporary Variables]\label{ex:loop_elimination2}
  For the rule $\alpha_4$ from  \Cref{fig:its-leading-ex},\linebreak \(\guard(\alpha_4)\) contains the
  constraint $\tv > 0$.  So $\guard(\alpha_4)$ implies the bound \(\tv \geq 1\)
  since $\tv$ must be instantiated by an integer.  Hence, we instantiate the rule
  $\alpha_4$ by replacing $\tv$ with \(1\).  Thus, the update $\{u / u -
  \tv\}$ of the instantiated rule $\alpha'_4$ becomes $\{u / u - 1\}$.
  Hence, now $u$ is a metering function
  for $\alpha'_4$ (whereas it was not a metering function for $\alpha_4$, as $u$'s value could
  decrease by more than 1 in each application of $\alpha_4$). Thus, \(\alpha'_4\) can be accelerated
  similarly to \(\alpha_1\), resulting in the rule
  \[
  \Ff_3(x,y,z,u) \tox{\tv_4}{}{}{} \Ff_3(x,y,z,u - \tv_4) \constr{0 < \tv_4 < u + 1}.
  \]
  Now the temporary variable
 $\tv_4$ that results from loop acceleration can be eliminated by instantiating it with
  the metering function $u$. In this way, we obtain
\[
  \alpha_{\overline{4}}\!: \;\Ff_3(x,y,z,u) \; \tox{u}{}{}{} \; \Ff_3(x,y,z,0) \;\; \constr{u >
    0}.
  \]
\end{example}

If $\bound$ is a polynomial, then we can use the following
generalization of an observation from \cite{int-polys} to check the side condition of
\Cref{thm:its-instantiation} that $\bound$ needs to map to $\ZZ$.
Note that this check does not take
the guard of the rule into account. So for the rule
\[
  \Ff(x, y) \to \Ff(x - 2, y - 1) \constr{x > 0 \land x = 2 \cdot y}
\]
with the metering function $\frac{x}{2}$ it would fail to recognize that $\sigma\left(\frac{x}{2}\right)$ is an integer for every model $\sigma$ of $x > 0 \land x = 2 \cdot y$.

\begin{lemma}[Polynomials Mapping to $\ZZ$]
  \label{lem:zz-polys}
  Let
  $f: \ZZ^k \to \RR$, where $f(x_1,\ldots,x_k)$ is a polynomial over the variables $x_1,\ldots,x_k$ with degrees
  $d_1,\ldots,d_k$ w.r.t.\ $x_1,\ldots,x_k$, respectively (i.e.,
for each $1 \leq i \leq k$,
  $f(x_1,\ldots,x_k)$ can be rearranged to the form $\sum_{j=0}^{d_i} p_j  \cdot x_i^j$
where
each $p_j$ is a polynomial over the variables $x_1,\ldots,x_{i-1},x_{i+1},\ldots,x_k$).
  If there are numbers $n_1,\ldots,n_k \in \ZZ$ such that $f(m_1,\ldots, m_k) \in \ZZ$ for
  all $m_1,\ldots,m_k \in \ZZ$ with $n_i \leq m_i \leq n_i + d_i + 1$, then $\img(f)
  \subseteq \ZZ$, i.e., then we have $f(m_1,\ldots,m_k) \in \ZZ$ for all $m_1,\ldots,m_k \in \ZZ$.
\end{lemma}
\begin{proof}
  We use induction on $d = \sum_{i=1}^k d_i$.
  If $d = 0$, then $f$ is a constant and thus the claim is trivial.
  If $d> 0$, then there exists a $1 \leq j \leq k$ with $d_j > 0$. Then $g(x_1,\ldots,x_k)
  = f(x_1,\ldots,x_{j-1},x_{j}+1,x_{j+1},\ldots,x_k) - f(x_1,\ldots,x_k)$ is a polynomial
  whose degree w.r.t.\ $x_j$ is $d_j-1$ and whose degree w.r.t.\ all $x_i$ with $i \neq j$
  is at most $d_i$.
  To see this, note that all monomials that do not contain $x_j$ vanish in $g(x_1,\ldots,x_k)$.
  Thus, $g(x_1,\ldots,x_k)$ is a finite sum of expressions of the form $m \cdot (x_j +
  1)^{e} \cdot p - m \cdot x_j^{e} \cdot p$ where $m \in \RR$, $e \leq d_j$, and $p$ is a
  product of $x_1,\ldots,x_{j-1},x_{j+1},\ldots,x_k$ where each $x_i$, $i \neq j$, occurs
  at most $d_i$ times.
  We get:
  \[
    \begin{array}{lll}
      &m \cdot (x_j + 1)^{e} \cdot p - m \cdot x_j^{e} \cdot p\\
      =& m \cdot \left( \sum_{i=0}^{e}\binom{e}{i} \cdot x_j^{i}\right) \cdot p - m \cdot
      x_j^{e} \cdot p & \text{by the Binomial theorem} \\
      =& m \cdot \left( \binom{e}{e} \cdot x_j^e + q \right) \cdot p - m \cdot x_j^{e}
      \cdot p & \text{where } q \text{ is a univariate polynomial over } x_j\\
      &&\text{whose degree is smaller than } e\\
      =& m \cdot (x_j^e + q ) \cdot p - m \cdot x_j^{e} \cdot p & \text{as } \binom{e}{e} = 1\\
      =& m \cdot q \cdot p
    \end{array}
  \]

  Moreover,
  since $f(m_1,\ldots,m_j+1, \ldots, m_k) \in \ZZ$ and
  $f(m_1,\ldots,m_j, \ldots, m_k)\in \ZZ$ for all $m_1,\ldots,\linebreak m_k \in \ZZ$ with $n_j \leq
  m_j \leq n_j + d_j$ and $n_i \leq m_i \leq n_i + d_i + 1$ if $i \neq j$, we also have
$g(m_1,\ldots,m_k) \in \ZZ$ for these $m_1,\ldots,m_k$.
  Thus, by the induction hypothesis, we obtain $\img(g) \subseteq \ZZ$.
  This means that we have $f(x_1,\ldots,x_{j-1},x_{j}+1,x_{j+1},\ldots,x_k) -
  f(x_1,\ldots,x_k) \in \ZZ$ for all $x_1,\ldots,x_k \in \ZZ$.
Since this construction can be done for every $1 \leq j \leq k$ with $d_j > 0$, we have
   $f(x_1,\ldots,x_{j-1},x_{j}+1,x_{j+1},\ldots,x_k) -
f(x_1,\ldots,x_k) \in \ZZ$
\emph{for all $1 \leq j \leq k$} and  all $x_1,\ldots,x_k \in \ZZ$.
Thus, we also have   $f(x_1,\ldots,x_{j-1},x_{j}-1,x_{j+1},\ldots,x_k) -
f(x_1,\ldots,x_k) \in \ZZ$
for all $1 \leq j \leq k$ and  all $x_1,\ldots,x_k \in \ZZ$.
  Since we also have $f(n_1,\ldots,n_k) \in \ZZ$ for some numbers $n_1,\ldots,n_k \in \ZZ$, this proves $\img(f) \subseteq
  \ZZ$.
\end{proof}

So if $\bound$ is a polynomial, then it suffices to check if instantiating the variables in $\bound$
by finitely many integers always results in an integer. More precisely, if the
polynomial $\bound$ contains the variables $x_1,\ldots,x_k$ of degrees $d_1,\ldots,d_k$,
respectively, then
we only check if
the polynomial maps all arguments from $\{0,\ldots,d_1+1\} \times \ldots \times \{0,\ldots,d_k+1\}$
to integers.
So we choose $n_i = 0$ for each $n_i$ from
\Cref{lem:zz-polys}. This is not a restriction, because if there is some other $n_i'$
such that $f(m_1,\ldots, m_k) \in \ZZ$ for
  all $m_1,\ldots,m_k \in \ZZ$ with $n_i' \leq m_i \leq n_i' + d_i + 1$, then by 
  \Cref{lem:zz-polys} we have
$\img(f)
  \subseteq \ZZ$, which implies
   $f(m_1,\ldots, m_k) \in \ZZ$ for
  all $m_1,\ldots,m_k \in \ZZ$ with $n_i = 0  \leq m_i \leq  d_i + 1 = n_i + d_i + 1$.

For instance, to check that  the polynomial $\frac{1}{2} x^2 +
  \frac{1}{2}x$ maps \emph{all} $x \in \ZZ$ to integers, it suffices to check this for just $x
  \in \{0,1,2,3\}$. In contrast, for the polynomial $\frac{x}{2}$, we would check its value for $x \in
  \{0,1,2\}$ and determine that it does not always yield an integer.

Thus, one can implement \Cref{thm:its-instantiation} using
\Cref{lem:zz-polys} (but, of course, one may also incorporate further sufficient
criteria). Similarly, as mentioned before,
the criterion of \Cref{lem:zz-polys} can also be used to check well-formedness of the
integer program if we permit non-integer constants in the
initial program.

To simplify the program, we delete the original
rules after instantiating or accelerating them.
If acceleration of a rule $\alpha$ still fails after eliminating all temporary variables
by instantiating $\alpha$ repeatedly, then $\alpha$ is removed completely. So in the end,
we just keep simple loops that have been accelerated. The following theorem shows that
deleting rules is always sound.

\begin{theorem}[Deletion]
  \label{thm:its-deletion}
  Let $\PP$ be a well-formed integer program, let \(\alpha \in \PP\), and let \(\PP' =
  \PP \setminus \{\alpha\}\).  Then $\PP'$ is well formed and the processor
  mapping \(\PP\) to \(\PP'\) is sound.
\end{theorem}
\begin{proof}
  Since $\PP$ is well formed, $\PP'$ is trivially well formed, too.  The
  processor is sound since every evaluation with $\PP \setminus \{ \alpha \}$ is
  also an evaluation with $\PP$.
\end{proof}

\begin{example}[Ex.\ \ref{ex:loop_elimination} and \ref{ex:loop_elimination2} Continued]\label{ex:deletion}
  After accelerating and instantiating all simple loops in the program
  from  \Cref{fig:its-leading-ex},
  we delete the original loops $\alpha_1$ and $\alpha_4$, resulting in the following integer program:
 \[
  \begin{array}{l@{\;\;}l@{\;\;\;}c@{\;\;}ll}
  \alpha_0\!:&\Ff_0(x,y,z,u) & \tox{1}{}{}{} & \Ff_1(x,0,z,u)\\
  \alpha_{\overline{1}}\!:&\Ff_1(x,y,z,u) & \tox{x}{}{}{} & \Ff_1(0,
  y + \frac{1}{2}x^2 + \frac{1}{2}  x,z,u) & \constr{x > 0}\\
  \alpha_2\!:&\Ff_1(x,y,z,u) & \tox{1}{}{}{} & \Ff_2(x,y,y,u) & \constr{x \leq 0}\\
  \alpha_3\!:&\Ff_2(x,y,z,u) & \tox{1}{}{}{} & \Ff_3(x,y,z,z-1) & \constr{z > 0}\\
  \alpha_{\overline{4}}\!:&\Ff_3(x,y,z,u) & \tox{u}{}{}{} & \Ff_3(x,y,z,0) & \constr{u > 0}\\
  \alpha_5\!:&\Ff_3(x,y,z,u) & \tox{1}{}{}{} & \Ff_2(x,y,z-1,u) & \constr{u \leq 0}
  \end{array}
  \]
\end{example}

\subsection{Chaining Rules}
\label{subsec:its-simplification}

After trying to accelerate all simple loops of a program, we can \emph{chain}
subsequent rules $\alpha_1, \alpha_2$ by adding a new rule $\alpha_{1.2}$ that
represents their combination.
Our notion of \emph{chaining} corresponds to the standard notion of \emph{unfolding}
\cite{BD77}, adapted to our program model.
Afterwards, the rules \(\alpha_1\) and
\(\alpha_2\) can (but need not) be deleted with \Cref{thm:its-deletion}.

\begin{theorem}[Chaining  for Tail-Recursive Integer Programs]
  \label{thm:its-chaining}
  Let $\PP$ be a well-formed tail-recursive integer program and let $\alpha_1,\alpha_2 \in \PP$
  where
  \[
    \begin{array}{rccll}
      \alpha_1: & f_1(\vect{x}) &\tox{\cc_1}{}{}& f_2(\vect{x})\,\mu & \constr{\phi_1} \text{ and}\\
      \alpha_2: & f_2(\vect{x}) &\tox{\cc_2}{}{}& t & \constr{\phi_2}.
    \end{array}
  \]
  W.l.o.g., let $\TV(\alpha_1) \cap \TV(\alpha_2) = \emptyset$ (otherwise, the
  temporary variables in $\alpha_2$ can be renamed accordingly).  Moreover, let
  $\alpha_{1.2}$ be the rule
  \[
    f_1(\vect{x}) \tox{\cc_1 + \cc_2\mu}{}{} t\mu \constr{\phi_1 \land \phi_2\mu}
  \]
  and let $\PP' = \PP \cup \{\alpha_{1.2}\}$.  Then $\PP'$ is well formed and
  the processor that maps \(\PP\) to \(\PP'\) is sound.
\end{theorem}
\begin{proof}
  We prove the more general \Cref{thm:its-non-linear-chaining} in \Cref{sec:non-linear}.
\end{proof}

One goal of chaining is to eliminate all accelerated simple loops.
Therefore, after accelerating all simple loops,
we chain all subsequent rules $\alpha', \alpha$ where $\alpha$ is a simple loop
and $\alpha'$ is not a simple loop.  Afterwards, we delete $\alpha$.  Moreover,
once $\alpha'$ has been chained with all subsequent simple loops, then we also
remove $\alpha'$, since its effect is now (mostly\footnote{Since accelerated rules $\alpha_\itr$ do not cover the case
  where $\alpha$ is not executed at all (see Footnote \ref{coverage}), chaining $\alpha'$ with $\alpha_\itr$ and
 deleting these rules afterwards does not cover those original evaluations where $\alpha'$ was
 not followed by any subsequent application of $\alpha$. However, since we are only
 interested in witnesses for maximal evaluations, this does not affect the soundness of
 our approach.}) covered by the newly introduced
chained rules.

\begin{example}[Ex.\ \ref{ex:deletion} Continued]\label{ex:chaining}
  We continue the transformation of the program from \Cref{fig:its-leading-ex}. Now we
  chain  \(\alpha_0\) with the accelerated simple loop
\(\alpha_{\overline{1}}\), and we chain \(\alpha_3\) with  the accelerated simple loop
\(\alpha_{\overline{4}}\).  This yields the following
integer program:
  \[
      \begin{array}{l@{\;\;}l@{\;\;\;}c@{\;\;}ll}
        \alpha_{0.\overline{1}}\!:&\Ff_0(x,y,z,u) & \tox{1+x}{}{}{} & \Ff_1(0,
\frac{1}{2}x^2 + \frac{1}{2}  x,z,u) & \constr{x > 0}\\
        \alpha_2\!:&\Ff_1(x,y,z,u) & \tox{1}{}{}{} & \Ff_2(x,y,y,u) & \constr{x \leq 0}\\
        \alpha_{3.\overline{4}}\!:&\Ff_2(x,y,z,u) & \tox{z}{}{}{} & \Ff_3(x,y,z,0) & \constr{z > 1}\\
        \alpha_5\!:&\Ff_3(x,y,z,u) & \tox{1}{}{}{} & \Ff_2(x,y,z-1,u) & \constr{u \leq 0}
      \end{array}
      \]
      In Rule  $\alpha_{3.\overline{4}}$, we simplified the guard $z > 0 \land z-1 > 0$ to
      $z > 1$.
\end{example}

Chaining also allows us to eliminate function symbols from the program by chaining all pairs of
rules $\alpha'$ and $\alpha$ where $\dest(\alpha') = \head(\alpha)$ and removing them
afterwards.
In this way we can transform loops consisting of several transitions into simple loops.
It is advantageous to eliminate symbols which are the target of
just one single rule first.  This heuristic
avoids eliminating the entry points of loops, if possible.

\begin{example}[Ex.\ \ref{ex:chaining} Continued]\label{ex:chaining2}
  So for the program in \Cref{ex:chaining},
 it would avoid chaining \(\alpha_5\)
 and \(\alpha_{3.\overline{4}}\)
where $\dest(\alpha_5) = \head(\alpha_{3.\overline{4}}) =
\Ff_2$,
because $\Ff_2$ is also the target of the rule $\alpha_2$.
 In this way,
we avoid constructing chained rules that correspond to a run from the ``middle''
of a loop to the ``middle'' of the next loop iteration.

Instead, we chain \(\alpha_{0.\overline{1}}\) and
\(\alpha_{2}\) as well as \(\alpha_{3.\overline{4}}\) and \(\alpha_{5}\) to
eliminate the function symbols \(\Ff_1\) and \(\Ff_3\). This leads to the following program.
 \[
      \begin{array}{l@{\;\;}l@{\;\;\;}c@{\;\;}ll}
        \alpha_{0.\overline{1}.2}\!:&\Ff_0(x,y,z,u) & \tox{2+x}{}{}{} & \Ff_2(0,
 \frac{1}{2}  x^2\!+\!\frac{1}{2} x,\frac{1}{2}  x^2\!+\!\frac{1}{2} x,u) & \constr{x > 0}\\
        \alpha_{3.\overline{4}.5}\!:&\Ff_2(x,y,z,u) & \tox{1 + z}{}{}{} & \Ff_2(x,y,z-1,0) & \constr{z > 1}
      \end{array}
      \]
\end{example}

\begin{algorithm}[t]
  \begingroup
  \setlist{itemsep=0pt,topsep=0pt}
  \flushleft
  While there is a rule $\alpha$ with $\head(\alpha) \neq \Ff_0$:
  \begin{enumerate}[label*=\arabic*., ref=\arabic*]
  \item Apply \emph{Deletion} to rules whose guard is proved unsatisfiable or whose root symbol is unreachable from $\Ff_0$.
    \label{alg-step:its-delete-unsat}
  \item While there is a non-accelerated simple loop \(\alpha\):
    \label{alg-step:its-acceleration-loop}
    \begin{enumerate}[label*=\arabic*., ref=\arabic{enumi}.\arabic*, leftmargin=2em]
    \item Try to \emph{accelerate} \(\alpha\).
      \label{alg-step:its-meter}
    \item If \ref{alg-step:its-meter} succeeded, resulting in $\overline{\alpha}$:
      \begin{enumerate}[label*=\arabic*., ref=\arabic{enumi}.\arabic{enumii}.\arabic*, leftmargin=2em]
      \item Try to \emph{instantiate} $\overline{\alpha}$ to eliminate the temporary variable introduced in Step \ref{alg-step:its-meter}.
        \label{alg-step:non-linear-eliminate-tv}
      \item If \ref{alg-step:non-linear-eliminate-tv} succeeded, apply \emph{Deletion} to
        $\overline{\alpha}$.
        \label{alg-step:delete-instantiate}
      \end{enumerate}
    \item If \ref{alg-step:its-meter} failed and $\alpha$ uses temporary variables:\\
      Try to \emph{instantiate} \(\alpha\) to eliminate a temporary variable.
      \label{alg-step:its-strength}
    \item Apply \emph{Deletion} to \(\alpha\).
      \label{alg-step:its-delete}
    \end{enumerate}
  \item Let \(S = \emptyset\).
    \label{alg-step:its-initialize-S}
  \item While there is an accelerated rule \(\alpha\):
    \label{alg-step:its-loop-chaining-loop}
    \begin{enumerate}[label*=\arabic*., ref=\arabic{enumi}.\arabic*, leftmargin=2em]
    \item For each $\alpha'$ with \(\head(\alpha') \neq \dest(\alpha') = \head(\alpha)\):\\
      Apply \emph{Chaining} to \(\alpha'\) and \(\alpha\) and add
 \(\alpha'\) to \(S\).
    \item Apply \emph{Deletion} to \(\alpha\).
    \end{enumerate}
  \item Apply \emph{Deletion} to each rule in \(S\).
    \label{alg-step:its-delete-S}
  \item While there is a function symbol \(f\) without simple loops but with incoming and outgoing rules (starting with symbols $f$ with just one incoming rule):\label{alg-step:its-straight-chaining-loop}
    \begin{enumerate}[label*=\arabic*., ref=\arabic{enumi}.\arabic*, leftmargin=2em]
    \item Apply \emph{Chaining} to each pair \(\alpha',\alpha\) where \(\dest(\alpha') = \head(\alpha) = f\).
    \item Apply \emph{Deletion} to each \(\alpha\) where \(\head(\alpha) = f\)
       or \(\dest(\alpha) = f\).
    \end{enumerate}
  \end{enumerate}
  \caption{Program Simplification for Tail-Recursive Integer Programs}
  \label{alg:its-program-simplification}
  \endgroup
\end{algorithm}

Our overall approach for program simplification is shown in
\Cref{alg:its-program-simplification}.
It transforms any
tail-recursive integer program into a \emph{simplified} program.
  \Cref{sec:asymptotic} will
  show how to analyze the (asymptotic) complexity of
  simplified programs.
  Of course, other strategies for the
  application of the processors would be possible, too.
Recall that the application of \emph{Acceleration}, \emph{Instantiation}, and
\emph{Chaining} always generates new rules (so, e.g., in Step
\ref{alg-step:non-linear-eliminate-tv}, instantiating 
$\overline{\alpha}$ generates a new rule $\widetilde{\alpha}$ and the original
non-instantiated rule $\overline{\alpha}$ is \emph{deleted} in Step \ref{alg-step:delete-instantiate}).
  The set \(S\) in the
Steps \ref{alg-step:its-initialize-S} -- \ref{alg-step:its-delete-S} is needed
to handle function symbols \(f\) with multiple simple loops.  The reason is that
each rule \(\alpha'\) with \(\dest(\alpha') = f\) should be chained with
\emph{each} of \(f\)'s simple loops before removing $\alpha'$.

\Cref{alg:its-program-simplification} terminates: The loop in Step
\ref{alg-step:its-acceleration-loop} terminates since each iteration either
decreases the number of temporary variables in \(\alpha\) or reduces the number
of non-accelerated simple loops.  In Step \ref{alg-step:its-loop-chaining-loop},
the number of accelerated rules is decreasing and for the loop in Step
\ref{alg-step:its-straight-chaining-loop}, the number of function symbols
decreases.  The overall loop of \Cref{alg:its-program-simplification} terminates as it reduces the number of function
symbols.
  The reason is that the program does not have simple loops anymore when
  the algorithm reaches Step \ref{alg-step:its-straight-chaining-loop} (as simple loops where
  acceleration fails are deleted in Step \ref{alg-step:its-delete} and accelerated rules are
  eliminated in Step \ref{alg-step:its-loop-chaining-loop}).  Thus, at
this point there is either a function symbol \(f\) which can be eliminated or
the program does not have a path of length \(2\), i.e., all rules have the root
$\Ff_0$.

\begin{example}[Ex.\ \ref{ex:chaining2} Continued]\label{ex:chaining3}
According to \Cref{alg:its-program-simplification}, in our example we go back to
Step \ref{alg-step:its-delete-unsat} and \ref{alg-step:its-acceleration-loop}
and apply \emph{Loop Acceleration} to the rule \(\alpha_{3.\overline{4}.5}\).
This rule has the metering function \(z - 1\) and its iterated update sets \(u\)
to \(0\) and \(z\) to \(z - \tv\) for a fresh temporary variable \(\tv\).  To
compute \(\alpha_{3.\overline{4}.5}\)'s iterated cost, we have to find an
under-approximation for the solution of the recurrence equations \(\cc^{(1)} = 1 + z
\) and \(\cc^{(\tv+1)} = \cc^{(\tv)} + 1 + z^{(\tv)}\).  After computing the
closed form $z - \tv$ of \(z^{(\tv)}\), the second equation simplifies to
\(\cc^{(\tv+1)} = \cc^{(\tv)} + 1 + z - \tv\), which results in the closed form
$\cc_\itr = \cc^{(\tv)} = \tv \cdot z-\frac{1}{2} \tv^2 + \frac{3}{2}
\tv$. Thus, we obtain the accelerated rule
\[
  \Ff_2(x,y,z,u) \toxx{\tv\cdot z-\frac{1}{2}  \tv^2+\frac{3}{2}  \tv}{}{}{} \Ff_2(x,y,z-\tv,0) \constr{0 < \tv <  z}.
\]
By instantiating $\tv$ with $z-1$ in Step \ref{alg-step:non-linear-eliminate-tv} and removing
$\alpha_{3.\overline{4}.5}$ in Step \ref{alg-step:its-delete}, we obtain the
following program.
\[
      \begin{array}{l@{\;\;}l@{\;\;\;}c@{\;\;}ll}
        \alpha_{0.\overline{1}.2}\!:&\Ff_0(x,y,z,u) & \tox{2+x}{}{}{} & \Ff_2(0,
 \frac{1}{2}  x^2\!+\!\frac{1}{2}  x,\frac{1}{2} x^2\!+\!\frac{1}{2}  x,u) & \constr{x > 0}\vspace*{.2cm}\\
        \alpha_{\overline{3.\overline{4}.5}}\!:&\Ff_2(x,y,z,u) & \toxx{\frac{1}{2} z^2 + \frac{3}{2}z - 2}{}{}{} & \Ff_2(x,y,1,0) &
        \constr{z > 1}
      \end{array}
      \]
A final chaining
step and deletion of $\alpha_{0.\overline{1}.2}$ and
$\alpha_{\overline{3.\overline{4}.5}}$ yields the simplified program with the following
single rule.
\setcounter{eq:finalFirstLeadingExCtr}{\value{equation}}
\begin{equation}
  \label{eq:finalFirstLeadingEx}
  \hspace*{-.7cm} \begin{array}{l@{\;\;}l@{\;}c@{\;}l@{\;\;}l}
  \alpha_{0.\overline{1}.2.\overline{3.\overline{4}.5}}\!: &
\Ff_0(x,y,z,u) & \toxx{\frac{1}{8}x^4 + \frac{1}{4}x^3 + \frac{7}{8}x^2 + \frac{7}{4}x}{}{}{} &
\Ff_2(0, \frac{1}{2} x^2\!+\!\frac{1}{2} x, 1, 0) &   \constr{\frac{1}{2} x^2 + \frac{1}{2}  x > 1} \hspace*{-.7cm}
\end{array}
\end{equation}
\end{example}


\section{Simplifying Arbitrary Recursive Integer Programs}
\label{sec:non-linear}

So far, we only considered tail-recursive programs, i.e., programs where all
rules have the form $f(\vect{x}) \tox{}{}{} g(\vect{t}) \constr{\phi}$.  We now
extend our technique to non-tail-recursive programs, i.e., we now also consider rules where
the right-hand side is a multiset of several terms.

\Cref{thm:its-instantiation,thm:its-deletion} are trivially applicable to
non-tail-recursive programs as well, i.e., we can still
instantiate temporary variables and we can still delete rules.  However,
\Cref{thm:its-acceleration,thm:its-chaining} have to be adapted.  In 
\Cref{subsec:rec-meter}
we  extend our notion of metering functions to (non-tail\=/)recursive
rules in order to adapt \Cref{thm:its-acceleration}.  Similar to the case of
tail-recursive programs where we started with considering simple loops, we first focus on \emph{simple recursions}, i.e., rules
$f(\vect{x}) \tox{\cc}{}{} \TTT \constr{\phi}$ whose degree
  $|\TTT|$ is greater than $1$ and where all terms
in $\TTT$ have the root symbol $f$.\footnote{The reason for excluding the case
 $|\TTT| =1$ is that in this way, simple recursions give rise to \emph{exponential} lower
 bounds, i.e., we can formulate the corresponding \Cref{thm:NonLinearMeteringFunctions} which does not hold if  $|\TTT| =1$.}
Our extended notion of metering functions then allows us to accelerate simple recursions in \Cref{subsec:rec-acceleration}.
Afterwards, in \Cref{subsec:rec-simplification} we show how to transform
more complex recursions into simple recursions or simple loops
 via Chaining
and \emph{Partial Deletion}, a new technique which is specific to non-tail-recursive
integer programs. Based on these techniques, we extend \Cref{alg:its-program-simplification} to a
procedure which transforms any
 integer program
into a simplified program.

\subsection{Under-Estimating the Depth of Recursions}
\label{subsec:rec-meter}

To understand the idea of metering functions for simple recursions, note that
repeatedly applying a simple recursive rule $f(\vect{x}) \tox{\cc}{}{} \TTT \constr{\phi}$
essentially yields an \emph{evaluation tree} of terms where each
inner node has $|\TTT|$ successors.  While metering functions for simple loops
under-estimate the length of evaluations, metering functions for simple
recursions under-estimate the height up to which such evaluation trees are
complete.  Hence, if $\bound$ is a metering function for a simple recursion
$\alpha$ of degree $d$, then the maximal number of consecutive applications of
$\alpha$ is in $\Omega(d^{\bound})$.

\begin{definition}[Metering Function for Simple Recursions]
  \label{def:its-non-linear-meter}
  Let $\alpha$ be a rule of the form
  \[
     f(\vect{x}) \tox{\cc}{}{} \TTT \constr{\phi}
  \]
  such that $\TTT$ contains no function symbol from $\Sigma$ except $f$.  We call an
  arithmetic expression $\bound$ a \emph{metering function} for \(\alpha\) if the
  following conditions are valid:
  \begin{eqnarray}
    \label{eq:its-non-linear-meter1}\neg\guard(\alpha) &\implies& \bound \leq 0\\
    \label{eq:its-non-linear-meter2}\guard(\alpha) &\implies& \bound\{\vect{x} / \vect{t}\} \geq \bound - 1 \text{ for all } f(\vect{t}) \in \TTT
  \end{eqnarray}
\end{definition}

 \Cref{def:its-non-linear-meter} is a generalization of
 \Cref{def:its-meter}, i.e., if $\alpha$ has degree $1$, then
\Cref{def:its-non-linear-meter} and \Cref{def:its-meter} coincide.
Moreover, note that conditional metering functions of the form $\charfun{\psi}
\cdot \bound$ for simple recursions can be inferred analogously to conditional
metering functions for simple loops (see \Cref{lem:irrelevant-constraints}): If
$\guard(\alpha)$ is $\phi \land \psi$ and $\guard(\alpha)$ implies $\psi\{\vect{x} / \vect{t}\}$ for all
$f(\vect{t}) \in \TTT$, then it suffices to check $\neg\phi \land \psi \implies
\bound \leq 0$ and $\phi \land \psi \implies \bound\{\vect{x} / \vect{t}\} \geq
\bound - 1$ for all $f(\vect{t}) \in \TTT$ in order to prove that $\charfun{\psi} \cdot
\bound$ is a metering function for $\alpha$.

\begin{example}[Metering Function for Fibonacci]
  According to \Cref{def:its-non-linear-meter}, $\frac{1}{2}  x - 1$ is a
  metering function for the recursive Fibonacci rule \eqref{eq:fib2} from \Cref{ex:its-fib}.  It
  satisfies \eqref{eq:its-non-linear-meter1}, as we have $\neg (x > 1) \implies
  \frac{1}{2}  x - 1 \leq 0$.  The recursive call $\fs{fib}(x-1)$ satisfies
  \eqref{eq:its-non-linear-meter2}, since we have
  \[
   \begin{array}{lllllll} x > 1 &\implies &\frac{1}{2} \, (x-1) - 1 &= &\frac{1}{2}  x -
     \frac{3}{2} &\geq &\frac{1}{2}  x - 2.
     \end{array}
  \]
  Finally, the recursive call $\fs{fib}(x-2)$ also satisfies
  \eqref{eq:its-non-linear-meter2}, as we have
  \[
     \begin{array}{lllllll} x > 1 &\implies &\frac{1}{2} \, (x-2) - 1 &=& \frac{1}{2}
       x - 2 &\geq &\frac{1}{2} x - 2.
       \end{array}
    \]
  \end{example}

  For a simple loop $\alpha$, we directly use its metering function as a lower
  bound for the number of consecutive applications of $\alpha$. But for simple
recursions we can infer a lower bound that is higher than its metering function.
The reason is that when computing a metering function $\bound$ for a simple recursion
$f(\vect{x}) \to \{ f(\vect{t_1}), \ldots, f(\vect{t_d})\} \constr{\varphi}$, we proceed as if we had $d$
separate tail-recursive rules $f(\vect{x}) \to  f(\vect{t_1}) \constr{\varphi}$, \ldots, $f(\vect{x}) \to
f(\vect{t_d}) \constr{\varphi}$. However, as the original rule initiates $d$ new
evaluations in each step, we obtain a lower bound on
the length of evaluations that is exponential in $b$.
In other words, since the metering function $b$ under-estimates the height of complete
evaluation trees (where every non-leaf node has $d$ children), the number of edges of the
tree is in $\Omega(d^b)$.

\begin{theorem}[Metering Functions Under-Estimate Simple Recursions]
  \label{thm:NonLinearMeteringFunctions}
  Let $\bound$ be a metering function for a well-formed simple recursion $\alpha$ with
  degree $d$.  Then for all integer substitutions $\sigma$ with $\VV(\alpha) \subseteq
  \dom(\sigma)$, there is an evaluation $\lhs(\alpha)\,\sigma \tox{}{m}{\alpha}
  \SSS$ for some configuration $\SSS$ with $m \geq \frac{d^{\bound\sigma} - 1}{d-1}$.
\end{theorem}
\begin{proof}
  First note that we have $d > 1$ since $\alpha$ is a simple recursion and hence
  $\frac{d^{\bound\sigma} - 1}{d-1}$
  is well defined.  As in the proof of \Cref{thm:MeteringFunctions}, let
  $m_\sigma \in \NN \cup \{ \omega
  \}$ be the length of the longest evaluation which starts with
  $\lhs(\alpha)\,\sigma$ and only applies the rule $\alpha$ repeatedly.
We prove that $m_\sigma \geq \frac{d^{\bound\sigma} - 1}{d-1}$.

  The
  case $m_\sigma = \omega$ is trivial.  For $m_\sigma \neq \omega$, we
  use induction on $m_\sigma$. In
  the base case $m_\sigma = 0$,
  we have $\sigma \not\models \guard(\alpha)$ and thus,
\eqref{eq:its-non-linear-meter1}
implies $b \sigma \leq 0$. Hence, $\frac{d^{\bound\sigma} - 1}{d-1} \leq 0 = m_\sigma$.

For the induction step $m_\sigma \geq 1$,
we must have $\sigma \models \guard(\alpha)$. Let
$\lhs(\alpha) = f(\vect{x})$. Then we obtain
 \begin{eqnarray}
    \label{eq:NonLinearMeteringFunctions-2a}
    \bound \,\{ \vect{x}/\vect{t}\}\, \sigma  &\geq& \bound\sigma - 1 \qquad \text{for all
      $f(\vect{t}) \in \rhs(\alpha)$, $\;\;$ by \eqref{eq:its-non-linear-meter2}}\\
    \label{eq:NonLinearMeteringFunctions-2b}
    \lhs(\alpha)\,\sigma = f(\vect{x}) \sigma  &\tox{}{}{\alpha}& \rhs(\alpha)\,\sigma
  \end{eqnarray}

 Let $\rhs(\alpha) = \{
f(\vect{t_1}), \ldots, f(\vect{t_d}) \}$. Then due to
\eqref{eq:NonLinearMeteringFunctions-2b}, the longest evaluation $\lhs(\alpha)\,\sigma \tox{}{m_\sigma}{\alpha}
\SSS$ has the form
\[\lhs(\alpha)\,\sigma \;\;
\tox{}{}{\alpha} \;\;
\rhs(\alpha)\,\sigma \; = \;
 \{
f(\vect{t_1})\sigma, \ldots, f(\vect{t_d})\sigma \} \;\;
\tox{}{m_{\{ \vect{x}/\vect{t_1}\} \circ \sigma} + \ldots +
m_{\{ \vect{x}/\vect{t_d}\} \circ \sigma}}{\alpha} \;\;
  \SSS,\]
  i.e., $m_\sigma = m_{\{ \vect{x}/\vect{t_1}\} \circ \sigma} + \ldots +
  m_{\{ \vect{x}/\vect{t_d}\} \circ \sigma} +1$. Since $\VV(\alpha) \subseteq
\dom( \{ \vect{x}/\vect{t_i}\} \circ \sigma )$
and $\{\vect{x} / \vect{t_i}\} \circ \sigma$ is an integer
  substitution (since $\alpha$ is well formed), the induction hypothesis implies
  $m_{\{ \vect{x}/\vect{t_i}\} \circ \sigma} \geq \frac{d^{\bound\,\{ \vect{x}/\vect{t_i}\} \, \sigma} - 1}{d-1}
\geq \frac{d^{\bound\sigma - 1} - 1}{d-1}$ by \eqref{eq:NonLinearMeteringFunctions-2a} for all
$1 \leq i \leq d$. Hence, we have $m_\sigma =
m_{\{ \vect{x}/\vect{t_1}\} \circ \sigma} + \ldots +
m_{\{ \vect{x}/\vect{t_d}\} \circ \sigma} +1
\geq d \cdot \frac{d^{\bound\sigma - 1} - 1}{d-1} + 1
= \frac{d^{\bound\sigma} - d}{d-1} + 1
= \frac{d^{\bound\sigma} - 1}{d-1}$.
\end{proof}

\begin{example}[Under-Estimating Fibonacci]
  Since $\frac{1}{2}  x - 1$ is a metering function for the recursive
  Fibonacci rule \eqref{eq:fib2},
the term $\fs{fib}(x)\, \sigma$ starts an evaluation of at least length $2^{\frac{1}{2}
    \cdot x\sigma - 1} - 1$ for each integer substitution $\sigma$ with
  $x \in \dom(\sigma)$.
\end{example}

\subsection{Accelerating Simple Recursions}
\label{subsec:rec-acceleration}

Using \Cref{def:its-non-linear-meter} and \Cref{thm:NonLinearMeteringFunctions},
we can now accelerate simple recursions. In contrast to the acceleration of simple loops
in \Cref{thm:its-acceleration}, here we disregard the result of the accelerated rule
and replace its result with $\emptyset$.
The reason is that otherwise the degree of the accelerated rule (i.e., the number of elements in its right-hand
side) would depend on the instantiation of the variable $\tv$ that
represents the height of the evaluation tree. This
cannot be expressed in our program model.
Moreover, we cannot compute
iterated updates, as a simple recursion $\alpha$
has \emph{several} updates which may be applied
in arbitrary order
when $\alpha$ is
applied repeatedly.
Since computing the iterated cost would require the iterated updates,
we do not compute the iterated cost anymore, but simply under-estimate each evaluation step with
cost 1. This suffices to infer a lower bound for the cost of the
accelerated rule which is exponential in the metering function.

\begin{theorem}[Recursion Acceleration]
  \label{thm:its-non-linear-acceleration}
  Let $\PP$ be a well-formed integer program, let $\alpha \in \PP$ be a simple
  recursion of degree $d$ such that
  \begin{equation}
    \label{eq:its-non-linear-acceleration-positive-costs}
    \guard(\alpha) \implies \cost(\alpha) \geq 1
  \end{equation}
  is valid,\footnote{If \eqref{eq:its-non-linear-acceleration-positive-costs} is not
    valid, then one can simply add the constraint $\cost(\alpha) \geq 1$ to the guard of the
    rule, since adding constraints to the guard is always sound as it can only decrease the derivation height (by disallowing some evaluations).}
    and let $\charfun{\psi} \cdot \bound$ be a metering function for $\alpha$.  Moreover,
    let $\alpha'$ be the rule
  \[
     \lhs(\alpha) \tox{\cc}{}{} \emptyset \constr{\guard(\alpha) \land \psi} \text{ where } \cc = \frac{d^{\bound} - 1}{d-1},
  \]
  and let $\PP' = \PP \cup \{\alpha'\}$.  Then $\PP'$ is well formed and the
  processor that maps \(\PP\) to \(\PP'\) is sound.
\end{theorem}
\begin{proof}
  Well-formedness is trivial due to the empty right-hand side $\emptyset$ of the rule
  $\alpha'$.
  To prove soundness, note that by
  \Cref{thm:NonLinearMeteringFunctions}, $\sigma \models \alpha'$ implies
  $\lhs(\alpha)\,\sigma \tox{\concretecosts}{m}{\alpha} \SSS$ with $m \geq
 \frac{d^{\bound\sigma} - 1}{d - 1}$ for some cost $\concretecosts$ and some
 configuration $\SSS$ (since  $\sigma \models \guard(\alpha')$
implies
 $\sigma \models  \psi$ which in turn implies
 $(\charfun{\psi} \cdot \bound) \, \sigma = \bound\sigma$).  Since $\guard(\alpha) \implies \cost(\alpha) \geq 1$ is valid,
  we get $\concretecosts \geq m \geq \frac{d^{b\sigma} - 1}{d - 1} = \cc\sigma$.
  Thus, for every evaluation with
  $\alpha'$, there is an evaluation with $\alpha$ which has at least the same
  cost.
\end{proof}

\begin{example}[Accelerating Fibonacci]
  \label{ex:its-accelerating-fib}
  Since the rule \eqref{eq:fib2} from \Cref{ex:its-fib} has cost $1$,
  \eqref{eq:its-non-linear-acceleration-positive-costs} is trivially valid.
  Thus, accelerating the rule \eqref{eq:fib2} yields
  \[
    \fs{fib}(x) \tox{2^{\frac{1}{2}  x - 1} - 1}{}{} \emptyset \constr{x > 1}.
  \]
  Afterwards, chaining this rule with \eqref{eq:fib1}
  and deleting all other rules yields the simplified program with the rule
\setcounter{eq:FibonacciAcceleratedCtr}{\value{equation}}  
  \begin{equation}
    \label{eq:FibonacciAccelerated}
    \Ff_0(x) \tox{2^{\frac{1}{2}  x - 1} - 1}{}{} \emptyset \constr{x > 1}.
  \end{equation}
\end{example}

\subsection{Simplifying Recursive Rules}
\label{subsec:rec-simplification}

\Cref{subsec:rec-acceleration} showed how to accelerate simple recursions.
If our implementation fails in applying \Cref{thm:its-non-linear-acceleration} to a simple recursion,
  then we again try to eliminate
temporary variables via \emph{Instantiation}, as in the case of simple loops.
  If
  all
 temporary variables were eliminated and we still fail to accelerate a simple
recursion, then further simplifications are possible. In particular, we can remove parts of the right-hand side via \emph{Partial
  Deletion}.

\begin{theorem}[Partial Deletion]
  \label{thm:its-partial-deletion}
  Let $\PP$ be a well-formed integer program and let $\alpha \in \PP$.
  Moreover, let $\alpha'$ be like $\alpha$, but $\rhs(\alpha') \subset
  \rhs(\alpha)$, and let $\PP' = \PP \cup \{\alpha'\}$.  Then $\PP'$ is well
  formed and the processor that maps \(\PP\) to \(\PP'\) is sound.
\end{theorem}
\begin{proof}
  Since $\PP$ is well formed, $\PP'$ is trivially well formed, too.  To prove
  soundness, we define
  \[\mbox{$\SSS \tox{\concretecosts}{}{\PP \circ \supseteq} \TTT\;\;$ 
  if $\;\;\SSS \tox{\concretecosts}{}{\PP} \circ \supseteq \TTT$.}\]
So $\SSS \tox{\concretecosts}{}{\PP \circ \supseteq} \TTT$ holds if $\SSS$ evaluates to
some configuration $\SSS'$ with cost $k$ (i.e., $\SSS \tox{\concretecosts}{}{\PP} \SSS'$)
and the configuration $\TTT$ results from $\SSS'$ by deleting some terms from $\TTT$
(i.e., $\SSS' \supseteq \TTT$).
  Then we clearly
  have $\dht_{\PP'}(\SSS) \leq \sup\{\concretecosts \in \RR \mid
 \SSS \tox{\concretecosts}{*}{\PP \circ \supseteq} \TTT \text{ for some }
  \TTT \in \CC
 \}$, as each
  $\tox{}{}{\PP'}$-sequence is also a  $\tox{}{}{\PP \circ
    \supseteq}$-sequence. To finish the proof, we show
  \[
   \sup\{\concretecosts \in \RR \mid
 \SSS \tox{\concretecosts}{*}{\PP \circ \supseteq} \TTT \text{ for some }
  \TTT \in \CC
 \} \leq \dht_{\PP}(\SSS).
 \]
This clearly implies $\rc_{\PP'}(n) \leq \rc_\PP(n)$, i.e., it implies that the processor
is sound.

  Consider an evaluation $\SSS_0 \tox{\concretecosts_1}{*}{\PP \circ \supseteq}
  \ldots \tox{\concretecosts_m}{*}{\PP \circ \supseteq} \SSS_m$.  We prove
  $\SSS_0 \tox{\concretecosts_1}{*}{\PP} \ldots \tox{\concretecosts_m}{*}{\PP}
  \SSS'_m \supseteq \SSS_m$ for some $\SSS'_m \in \CC$ by induction on $m$. The
  case $m = 0$ is trivial. If $m > 0$, the induction hypothesis implies
  \[
    \SSS_0 \tox{\concretecosts_1}{*}{\PP} \ldots \tox{\concretecosts_{m-1}}{*}{\PP} \SSS'_{m-1} \supseteq \SSS_{m-1} \text{ for some } \SSS'_{m-1} \in \CC.
  \]
  Moreover, $\SSS_{m-1} \tox{\concretecosts_m}{}{\PP \circ \supseteq} \SSS_m$
  implies $\SSS_{m-1} \tox{\concretecosts_m}{}{\PP} \widetilde{\SSS_m} \supseteq
  \SSS_m$ for some $\widetilde{\SSS_m} \in \CC$, i.e., we have $s
  \tox{\concretecosts_m}{}{\PP} \QQQ$ and $\widetilde{\SSS_m} = (\SSS_{m-1}
  \setminus \{s\}) \cup \QQQ$ for some $s \in \SSS_{m-1}$ and some $\QQQ \in
  \CC$. Since $\SSS_{m-1} \subseteq \SSS'_{m-1}$,
we get
  \[
    \SSS'_{m-1} \tox{\concretecosts_m}{}{\PP} (\SSS'_{m-1} \setminus \{s\}) \cup \QQQ = \SSS'_m
  \]
  and thus $\SSS_0
  \tox{\concretecosts_1}{*}{\PP} \ldots \tox{\concretecosts_{m}}{*}{\PP}
  \SSS'_{m}$ with $\SSS'_{m} \supseteq \widetilde{\SSS_m} \supseteq \SSS_m$, as
  desired.
\end{proof}

\begin{example}[Partial Deletion to Enable Recursion Acceleration]
  \label{ex:partial-deletion}
  Consider the simple recursion $\Ff(x,y) \tox{}{}{} \{ \Ff(x-1,y), \Ff(x-y,y)
  \} \constr{x > 0 \land y > x}$. As
  $\Ff(x-y,y)$ cannot be reduced any further if $y > x$, we cannot find
    a useful metering
  function for this rule and hence, \emph{Recursion
    Acceleration} fails. By applying \emph{Partial Deletion}, we obtain the
  simple loop $\Ff(x,y) \tox{}{}{} \Ff(x-1,y) \constr{x > 0 \land y > x}$, which can easily
  be accelerated via \Cref{thm:its-acceleration}. In this way, we can infer that the
  original non-tail-recursive rule can be applied at least linearly often.
\end{example}

As shown above, \emph{Recursion Acceleration} is useful to handle programs with non-linear recursion
like the Fibonacci program, where the result is composed of two recursive calls.  However,
non-tail-recursion also occurs when composing a recursive call with the call of an
auxiliary function.

\begin{example}[Non-Tail-Recursive $\fs{facSum}$ Program \cite{koat}]
  \label{ex:its-fac}
  Consider the following imperative program.
\begin{lstlisting}
int fac(int x) {
  if (x > 1) return x * fac(x-1);
  else return 1;
}

int facSum(int x) {
  if (x > 0) return fac(x) + facSum(x-1);
  else return 1;
}
\end{lstlisting}
  Here, \code{fac(x)} computes $\code{x}!$ and \code{facSum(x)} computes $\code{0}! + \ldots +
  \code{x}!$. The program is not tail-recursive, because the last action of $\code{facSum}$ is
  not the recursive call, but an addition. The integer program below represents its recursive
  structure, i.e., it can be obtained from the above program by a suitable
  abstraction.
  \begin{alignat}{2}
      \Ff_0(x) & \tox{0}{}{}  \fs{facSum}(x) && \label{facsum-1} \\
      \fs{facSum}(x) & \tox{1}{}{}  \{\fs{fac}(x),  \fs{facSum}(x - 1) \} && \constr{x >
        0} \label{facsum-2} \\
      \fs{facSum}(x) & \tox{1}{}{}  \emptyset && \constr{x \leq 0}  \nonumber \\
      \fs{fac}(x) & \tox{1}{}{}  \fs{fac}(x - 1) && \constr{x > 1} \label{fac-rec}\\
      \fs{fac}(x) & \tox{1}{}{}  \emptyset && \constr{x \leq 1} \label{fac-non-rec}
  \end{alignat}
  \end{example}

To analyze this integer program, we first accelerate and chain the recursive
rule \eqref{fac-rec} as in \Cref{thm:its-acceleration} and \Cref{thm:its-chaining}.

\begin{example}[Accelerating $\fs{fac}$]
  \label{ex:its-accelerate-fac}
  Clearly, $x - 1$ is a metering function for the rule \eqref{fac-rec}. Accelerating it
using this metering function yields
  \[
    \fs{fac}(x) \tox{\tv}{}{}  \fs{fac}(x - \tv)  \constr{x > 1 \land 0 < \tv < x}.
  \]
  Instantiating $\tv$ with $x-1$ via \Cref{thm:its-instantiation}
  results in
   \begin{equation}\label{fac-rec-acc}
    \fs{fac}(x) \tox{x-1}{}{}  \fs{fac}(1)  \constr{x > 1}.
  \end{equation}
\end{example}

At this point, we would like to chain the recursive $\fs{facSum}$-rule \eqref{facsum-2} with the
$\fs{fac}$-rule \eqref{fac-rec-acc}.  However,
\Cref{thm:its-chaining} is only applicable to tail-recursive rules.  Hence, we
now generalize \Cref{thm:its-chaining} to arbitrary rules.

\begin{theorem}[Chaining for Arbitrary Integer Programs]
  \label{thm:its-non-linear-chaining}
  Let $\PP$ be a well-formed integer program and let $\alpha_1,\alpha_2 \in \PP$
  where
  \[
    \begin{array}{rccll}
      \alpha_1: & f_1(\vect{x}) &\tox{\cc_1}{}{}& \SSS & \constr{\phi_1} \text{ with } f_2(\vect{x})\,\mu \in \SSS \text{ and}\\
      \alpha_2: & f_2(\vect{x}) &\tox{\cc_2}{}{}& \TTT & \constr{\phi_2}.
    \end{array}
  \]
  W.l.o.g., let $\TV(\alpha_1) \cap \TV(\alpha_2) = \emptyset$ (otherwise, the
  temporary variables in $\alpha_2$ can be renamed accordingly).  Moreover, let
  $\alpha_{1.2}$ be the rule
  \[
     f_1(\vect{x}) \tox{\cc_1 + \cc_2\mu}{}{} (\SSS \setminus
    \{f_2(\vect{x})\,\mu\}) \cup \TTT\,\mu \constr{\phi_1 \land \phi_2\mu}
  \]
  and let $\PP' = \PP \cup \{\alpha_{1.2}\}$.  Then $\PP'$ is well formed and
  the processor that maps \(\PP\) to \(\PP'\) is sound.
\end{theorem}
\begin{proof}
  To
   prove the soundness of the processor, we show that every evaluation step with $\alpha_{1.2}$
  can be simulated by two evaluation steps with the rules $\alpha_1, \alpha_2$
  of the same cost.  Let $\sigma$ be an integer substitution with $\sigma
  \models \alpha_{1.2}$. Then we have
  \[
    f_1(\vect{x})\,\sigma \tox{\cc_1\sigma + \cc_2\mu\sigma}{}{\alpha_{1.2}} (\SSS\sigma
    \setminus \{f_2(\vect{x})\,\mu\sigma\}) \cup \TTT\mu\sigma.
  \]
  Since $\sigma \models \phi_1$, we have
  \[
    f_1(\vect{x})\,\sigma \tox{\cc_1\sigma}{}{\alpha_1} \SSS\sigma.
  \]
  Since $\sigma \models \phi_2\mu$ implies $\mu \circ \sigma \models \phi_2$
  we have
  \[
    f_2(\vect{x})\,\mu\sigma \tox{\cc_2\mu\sigma}{}{\alpha_2} \TTT\mu\sigma.
  \]
  As $f_2(\vect{x})\mu \in \SSS$, this implies
  \[
    \SSS\sigma \tox{\cc_2\mu\sigma}{}{\alpha_2} (\SSS\sigma \setminus
    \{f_2(\vect{x})\,\mu\sigma\}) \cup  \TTT\mu\sigma.
  \]
  Thus, we have $f_1(\vect{x})\sigma \tox{\cc_1\sigma +
    \cc_2\mu\sigma}{2}{\PP} (\SSS\sigma \setminus \{f_2(\vect{x})\,\mu\sigma\})
  \cup \TTT\mu\sigma$, as desired.
As $\alpha_1$ and $\alpha_2$ are well formed, this also proves that $\PP'$ is well formed.
\end{proof}

Note that \Cref{thm:its-non-linear-chaining} coincides with \Cref{thm:its-chaining} if the degree of $\alpha_1$ and
$\alpha_2$ is $1$.
\Cref{thm:its-non-linear-chaining} allows us to continue the transformation of the program
in \Cref{ex:its-accelerate-fac}.

\begin{example}[Chaining $\fs{facSum}$ and $\fs{fac}$]
  \label{ex:its-non-linear-chaining-facsum}
  Chaining the recursive $\fs{facSum}$-rule \eqref{facsum-2} of \Cref{ex:its-accelerate-fac}
  with the accelerated $\fs{fac}$-rule \eqref{fac-rec-acc} yields
  \[
    \fs{facSum}(x) \tox{x}{}{} \{ \fs{facSum}(x - 1), \fs{fac}(1) \} \constr{x > 1}.
    \]
 Chaining this rule with the non-recursive $\fs{fac}$-rule
\eqref{fac-non-rec}
  results in
  \[
    \fs{facSum}(x) \tox{x+1}{}{} \fs{facSum}(x - 1) \constr{x > 1}.
  \]
  The iterated update and cost of this rule are $x\mu^\tv = x - \tv$ and
  \[
  \begin{array}{lllll}
  \sum_{i=0}^{\tv-1} (1 + x)\,\mu^i  &= &\sum_{i=0}^{\tv-1} (1 + x - i) &= &x \cdot \tv -
  \frac{1}{2}  \tv^2 + \frac{3}{2}  \tv.
  \end{array}
  \]
  Thus, accelerating it via \Cref{thm:its-acceleration} with the metering
  function $x - 1$ results in
  \begin{equation}
    \label{eq:its-facsum-finish-1}
    \fs{facSum}(x) \tox[0pt]{x \cdot \tv - \frac{1}{2} \tv^2 + \frac{3}{2} \tv}{}{} \fs{facSum}(x - \tv) \constr{0 < \tv < x}
  \end{equation}
  since $0 < \tv < x$ implies $x > 1$.
  By instantiating $\tv$ with $x-1$ and chaining \eqref{facsum-1} with
  the resulting rule, we obtain
  \begin{equation}
    \label{eq:FacSumSimplified}
    \Ff_0(x) \tox[0pt]{\frac{1}{2}x^2 + \frac{3}{2} x - 2}{}{} \fs{facSum}(1) \constr{1 < x}.
  \end{equation}
  Finally, by deleting all other rules, we obtain a simplified program.
 \end{example}

\Cref{alg:non-linear-program-simplification} shows how
\Cref{alg:its-program-simplification} can be adapted in order to handle
non-tail-recursive programs as well.
The first additional step is
Step \ref{alg-step:non-linear-delete-sinks}, which deletes sinks (i.e., function
symbols without any rules) from right-hand sides of non-tail-recursive rules $\alpha$.
This simplifies the rules and possibly even transforms them into
tail-recursive rules.\footnote{If the rule
is already tail-recursive, i.e., the
right-hand side only contains a single term, and this term is a sink, then deleting this term
does not help much to simplify the program. As mentioned in \Cref{sink footnote}, we transform rules
with empty right-hand side into rules with the right-hand side  ``$\fs{sink}$'' to simplify the
formalization.} Note that \emph{Partial Deletion} adds a new rule $\alpha'$ with fewer terms
in its right-hand side (thus, we then \emph{delete} the original rule $\alpha$ afterwards).

The second change is that we
apply \emph{Partial Deletion} in Step \ref{alg-step:linearize} if we failed to accelerate
a simple recursion that does not contain any temporary variables anymore.
In this way, the degree of the simple recursion is
reduced, which may simplify its acceleration as in \Cref{ex:partial-deletion}. In our
implementation, we first try all partial deletions which result in rules of degree 2 (if any).
If we fail to accelerate any of the resulting rules, then we also try all partial
deletions which result in rules of degree 1.

\Cref{alg:non-linear-program-simplification} terminates and thus it transforms any
integer program into a simplified program: The loop in Step
\ref{alg-step:non-linear-delete-sinks} reduces the degree of some rule in each
iteration.
The loop in Step \ref{alg-step:non-linear-acceleration-loop} either reduces the number
of non-accelerated loops or recursions, or it reduces the number of temporary variables or
the degree of some rule in each iteration. In Step
\ref{alg-step:non-linear-loop-chaining-loop},  the number of accelerated rules
is decreasing.
The loop in Step \ref{alg-step:non-linear-straight-chaining-loop} terminates as it
reduces the number of function symbols with outgoing rules in each
iteration. Finally, the overall loop of
\Cref{alg:non-linear-program-simplification}
terminates as well, because after having finished Step
\ref{alg-step:non-linear-straight-chaining-loop} the first time,
the number of function symbols decreases in each further iteration of the algorithm.
To see this, note that
there is no simple loop or simple recursion anymore when the loop of Step
\ref{alg-step:non-linear-loop-chaining-loop} terminates. Thus, after finishing the loop of Step
\ref{alg-step:non-linear-straight-chaining-loop},  the only function
symbol is either $\Ff_0$  (and hence the overall loop terminates) or there is at least one function
symbol $f \neq \Ff_0$ without incoming or without outgoing rules.
Thus, in the
next iteration, this function symbol is removed. The reason is that
if $f$ has no incoming
rules, then all rules $\alpha$ with $\head(\alpha) = f$ are deleted in Step
\ref{alg-step:non-linear-delete-unsat}. If
$f$ has no outgoing rules, then all occurrences of $f$ in right-hand sides
are deleted in Step \ref{alg-step:non-linear-delete-sinks}.

\begin{algorithm}[t]
  \begingroup
  \setlist{itemsep=0pt,topsep=0pt}
  \flushleft
  While there is a rule $\alpha$ with $\head(\alpha) \neq \Ff_0$:
  \begin{enumerate}[label*=\arabic*., ref=\arabic*]
  \item Apply \emph{Deletion} to rules whose guard is proved unsatisfiable or whose root symbol is unreachable from $\Ff_0$.
    \label{alg-step:non-linear-delete-unsat}
  \item While there is a non-tail-recursive rule $\alpha$
    whose right-hand side contains a symbol $f$ without outgoing rules:
    \label{alg-step:non-linear-delete-sinks}
    \begin{enumerate}[label*=\arabic*., ref=\arabic{enumi}.\arabic*, leftmargin=2em]
    \item Apply \emph{Partial Deletion} to an occurrence of $f$ in $\rhs(\alpha)$ and
      apply \emph{Deletion} to $\alpha$ afterwards.
    \end{enumerate}
  \item While there is a simple recursion or a non-accelerated simple loop \(\alpha\):
    \label{alg-step:non-linear-acceleration-loop}
    \begin{enumerate}[label*=\arabic*., ref=\arabic{enumi}.\arabic*, leftmargin=2em]
    \item Try to \emph{accelerate} \(\alpha\).
      \label{alg-step:non-linear-meter}
    \item If \ref{alg-step:non-linear-meter} succeeded, resulting in $\overline{\alpha}$, and $\alpha$ is a simple loop:
      \begin{enumerate}[label*=\arabic*., ref=\arabic{enumi}.\arabic{enumii}.\arabic*, leftmargin=2em]
      \item Try to \emph{instantiate} $\overline{\alpha}$ to eliminate the temporary variable introduced in Step \ref{alg-step:non-linear-meter}.
        \label{alg-step:its-eliminate-tv}
      \item If \ref{alg-step:its-eliminate-tv} succeeded, apply \emph{Deletion} to $\overline{\alpha}$.
      \end{enumerate}
    \item If \ref{alg-step:non-linear-meter} failed:\\
      If $\TV(\alpha) \neq \emptyset$, then try to \emph{instantiate} \(\alpha\) to eliminate a temporary variable.\\
      Otherwise, if the degree of $\alpha$ is greater than $1$, then apply \emph{Partial Deletion} to $\alpha$.
      \label{alg-step:linearize}
    \item Apply \emph{Deletion} to \(\alpha\).
      \label{alg-step:non-linear-delete}
    \end{enumerate}
  \item Let \(S = \emptyset\).
    \label{alg-step:non-linear-initialize-S}
  \item While there is an accelerated rule \(\alpha\):
    \label{alg-step:non-linear-loop-chaining-loop}
    \begin{enumerate}[label*=\arabic*., ref=\arabic{enumi}.\arabic*, leftmargin=2em]
    \item For each $\alpha'$
      where $\rhs(\alpha')$ contains $\head(\alpha)$
and where $\head(\alpha') \neq \head(\alpha)$:\\
      Apply \emph{Chaining} to \(\alpha'\) and \(\alpha\) and add \(\alpha'\) to \(S\).
    \item Apply \emph{Deletion} to \(\alpha\).
    \end{enumerate}
  \item Apply \emph{Deletion} to each rule in \(S\).
    \label{alg-step:non-linear-delete-S}
  \item While there is a function symbol \(f\) without simple recursions or simple loops but with incoming and outgoing rules (starting with symbols $f$ with just one incoming rule):\label{alg-step:non-linear-straight-chaining-loop}
    \begin{enumerate}[label*=\arabic*., ref=\arabic{enumi}.\arabic*, leftmargin=2em]
    \item Apply \emph{Chaining} to each pair \(\alpha',\alpha\) where \(\head(\alpha) = f\) occurs in $\rhs(\alpha')$.
    \item Apply \emph{Deletion} to each \(\alpha\) where \(\head(\alpha) = f\) or where
      $\rhs(\alpha)$ contains no function symbol from $\Sigma$ except $f$.
    \end{enumerate}
  \end{enumerate}
  \caption{Program Simplification for Arbitrary Integer Programs}
  \label{alg:non-linear-program-simplification}
  \endgroup
\end{algorithm}


\section{Asymptotic Lower Bounds}
\label{sec:asymptotic}

After applying \Cref{alg:non-linear-program-simplification}, all programs are simplified and thus,
we assume that
\(\PP\) is a simplified
program throughout this section.
So all rules $\alpha \in \PP$ have the same left-hand side $\Ff_0(\vect{x})$.
Now for any integer substitution $\sigma$,
the derivation height of $\Ff_{0}(\vect{x})\,\sigma$ in the simplified program $\PP$ is
\begin{equation}
  \label{eq:its-concrete-lower-bound}
 \dht_{\PP}(\Ff_{0}(\vect{x})\,\sigma) \; = \; \max \,\{\cost(\alpha)\, \sigma \mid \alpha \in \PP, \sigma \models \alpha\},
\end{equation}
 i.e.,
\eqref{eq:its-concrete-lower-bound} is the maximal cost of those rules whose guard is
satisfied by $\sigma$.
Thus, if $\PP$ results from the transformation of an integer program $\widetilde{\PP}$, then
\eqref{eq:its-concrete-lower-bound} is a lower bound on  $\dht_{\widetilde{\PP}}(\Ff_{0}(\vect{x})\,\sigma)$.
So for the  program in
\Cref{fig:its-leading-ex} which was transformed into the
simplified program with the only rule
\setcounter{auxctr}{\value{equation}}
\setcounter{equation}{\value{eq:finalFirstLeadingExCtr}}
\begin{equation}
   \hspace*{-.7cm} \begin{array}{l@{\;\;}l@{\;}c@{\;}l@{\;\;}l}
  \alpha_{0.\overline{1}.2.\overline{3.\overline{4}.5}}\!: &
\Ff_0(x,y,z,u) & \toxx{\frac{1}{8}x^4 + \frac{1}{4}x^3 + \frac{7}{8}x^2 + \frac{7}{4}x}{}{}{} &
\Ff_2(0, \frac{1}{2} x^2\!+\!\frac{1}{2} x, 1, 0) &   \constr{\frac{1}{2} x^2 +
  \frac{1}{2}  x > 1},  \hspace*{-.7cm}
\end{array}
\end{equation}
we obtain the lower bound
\setcounter{equation}{\value{auxctr}}
\begin{equation}
  \label{eq:its-concrete-lower-bound-ex}
  \tfrac{1}{8}x^4 + \tfrac{1}{4}x^3 + \tfrac{7}{8}x^2 + \tfrac{7}{4}x
\end{equation}
for all integer substitutions with $\sigma \models \frac{1}{2}  x^2 + \frac{1}{2} x >
1$.
However, in general such bounds do not provide an intuitive understanding of the program's complexity
and they are also not suitable to detect possible attacks.
The reason is that both $\cost(\alpha)$ and $\guard(\alpha)$ may be complicated and, even more
importantly, they may be interdependent. Then,  it is not sufficient to only regard the
cost of  a rule in order to draw
conclusions on the resulting complexity.
To see this, consider a simplified rule with cost $\tv$ and guard $\phi$.
Its complexity can, e.g., be unbounded if $\phi$ does not impose any bound on
$\tv$, exponential if
 $\phi$ only implies
$\tv \leq b$ for
arithmetic expressions $b$ that are
exponential in the program variables, or constant if $\phi$ implies $\tv \leq e$ for some
$e \in \NN$.
But even without temporary variables, there can be subtle interdependencies between the cost and the guard of a transition, as the following example illustrates.

\begin{example}[Sub-Linear Bounds]
  \label{ex:its-sqrt}
  Let
  \[
  \PP = \left\{\Ff_0(x,y) \tox[-1pt]{y}{}{} \Ff(x,y) \constr{x > y^2}\right\}.
  \]
 The runtime complexity of this program is sub-linear, even though the cost function $y$
 is linear.
 The reason is that to achieve a linear increase of the cost $y$, a quadratic increase of
 $x$ and thus of the input size is required. So in general, it is not correct to simply take the cost of a rule as a
 lower bound on its runtime (e.g., in this example
 we have $\rc_{\PP}(n) \in\Omega(\sqrt{n})$, whereas 
 ``$\rc_{\PP}(n) \in \Omega(n)$'' would be incorrect). 
\end{example}

Hence, we now show how to derive \emph{asymptotic} lower bounds for simplified programs.
These asymptotic bounds can easily be understood (i.e., a high lower bound can help
programmers to improve their program to make it more efficient) and they identify
potential attacks.

To derive asymptotic bounds, we use so-called \emph{limit problems},
which are introduced in  \Cref{subsec:its-asymptotic-bounds-and-limit-problems}.
A limit problem is an abstraction of the guard $\phi$ of a rule which allows us to analyze
how to satisfy $\phi$, presuming that all variables are instantiated with ``large enough''
values.
More precisely, a solution of a limit problem is a \emph{family} of substitutions
$\sigma_n$ which is parameterized by a variable $n$.
This family of substitutions satisfies $\phi$ for large enough $n$ and can be found using
the calculus presented in \Cref{subsec:its-transforming-limit-problems}.
Thus, applying $\sigma_n$ to the cost of a rule yields an expression that only contains the
single variable $n$, even if the rule has a multivariate cost function. Hence, this allows us
to deduce an asymptotic bound. In \Cref{sec:asymptotic-smt} we present an alternative
approach to find solutions of limit problems via SMT solving which can be combined with
the calculus of  \Cref{subsec:its-transforming-limit-problems} in order to improve efficiency.

\subsection{Limit Problems}
\label{subsec:its-asymptotic-bounds-and-limit-problems}

While the derivation height $\dht_{\PP}$ is defined on configurations like
$\Ff_0(\vect{x})\,\sigma$, asymptotic bounds are usually defined for functions on $\NN$
like the runtime complexity $\rc_\PP$. Recall that
according to  \Cref{def:runtime complexity},
$\rc_\PP(n)$ is the maximal cost of
any evaluation starting with a configuration $\Ff_0(\vect{n})$ where the size
$\size{\vect{n}}$ of the input is at most $n$.
Thus, our goal is to derive an asymptotic lower bound for \(\rc_{\PP}\) from a concrete
bound on $\dht_{\PP}$ like \eqref{eq:its-concrete-lower-bound-ex}.
So for the  program $\PP$ in \eqref{eq:finalFirstLeadingEx},
we would like to derive $\rc_{\PP}(n) \in \Omega(n^4)$.
However, as discussed above, in general the cost \eqref{eq:its-concrete-lower-bound-ex}
of a rule does not directly give rise to the desired asymptotic lower bound.

To infer an asymptotic lower bound from a rule \(\alpha \in \PP\), we try to find an
infinite family of integer substitutions $\sigma_n$ with $\VV(\alpha) \subseteq
\dom(\sigma_n)$ (parameterized by $n \in \NN$) such
that there is an $n_0 \in \NN$ with $\sigma_n \models \guard(\alpha)$ for all $n \geq
n_0$. Note that both $\size{\vect{x}\sigma_n}$ and $\cost(\alpha)\,\sigma_n$ are
arithmetic expressions
that only contain the single variable $n$, and we have
$\rc_{\PP}(\size{\vect{x}\sigma_n}) \in \Omega(\cost(\alpha)\,\sigma_n)$, since for all $n
\geq n_0$ we obtain
\[
    \rc_{\PP}(\size{\vect{x}\sigma_n}) \geq  \dht_{\PP}(\Ff_0(\vect{x})\,\sigma_n) \geq
    \cost(\alpha)\,\sigma_n.
\]
For the program $\PP$ containing only the rule
\setcounter{auxctr}{\value{equation}}
\setcounter{equation}{\value{eq:finalFirstLeadingExCtr}}
\begin{equation}
   \hspace*{-.7cm} \begin{array}{l@{\;\;}l@{\;}c@{\;}l@{\;\;}l}
  \alpha_{0.\overline{1}.2.\overline{3.\overline{4}.5}}\!: &
\Ff_0(x,y,z,u) & \toxx{\frac{1}{8}x^4 + \frac{1}{4}x^3 + \frac{7}{8}x^2 + \frac{7}{4}x}{}{}{} &
\Ff_2(0, \frac{1}{2} x^2\!+\!\frac{1}{2} x, 1, 0) &   \constr{\frac{1}{2} x^2 +
  \frac{1}{2}  x > 1},  \hspace*{-.7cm}
\end{array}
\end{equation}
our approach will infer the family
$\sigma_n$ with
\setcounter{equation}{\value{auxctr}}
\begin{equation}
  \label{eq:sigmanFirstLeadingEx}
    x\sigma_n = n \quad \text{ and } \quad
    y\sigma_n = z\sigma_n = u\sigma_n = 0.
\end{equation}
So for $\vect{x}$ consisting of $x, y, z, u$, this implies \[
\begin{array}{rcl}
\rc_\PP(\size{\vect{x} \sigma_n}) &=&
\rc_\PP(\size{x\sigma_n} + \size{y\sigma_n} + \size{z\sigma_n} + \size{u\sigma_n})\\
&=&
\rc_\PP(|n|)\\
&\geq& \cost(\alpha_{0.\overline{1}.2.\overline{3.\overline{4}.5}})\,\sigma_n\\
&=&
\eqref{eq:its-concrete-lower-bound-ex} \,\sigma_n\\
&=& \frac{1}{8}  n^4 + \frac{1}{4}  n^3 + \frac{7}{8}n^2 + \frac{7}{4} n.
\end{array}
\]

To find such a family of substitutions $\sigma_n$, we first normalize all constraints in
  \(\guard(\alpha)\) such that they have the form \(a > 0\) or $a \geq 0$.\footnote{Note
    that while the variables range over $\ZZ$, there may be non-integer expressions in
    $\guard(\alpha)$ which result from non-integer metering functions. Thus, we allow
    both constraints of the form $a > 0$ and $a \geq 0$ in normalized guards, since
    transforming $a > 0$ to $a -1 \geq 0$ would be incorrect in general.}
  Now we search for substitutions $\sigma_n$ such that for large enough $n \in \NN$,
  $\sigma_n$ is a model for a formula of the form
``\(\bigwedge_{i = 1}^k (a_i \, \circ_i \, 0)\)'' where $\circ_i \in \{{>},{\geq}\}$.
Obviously, such a formula is satisfied for large enough $n$ if all expressions
\(a_i\sigma_n\) are positive constants or increase infinitely towards $\infty$.
Thus, we introduce a technique which tries to find out whether fixing the valuations of
some variables and increasing or decreasing the valuations of others results in positive
resp.\ increasing valuations of $a_1,\ldots,a_k$.
Our technique operates on \emph{limit problems} of the form
$\{a_1^{\bullet_1},\ldots,a_k^{\bullet_k}\}$ where $a_i$ is an arithmetic expression and
$\bullet_i \in \{+, -, +_!, -_!\}$ for all $1 \leq i \leq k$.
Here, \(a^+\) (resp.\ \(a^-\)) means that $a$ has to grow towards $\infty$
(resp.\ $-\infty$) and \(a^{+_!}\) (resp.\ \(a^{-_!}\)) means that $a$ has to be a
positive (resp.\ negative) constant.
So we represent $\guard(\alpha)$ by an \emph{initial limit problem} $\{
a_1^{\bullet_1},\ldots,a_k^{\bullet_k} \}$ where $\bullet_i \in \{ +, +_! \}$ for all $1
\leq i \leq k$.
By choosing $\bullet_i$ to be $+$ or $+_!$, $a_i^{\bullet_i}$ means that $a_i$ grows towards
$\omega$ or it is a positive constant, i.e., this ensures that $a_i > 0$ holds for large
enough $n$.
Our implementation leaves the choice of
the $\bullet_i \in \{+, +_! \}$ in the initial limit problem open as long as possible and
only fixes it when this is needed in order to \emph{solve} the limit problem (see the end
of  \Cref{subsec:its-transforming-limit-problems} for further details).
To solve a limit problem $L$, we search for a \emph{solution} $\sigma_n$ of $L$,
where this concept is defined in terms of \emph{limits} of functions.

\begin{definition}[Limit]
  \label{def:intro-limit}
  For
  each $f: \NN \to \RR$ we have $\lim_{n\mapsto \infty} f(n) = \infty$
  (resp.\  $\lim_{n\mapsto \infty} f(n) = -\infty$) if for every $m \in \RR$ there is an $n_0 \in \NN$
  such that $f(n) \geq m$ (resp.\ $f(n) \leq m$) holds for all $n \geq n_0$.
  Similarly, we have $\lim_{n\mapsto \infty} f(n) = m$ if for every $\varepsilon \in \RR$
  with $\varepsilon > 0$ there is an $n_0 \in \NN$ such that
  $|f(n) - m| < \varepsilon$ holds for all $n \geq n_0$.
\end{definition}
Now a family of substitutions $\sigma_n$ is a solution for a limit problem $\{
a_1^{\bullet_1},\ldots,a_k^{\bullet_k} \}$ if $\lim_{n \mapsto \infty} a_i\sigma_n$
complies with $\bullet_i$ for each $1 \leq i \leq k$.

\begin{definition}[Solutions of Limit Problems]
  \label{def:its-solutions-limit-problem}
  $\!$For any function $f\!: \NN \to \RR$ and any $\bullet\!\in\!\{+, -, +_!, -_!\}$, we say that
  $f$ \emph{satisfies} $\bullet$ if:
  \[
    \begin{array}{l@{\;\;\;}l@{\hspace{2cm}}l@{\;\;\;}l}
      \lim_{n \mapsto \infty} f(n) = \phantom{-}\infty, & \text{if } \bullet = {+} & \exists m \in \RR .\ \lim_{n \mapsto \infty} f(n) = m > 0,& \text{if } \bullet = {+_!}\\
      \lim_{n \mapsto \infty} f(n) = -\infty, & \text{if } \bullet = {-} & \exists m \in \RR .\ \lim_{n \mapsto \infty} f(n) = m < 0, & \text{if } \bullet =  {-_!}
    \end{array}
  \]
  A family \(\sigma_n\) of integer substitutions with $\VV(L) \subseteq \dom(\sigma_n)$ is a \emph{solution} of a limit problem
  \(L\) if for every $a^\bullet \in L$, the function $\lambda n .\ a\sigma_n$ satisfies
  $\bullet$. For any arithmetic expression $a$ with $\VV(a) \subseteq \{ n \}$, ``$\lambda n.\ a$'' is the function from $\NN \to \RR$ that maps
  any $n \in \NN$ to the value of $a$.
\end{definition}

\begin{example}[Solution of the Limit Problem of the Program \eqref{eq:finalFirstLeadingEx}]
  \label{ex:its-leading-ex-solution}
  The program \eqref{eq:finalFirstLeadingEx} has the guard $\frac{1}{2}
  x^2+\frac{1}{2} x > 1$ which normalizes to $\frac{1}{2}  x^2+\frac{1}{2} x - 1
  > 0$. Hence, the resulting initial limit problem could be $\{ (\frac{1}{2}  x^2+\frac{1}{2} x - 1)^{+} \}$. It
  is solved by the family of substitutions $\sigma_n$ from
  \eqref{eq:sigmanFirstLeadingEx}.
  The reason is that $\lim_{n \mapsto \infty} (\lambda n.\  (\frac{1}{2}  x^2+\frac{1}{2} x - 1)\, \sigma_n) = \infty$,
  i.e., the function $\lambda n.\  (\frac{1}{2}  x^2+\frac{1}{2} x - 1) \,\sigma_n$ satisfies $+$.
  Thus, there is an $n_0$ such that $\sigma_n \models
  \guard(\alpha_{0.\overline{1}.2.\overline{3.\overline{4}.5}})$
  holds for all $n \geq n_0$.
\end{example}

In \Cref{subsec:its-transforming-limit-problems,sec:asymptotic-smt} we will show how to infer such solutions of
limit problems automatically.
The following theorem clarifies how to deduce an asymptotic lower bound from a solution of
a limit problem.

\begin{theorem}[Asymptotic Bounds for Simplified Programs]
  \label{thm:its-asymptotic-bounds}
  Given a rule \(\alpha\) of a simplified program \(\PP\) with the program variables
  $\vect{x}$ and $\guard(\alpha) = (a_1
  \circ_1 0) \land \dots \land (a_k \circ_k 0)$ where $\circ_1,\ldots,\circ_k \in
  \{{>},{\geq}\}$, let the family $\sigma_n$ be a solution of an initial limit problem
  \(\{a_1^{\bullet_1},\dots,a_k^{\bullet_k}\}\) with $\bullet_1,\ldots,\bullet_k \in \{+,
  +_!\}$.
  Then $\rc_{\PP}(\size{\vect{x}\sigma_n}) \in \Omega(\cost(\alpha)\,\sigma_n)$.
\end{theorem}
\begin{proof}
  Since $\sigma_n$ is a solution of $\{a_1^{\bullet_1},\dots,a_k^{\bullet_k}\}$, there is
  an $n_0 \in \NN$ such that for all $n \geq n_0$, we have $\sigma_n \models  a_1 > 0
  \land \dots \land a_k > 0$, which implies $\sigma_n \models  \guard(\alpha)$.
  Hence, for all $n \geq n_0$, we obtain:
  \[
    \begin{array}{rcl@{\qquad}l}
      \rc_{\PP}(\size{\vect{x}\sigma_n}) &\geq & \dht_{\PP}(\lhs(\alpha)\,\sigma_n) \\
                                                                &\geq &
      \cost(\alpha)\,\sigma_n & \text{as $\sigma_n \models \guard(\alpha)$}
    \end{array}
  \]
  This implies $\rc_{\PP}(\size{\vect{x}\sigma_n}) \in \Omega(\cost(\alpha)\,\sigma_n)$.
\end{proof}
Of course, if $\PP$ has several rules, then we try to take the one which results in the
highest lower bound.

\begin{example}[Asymptotic Bound for  the Program \eqref{eq:finalFirstLeadingEx}]
  \label{ex:its-leading-ex-asymptotic}
  We continue \Cref{ex:its-leading-ex-solution} with the program $\PP = \{
  \eqref{eq:finalFirstLeadingEx} \}$. For $\vect{x} =(x,y,z,u)$,
according to \Cref{thm:its-asymptotic-bounds}, we get the asymptotic lower bound
\begin{equation}
  \label{ExBoundForProgram}   \rc_{\PP}(\size{\vect{x}\sigma_n}) \;\in\;
  \Omega(\cost(\eqref{eq:finalFirstLeadingEx})\,\sigma_n).
  \end{equation}
Note that $\cost(\eqref{eq:finalFirstLeadingEx})\,\sigma_n = \frac{1}{8}  n^4 +
\frac{1}{4}  n^3 + \frac{7}{8}n^2 + \frac{7}{4} n$. Hence,
\eqref{ExBoundForProgram} is equivalent to
\[ \rc_{\PP}(\size{\vect{x}\sigma_n}) \;\in\; \Omega(n^4).\]
\end{example}

Up to now, we only took the guard $\bigwedge_{i = 1}^k (a_i \, \circ_i \, 0)$ of a rule
$\alpha$ into account in the initial limit problem
$\{a_1^{\bullet_1},\ldots,a_k^{\bullet_k}\}$. This has the disadvantage that
solutions of this limit problem do not necessarily try to maximize the cost of the rule. For
example, for the rule
\[  \Ff_0(x, y) \tox[-1pt]{x \cdot y}{}{} \Ff(0, y) \constr{x > 0},
\]
we would obtain the initial limit problem $\{x^+\}$ which is solved by the family of
substitutions $\sigma_n = \{ x/n, y/0 \}$. According to \Cref{thm:its-asymptotic-bounds},
this only allows us to 
infer $\rc_\PP(n) \in \Omega(\cost(\alpha)\sigma_n)$, where $\cost(\alpha)\sigma_n
= 0$, i.e., it only allows us to infer a constant lower bound. To 
obtain non-trivial lower bounds instead,
one should extend the initial limit problem
\(\{a_1^{\bullet_1},\dots, a_k^{\bullet_k}\}\) of a rule $\alpha$ by $\cost(\alpha)^+$.
In this way, one searches for families of substitutions $\sigma_n$ where
$\cost(\alpha)\,\sigma_n$ grows towards $\omega$, i.e., where
$\cost(\alpha)\,\sigma_n$ depends on $n$ and is not constant.
So in our example, we should start with the initial limit problem
$\{x^+, (x\cdot y)^+\}$  which has the solution $\sigma_n = \{ x/n, y/n \}$. By
\Cref{thm:its-asymptotic-bounds},
one now obtains the quadratic lower bound $\rc_\PP(n) \in \Omega(n^2)$, since $\cost(\alpha)\sigma_n
= n^2$.

The costs are \emph{unbounded} (i.e., they 
depend on temporary variables) if the initial limit problem
\(\{a_1^{\bullet_1},\dots,a_k^{\bullet_k},\cost(\alpha)^+\}\) has a solution $\sigma_n$
where $x\sigma_n$ is constant for all program variables $x$.
Then we can even infer $\rc_{\PP}(n) \in \Omega(\omega)$.

\begin{example}[Unbounded Loops Continued]
  \label{ex:its-unbounded-asymptotic}
  By chaining
the initial rule
$\Ff_0(x,y) \tox{0}{}{} \Ff(x,y)$ of the program from \Cref{ex:its-unbounded} with
the accelerated rule
\setcounter{auxctr}{\value{equation}}
\setcounter{equation}{\value{eq:unboundedAcceleratedCtr}}
 \begin{equation}
     \Ff(x, y) \tox[-1pt]{\tv_1 \cdot y}{}{} \Ff(x + \tv_1, y) \constr{0 < x \land 0 < \tv_1}
      \end{equation}
from \Cref{ex:its-unbounded2}, we
  obtain
  \[
    \Ff_0(x, y) \tox[-1pt]{\tv_1 \cdot y}{}{} \Ff(x + \tv_1, y) \constr{0 < x \land 0 < \tv_1}.
  \]
  The resulting initial limit problem \(\{x^{+_!}, \tv_1^+,
  (\tv_1 \cdot y)^+\}\)  has the solution \(\sigma_n = \{ x/1, \, y/1, \, \tv_1/n \}\),
  which implies \(\rc_{\PP}(n) \in
  \Omega(\omega)\).
\end{example}
\setcounter{equation}{\value{auxctr}}

\Cref{thm:its-asymptotic-bounds} results in bounds ``$\rc_{\PP}(\size{\vect{x}\sigma_n})
\in \Omega(\cost(\alpha)\,\sigma_n)$'' which depend on the sizes
\(\size{\vect{x}\sigma_n}\).
Let \(f(n) = \rc_{\PP}(n)\), \(g(n) = \size{\vect{x}\sigma_n}\), and let
$\Omega(\cost(\alpha)\,\sigma_n)$ have the form $\Omega(n^k)$ or $\Omega(k^n)$ for some $k
\in \NN$.
Moreover for all \(x \in \vect{x}\), let \(x\sigma_n\) be a polynomial of at most degree
$d$, i.e., let \(g(n) \in \OO(n^d)\).
Then, based on an observation from \cite{jar17},\footnote{In the second case of \Cref{lem:its-to-irc}, we
fix a small inaccuracy from \cite{ijcar16} where we inadvertently wrote $f(n)
\in \Omega(k^{\sqrt[d]{n}})$. Since \Cref{lem:its-to-irc} is very similar to Lemma
  24 from our paper \cite{jar17}, we
  omit its proof here. The proof can be found in Appendix A.}
we can infer a lower bound for
$f(n) = \rc_{\PP}(n)$ instead of $f(g(n)) = \rc_{\PP}(\size{\vect{x}\sigma_n})$.
Moreover, if $g(n)= \size{\vect{x}\sigma_n}$ is constant whereas
$\Omega(\cost(\alpha)\,\sigma_n)$ is not constant, then the lemma allows us to infer that
\(f(n) = \rc_{\PP}(n) \in \Omega(\omega),\) as in \Cref{ex:its-unbounded-asymptotic}.

\setcounter{sectionctr}{\value{section}}
\setcounter{lemmactr}{\value{theorem}}

\begin{lemma}[Bounds for Function Composition]
  \label{lem:its-to-irc}
  Let
    $f: \NN \to \RR_{\geq 0} \cup \{ \omega \}$
    and $g: \NN \to \NN$
      where $g(n) \in \OO(n^d)$ for some $d \in \NN$ with $d > 0$.
  Moreover, let $f(n)$ be weakly monotonically increasing for large enough $n$.
  \begin{itemize}
  \item[(a)] If  $g(n)$ is strictly monotonically increasing for large enough $n$ and
    $f(g(n)) \in \Omega(n^k)$ with $k \in \NN$, then  $f(n) \in \Omega(n^{\frac{k}{d}})$.
  \item[(b)] If  $g(n)$ is strictly monotonically increasing for large enough $n$ and $f(g(n)) \in \Omega(k^n)$ with $k > 1$, then $f(n) \in \Omega(e^{\sqrt[d]{n}})$
    for some number $e \in \RR$ with $e > 1$.
    \item[(c)] If $g(n) \in \OO(1)$ and $f(g(n)) \notin \OO(1)$, then
    $f(n) \in \Omega(\omega)$.
  \end{itemize}
\end{lemma}
\begin{example}[Asymptotic Bound for  Program \eqref{eq:finalFirstLeadingEx} Continued]
  \label{ex:its-leading-asymptotic-bound}
  In \Cref{ex:its-leading-ex-asymptotic}, for $\vect{x} = (x, y, z, u)$, we
  inferred
  \(\rc_{\PP}(\size{\vect{x}\sigma_n}) \in \Omega(n^4)\) where \(x\sigma_n = n\) and \(y\sigma_n = z\sigma_n = u\sigma_n = 0\).
  Hence, we have $\size{\vect{x}\sigma_n} = |n| = n \in \OO(n^1)$.
  By \Cref{lem:its-to-irc}(a), we obtain \(\rc_{\PP}(n) \in \Omega(n^{\frac{4}{1}}) =
  \Omega(n^4)\) for the program $\PP = \{ \eqref{eq:finalFirstLeadingEx} \}$.
Due to the soundness of the processors for program simplification in \Cref{sec:simplification}, we also
  have  \(\rc_{\widetilde{\PP}}(n) \in
  \Omega(n^4)\) for the program $\widetilde{\PP}$ in \Cref{fig:its-leading-ex}.
\end{example}

In cases like \Cref{ex:its-sqrt}, \Cref{lem:its-to-irc} even allows us to infer sub-linear bounds.

\begin{example}[\Cref{ex:its-sqrt} Continued]
  \label{ex:its-sqrt-cont}
  Reconsider the program from \Cref{ex:its-sqrt}.
  By \Cref{def:its-solutions-limit-problem}, the family \(\sigma_n\) with \(x\sigma_n=n^2
  + 1\) and \(y\sigma_n = n\) is a solution of the initial limit problem
  \(\{(x-y^2)^{+_!}, y^+\}\). (Our implementation chooses $(x-y^2)^{+_!}$ instead of
  $(x-y^2)^{+}$ in the initial limit problem, because in this way, the limit problem can be
  solved by our technique, see \Cref{subsec:its-transforming-limit-problems}.)
  Due to \Cref{thm:its-asymptotic-bounds}, this proves
  \(\rc_{\PP}(\size{\vect{x}\sigma_n}) \in \Omega(n)\) for the program variables
  $\vect{x} = (x,y)$.
  As $\size{\vect{x}\sigma_n} = n^2 + 1 + n \in \OO(n^2)$, \Cref{lem:its-to-irc}(a) results
  in \(\rc_{\PP}(n) \in \Omega(n^{\frac{1}{2}}) = \Omega(\sqrt{n})\).

If the cost of the rule from \Cref{ex:its-sqrt} was $2^y$, then $\sigma_n$ would still be a
solution of the initial limit problem \(\{(x-y^2)^{+_!}, (2^y)^+\}\). So we
would obtain \(\rc_{\PP}(\size{\vect{x}\sigma_n}) \in \Omega(2^n)\) due to
\Cref{thm:its-asymptotic-bounds} and thus \(\rc_{\PP}(n) \in
\Omega(e^{\sqrt{n}})\) for an $e > 1$ due to
\Cref{lem:its-to-irc}(b). Intuitively, the exponent $\sqrt{n}$ expresses that the
cost grows exponentially w.r.t.\ $y$, where the guard $x > y^2$ implies $|y| \in
\OO(\sqrt{\size{\vect{x}}})$, i.e., $y$ is bounded by the square root of the
input size.

The reason why we cannot specify $e$ in \Cref{lem:its-to-irc}(b) is that it depends on the coefficients of
$g(n) = \size{\vect{x}\sigma_n}$, but \Cref{lem:its-to-irc} only requires $g(n) \in
\OO(n^d)$. Thus, a variant of \Cref{lem:its-to-irc} where the polynomial $g$ is
known would allow us to compute $e$. Our implementation simply reports that the
runtime is at least exponential if \Cref{lem:its-to-irc}(b) applies and
$d=1$.
\end{example}

\subsection{Transforming Limit Problems}
\label{subsec:its-transforming-limit-problems}

In order to use \Cref{thm:its-asymptotic-bounds} (and \Cref{lem:its-to-irc}) for the
automatic inference of lower bounds, we still have to show how to find a family of
substitutions $\sigma_n$ automatically that solves the initial limit problem of a
program's rule.

A limit problem $L$ is \emph{trivial} if all expressions in \(L\) are variables and there
is no variable \(x\) with \(x^{\bullet_1},x^{\bullet_2} \in L\) and \(\bullet_1 \neq
\bullet_2\).
For trivial limit problems $L$ we can immediately obtain a particular solution
$\sigma^{L}_n$ which instantiates variables ``according to $L$''.

\begin{lemma}[Solving Trivial Limit Problems]
  \label{Solving Trivial Limit Problems}
  Let  \(L\) be a trivial limit problem.
  Then  \(\sigma_n^{L}\) is a solution of $L$ where for all \(n \in \NN\),
  \(\sigma_n^{L}\) is defined as follows:
  \[
    x\sigma_n^{L} =
    \begin{cases}[r]
      n & \text{if } x^{+_{\phantom{!}}} \in L\\
      -n & \text{if } x^{-_{\phantom{!}}} \in L\\
      1 & \text{if } x^{+_!} \in L\\
      -1 &\text{if } x^{-_!} \in L\\
      0 &\text{otherwise}\\
    \end{cases}
  \]
\end{lemma}
\begin{proof}
  If $x^+ \in L$ (resp.\ $x^- \in L$), then $x\sigma_n^{L} = n$ (resp.\ $x\sigma_n^{L} =
  -n$) and thus, $\lim_{n \mapsto \infty} x\sigma_{n} = \lim_{n \mapsto \infty} n =
  \infty$ (resp.\ $\lim_{n \mapsto \infty} x\sigma_{n} = \lim_{n \mapsto \infty} -n =
  -\infty$), i.e., $\lambda n .\  x\sigma_n$ satisfies $+$ (resp.\ $-$).
  If $x^{+_!} \in L$ (resp.\ $x^{-_!} \in L$), then $x\sigma_n^{L} = 1$ (resp.\ $x\sigma_n^{L} = -1$).
  Thus, $\lim_{n \mapsto \infty} x\sigma_{n} = 1$ (resp.\ $\lim_{n \mapsto \infty}
  x\sigma_{n} = -1$), i.e., $\lambda n .\  x\sigma_n$ satisfies $+_!$ (resp.\ $-_!$).
  Hence, $\sigma_{n}^{L}$ is a solution of $L$.
\end{proof}

For instance, if \(\VV(\alpha) = \{x,y,\tv\}\) and \(L=\{x^+,y^{-_!}\}\), then $L$ is a
trivial limit problem and \(\sigma_n^L\) with \(x\sigma_n^L = n, y \sigma_n^L = -1\), and
\(\tv\,\sigma_n^L = 0\) is a solution for \(L\).

However, in general the initial limit problem $L =
\{a_1^{\bullet_1},\dots,a_k^{\bullet_k}, \cost(\alpha)^+\}$ is not trivial.
Therefore, we now define a transformation $\leadsto$ to simplify limit problems until one reaches a trivial problem.
With our transformation, $L \leadsto L'$ ensures that each solution of $L'$ also gives
rise to a solution of $L$.

If $L$ contains $f(a_1,a_2)^\bullet$ for some standard arithmetic operation $f$ like
addition, subtraction, multiplication, division, or exponentiation, we use a so-called
\emph{limit vector} \((\bullet_1,\bullet_2)\) with $\bullet_i\!\in\!\{+,-,+_!,-_!\}$ to
characterize for which kinds of arguments the operation $f$ is increasing (if $\bullet =
+$), decreasing (if $\bullet = -$), or a positive or negative constant (if
$\bullet = +_!$ or $\bullet = -_!$).\footnote{
  To ease the presentation, we restrict ourselves to binary operations $f$.
  For operations of arity \(k\), one would need limit vectors of the form $(\bullet_1, \ldots, \bullet_k)$.
}
Then $L$ can be transformed into the new limit problem \((L \setminus
\{f(a_1,a_2)^{\bullet}\}) \cup \{a_1^{\bullet_1},a_2^{\bullet_2}\}\).

For example, \((+,+_!)\) and $(+, -_!)$ are increasing limit vectors for subtraction.
The reason is that $a_1 - a_2$ is increasing if $a_1$ is increasing and $a_2$ is a constant.
Hence, our transformation $\leadsto$ allows us to replace $(a_1 - a_2)^+$ by $a_1^+$ and
$a_2^{+_!}$
or by  $a_1^+$ and
$a_2^{-_!}$.

\begin{definition}[Limit Vectors]
Let $f: \RR \to \RR$ be a function and let $\bullet_1,\bullet_2 \in \{+,-,+_!,-_!\}$.
We say that $(\bullet_1,\bullet_2)$ is an
\emph{increasing} (resp.\ \emph{decreasing}) \emph{limit vector} for $f$ if the function
$\lambda n .\  f(g(n),h(n))$ satisfies $+$ (resp.\ $-$) for any functions $g$ and $h$ that
satisfy $\bullet_1$ and $\bullet_2$, respectively.
Similarly, $(\bullet_1,\bullet_2)$ is a \emph{positive} (resp.\ \emph{negative})
\emph{limit vector} for $f$ if $\lambda n .\  f(g(n),h(n))$ satisfies $+_!$ (resp.\ $-_!$)
for any functions $g$ and $h$ that satisfy $\bullet_1$ and $\bullet_2$, respectively.
\end{definition}

With this definition, $(+,+_!)$ and $(+,-_!)$ are indeed an increasing limit vectors for subtraction,
since $\lim_{n \mapsto \infty} g(n) = \infty$ and $\lim_{n \mapsto \infty} h(n)
= m$ with $m > 0$ or $m < 0$ implies $\lim_{n \mapsto \infty} (g(n) - h(n)) = \infty$.
In other words, if $g(n)$ satisfies $+$ and $h(n)$ satisfies $+_!$ or  $-_!$, then $g(n) - h(n)$
satisfies $+$ as well.
In contrast, $(+,+)$ is not an increasing limit vector for subtraction.
To see this, consider the functions $g(n) = h(n) = n$. Both $g(n)$ and $h(n)$ satisfy $+$,
whereas $g(n) - h(n) = 0$ does not satisfy $+$.
Similarly, $(+_!,+_!)$ is not a positive limit vector for subtraction, since for $g(n) =
1$ and $h(n) = 2$, both $g(n)$ and $h(n)$ satisfy $+_!$, but $g(n) - h(n) = -1$ does not
satisfy $+_!$.

Limit vectors can be used to simplify limit problems,
as in  \ref{it:its-transform-limit-problem-A} in the following definition.
Moreover, for numbers $m \in \RR$, one can easily simplify constraints of the form
$m^{+_!}$ and $m^{-_!}$ (e.g., $2^{+_!}$ is obviously satisfied since $2 > 0$),
as in \ref{it:its-transform-limit-problem-B}.

\begin{definition}[\(\leadsto\)]
  \label{def:its-leadstoarrow}
  Let \(L\) be a limit problem. We have:
  \begin{enumerate}[series=leadsto,label=(\Alph*)]
  \item $L \cup \{ f(a_1,a_2)^\bullet \} \leadsto L \cup \{a_1^{\bullet_1},
    a_2^{\bullet_2} \}$ if $\bullet$ is $+$ (resp.\  $-,+_!,-_!$) and
    $(\bullet_1,\bullet_2)$ is an increasing (resp.\  decreasing, positive, negative)
    limit vector for \(f\)
    \label{it:its-transform-limit-problem-A}
  \item $L \cup \{m^{+_!}\} \leadsto L$ if $m \in \RR$ and $m > 0$, $L \cup \{m^{-_!}\} \leadsto L$ if $m \in \RR$ and $m < 0$
    \label{it:its-transform-limit-problem-B}
  \end{enumerate}
\end{definition}

However, transforming a limit problem with $\leadsto$ may also result in
\emph{contradictory} limit problems that contain $x^{\bullet_1}$ and $x^{\bullet_2}$ where
$\bullet_1 \neq \bullet_2$, as the following example illustrates.

\begin{example}[Contradictory Limit Problems -- \Cref{ex:its-sqrt-cont} Continued]
  \label{ex:its-leading-leadsto-derivation}
 The initial limit problem  $\{ (x - y^2)^{+_!}, y^+ \}$  from \Cref{ex:its-sqrt-cont}
 cannot be solved with the current transformation rules.
While $(+, +_!)$ and $(+, -_!)$ are \emph{increasing} limit vector  for subtraction,
the only \emph{positive} limit vector  for subtraction
 is \((+_!,-_!)\). Thus, by \ref{it:its-transform-limit-problem-A} one obtains $\{ (x -
 y^2)^{+_!}, y^+ \} \leadsto
 \{ x^{+_!}, (y^2)^{-_!}, y^+ \}$ which contains the unsolvable requirement $(y^2)^{-_!}$.

 As an alternative, in \Cref{ex:its-sqrt-cont} one could regard the
 initial limit problem  $\{ (x - y^2)^{+}, y^+ \}$ instead. However, here the
 transformation rules fail as well. We
  have
   \[
    \def\arraystretch{1.2}
    \begin{array}{l@{\;\;}l@{\;\;}l}
     &  \{ (x - y^2)^{+}, y^+ \}&\\
      \leadsto& \{ x^+, (y^2)^{+_!}, y^+\}&
      \text{by \ref{it:its-transform-limit-problem-A} with the increasing limit vector \((+,+_!)\) for subtraction}\\

                                                                           \leadsto& \{
      x^{+_!}, y^{+_!}, y^+\}&
\text{by \ref{it:its-transform-limit-problem-A} with the positive limit vector \((+_!,+_!)\) for multiplication}\\
    \end{array}
    \]
    However, the resulting problem  is contradictory, as it
contains both \(y^{+_!}\) and \(y^+\).
\end{example}

Recall that the guard of the rule from \Cref{ex:its-sqrt} implies $x \geq
y^2 + 1$. If
we substitute \(x\) with its lower bound \(y^2 + 1\) in
the beginning,
then we can reduce the initial limit problem
$\{ (x - y^2)^{+_!}, y^+ \}$
to a trivial one. Hence, we now extend
$\leadsto$ by allowing to apply substitutions.

\begin{definition}[\(\leadsto\) Continued]
  \label{def:its-leadstoarrow-substitutions}
  Let \(L\) be a limit problem and let \(\theta\) be a substitution such that \(x \notin
  \VV(x\theta)\) for all $x \in \dom(\theta)$ and $\theta \circ \sigma$ is an integer
  substitution for each integer substitution $\sigma$ whose domain includes all variables
  occurring in the range of $\theta$.
  Then we have:\footnotemark
  \begin{enumerate}[resume*=leadsto]
  \item $L \leadstox{\theta} L\theta$
    \label{it:its-transform-limit-problem-C}
  \end{enumerate}
\end{definition}
\footnotetext{
  The other rules for $\leadsto$ are implicitly labeled with the identical substitution $\emptyset$.
}

\begin{example}[Applying Substitutions to Limit Problems -- \Cref{ex:its-leading-leadsto-derivation} Continued]
  \label{ex:leading-leadsto-derivation2}
  For the initial limit problem $\{ (x - y^2)^{+_!}, y^+ \}$ from \Cref{ex:its-sqrt-cont},
  we now have
  \[
    \begin{array}{rll}
      \{ (x - y^2)^{+_!}, y^+ \} &\leadstox{\{x / y^2 + 1\}}& \{ 1^{+_!}, y^+ \}\\
                                 &\leadsto& \{ y^+\}\\
    \end{array}
  \]
  i.e., we obtain the trivial limit problem $\{y^+\}$.
  Note that, given an integer substitution $\sigma$ with $y \in \dom(\sigma)$, $\{x / y^2 + 1
  \} \circ \sigma$ is an integer substitution as well.
  By \Cref{Solving Trivial Limit Problems}, the family $\sigma_n^{\{y^+\}}$ with $y \sigma^{\{y^+\}}_n = n$ solves
the resulting trivial limit problem $\{y^+\}$. To obtain a solution for the initial
limit problem    $\{ (x - y^2)^{+_!}, y^+ \}$ one has to
take the substitution $\{x / y^2 + 1\}$ into account that was used in its
transformation. In this way, we get the solution $\sigma_n = \{x / y^2 + 1\} \circ
\sigma_n^{\{y^+\}}$ for the initial limit problem
where $x \, \sigma_n = n^2 + 1$ and $y \, \sigma_n = n$. Thus, we obtain \(\rc_{\PP}(n)
\in \Omega(n^{\frac{1}{2}}) = \Omega(\sqrt{n})\), as in \Cref{ex:its-sqrt-cont}.
\end{example}

Although \Cref{def:its-leadstoarrow-substitutions} requires that $\theta \circ \sigma$ is
an
integer substitution whenever $\sigma$ is an integer substitution,
it is also useful to handle limit problems which
contain expressions that do not evaluate to integer numbers.

\begin{example}[Non-Integer Metering Functions Continued]
  \label{ex:CombiningSubstitutions}
  By chaining\footnote{Note that we cannot instantiate $\tv$ with the metering function that was used to accelerate the loop from \Cref{ex:its-rational-metering}, as it does not map to the integers, i.e., the prerequisites of \Cref{thm:its-instantiation} are not satisfied.} the only initial rule of the program
in
  \Cref{ex:its-rational-metering} with
  the accelerated rule \eqref{eq:its-rational-metering}, we obtain
  \begin{equation}
\label{eq:its-rational-metering-chained}
    \begin{array}{llll}
    \Ff_0(x) &\tox{\tv}{}{}{} &\Ff(x - 2 \tv) &\constr{0 < \tv < \frac{1}{2} x + 1}.
\end{array}
    \end{equation}
  For the initial limit problem $\{\tv^+, (\frac{1}{2}  x +1 -\tv)^{+_!} \}$ we
  get
  \[
    \begin{array}{rll}
      \{\tv^+, (\frac{1}{2}  x + 1 -\tv)^{+_!} \}
                                                             &\leadstox{\{x/2 \, \tv-1\}}& \{\tv^+, \frac{1}{2}^{+_!}\}\\
                                                             &\leadsto& \{\tv^+\}
    \end{array}
  \]
  by
\ref{it:its-transform-limit-problem-C} and
  \ref{it:its-transform-limit-problem-B}. (Our implementation first leaves it open
  whether to choose $(\frac{1}{2}  x +1 -\tv)^{+}$ or $(\frac{1}{2}  x +1 -\tv)^{+_!}$,
  but when transforming the arithmetic expression to $\frac{1}{2}$ by
  \ref{it:its-transform-limit-problem-C}, it finds out that one should use $(\frac{1}{2}
  x +1 -\tv)^{+_!}$ in order to solve the limit problem. We describe the strategy used by
  our implementation and its heuristic to find suitable substitutions $\theta$ for the
  application of rule \ref{it:its-transform-limit-problem-C} at the end of this subsection.)
  By \Cref{Solving Trivial Limit Problems}, the family $\sigma_n^{\{\tv^+\}}$ with $\tv\,\sigma^{\{\tv^+\}}_n = n$
solves
the resulting trivial limit problem $\{\tv^+\}$. Again, to obtain a solution for the original initial
limit problem    $\{\tv^+, (\frac{1}{2}  x + 1 -\tv)^{+_!} \}$ one has to
take the substitution $\{x/2 \, \tv-1\}$ into account, resulting in $\sigma_n$ with $x\sigma_n = 2 \, n-1$ and $\tv\,\sigma_n = n$.
Thus, by \Cref{thm:its-asymptotic-bounds}
we have $\rc_{\{\eqref{eq:its-rational-metering-chained}\}}(\size{x \sigma_n}) =
\rc_{\{\eqref{eq:its-rational-metering-chained}\}}(2 \, n-1) \in
\Omega(\cost(\eqref{eq:its-rational-metering-chained})\,\sigma_n) = \Omega(\tv\,\sigma_n) =
\Omega(n)$. As $2 \, n - 1 \in \OO(n)$, \Cref{lem:its-to-irc}(a) implies
$\rc_{\{\eqref{eq:its-rational-metering-chained}\}}(n) \in \Omega(n)$. By the soundness of the
processors for program
simplification in \Cref{sec:simplification}, we also have $\rc_\PP(n) \in \Omega(n)$ for the original
program in  \Cref{ex:its-rational-metering}.
\end{example}

\Cref{def:its-leadstoarrow-substitutions} requires
that $\theta \circ \sigma$ is an integer substitution for every integer
substitution $\sigma$ whose domain includes all variables
occurring in the range of $\theta$. To check this side-condition automatically, one can again use \Cref{lem:zz-polys}:
If the range of $\theta$ consists of polynomials, then for every $x \in \dom(\theta)$
we only have to check if instantiating the polynomial $x\, \theta$ by finitely many suitable integers
again results in an integer. More precisely,
if $x \, \theta$ contains the variables $x_1,\ldots,x_k$ of degrees $d_1,\ldots,d_k$,
respectively, we check if
$x \, \theta$ maps all arguments from $\{0,\ldots,d_1+1\} \times \ldots \times \{0,\ldots,d_k+1\}$
to integers.

However, up to now we cannot prove that, e.g., a rule \(\alpha\) with \(\guard(\alpha) =
x^2 - x > 0\) and \(\cost(\alpha) = x\) has a linear lower bound, since \((+,+)\) is not
an increasing limit vector for subtraction.
To handle such cases, the following transformation rules allow us to neglect polynomial
sub-expressions if they are ``dominated'' by other polynomials of higher degree or by
exponential sub-expressions.

\begin{definition}[\(\leadsto\) Continued]\label{def:its-leadstoarrow-continued}
  Let \(L\) be a limit problem, let \(\pm \in \{+,-\}\), and let $a,b,c$ be univariate
  polynomials over the same variable.
  Then we have:
  \begin{enumerate}[resume*=leadsto]
  \item $L \cup \{(a \pm b)^{+} \} \leadsto L  \cup \{ a^+\}$
and
      $L \cup \{(a \pm b)^{-} \} \leadsto L  \cup \{ a^-\}$
    if the degree of $a$ is greater than the degree of $b$
    \label{it:its-transform-limit-problem-D}
  \item $L \cup \{(a^c \pm b)^+ \} \leadsto L \cup \{ (a-1)^+, c^+ \}$
and 
$L \cup \{(a^c \pm b)^+ \} \leadsto L \cup \{ (a-1)^+, c^{+_!} \}$
    \label{it:its-transform-limit-problem-E}
  \end{enumerate}
\end{definition}

Thus, $\{ (x^2 - x)^+ \} \leadsto \{ (x^2)^+ \} = \{ (x \cdot x)^+ \} \leadsto \{ x^+ \}$
by \ref{it:its-transform-limit-problem-D} and \ref{it:its-transform-limit-problem-A} with
the increasing limit vector $(+, +)$ for multiplication.
The intuition for \ref{it:its-transform-limit-problem-E} is that any exponential expression
$a^c$ dominates any polynomial expression $b$ provided that the base $a$ is greater than 1
and the exponent $c$ grows towards $\omega$.

\begin{example}[\Cref{ex:its-accelerating-fib} Continued]
  We continue \Cref{ex:its-accelerating-fib}, where  the Fibonacci program was
  simplified to the program consisting just of the
  rule
\setcounter{auxctr}{\value{equation}}
\setcounter{equation}{\value{eq:FibonacciAcceleratedCtr}}
 \begin{equation}
       \Ff_0(x) \tox{2^{\frac{1}{2}  x - 1} - 1}{}{} \emptyset \constr{x > 1}.
  \end{equation}
Here, we
  obtain the initial limit problem $\{(x-1)^+, (2^{\frac{1}{2}  x - 1} - 1)^+ \}$.
  We get:
  \[
    \begin{array}{l@{\;\;}l@{\;\;}l}
      &  \{(x-1)^+, (2^{\frac{1}{2}  x - 1} - 1)^+ \}\\[.1cm]
      \leadsto& \{x^+, (2^{\frac{1}{2}
    x - 1} - 1)^+\}
&\text{by \ref{it:its-transform-limit-problem-D}} \\[.1cm]
                                                       \leadsto& \{x^+,
      1^{+_!},(\frac{1}{2} x - 1)^+ \}&
      \text{by \ref{it:its-transform-limit-problem-E}}\\[.1cm]
   \leadsto& \{x^+,
      1^{+_!},(\frac{1}{2} x)^+ \}&
 \text{by \ref{it:its-transform-limit-problem-A} with the increasing limit vector $(+,
   +_!)$ for subtraction}\\[.1cm]
  \leadsto& \{x^+,
      1^{+_!},(\frac{1}{2})^{+_!} \}&
 \text{by \ref{it:its-transform-limit-problem-A} with the increasing limit vector $(+_!,
   +)$ for multiplication}\\[.1cm]   \leadsto^2& \{x^+\}&
 \text{by \ref{it:its-transform-limit-problem-B} (twice)}
    \end{array}
    \]
    By \Cref{Solving Trivial Limit Problems}, the family $\sigma_n^{\{x^+\}}$ with
$x\sigma_n^{\{x^+\}} = n$ solves
the resulting trivial limit problem $\{x^+\}$ and hence, it also solves the initial limit
problem of Rule \eqref{eq:FibonacciAccelerated}. Thus, \Cref{thm:its-asymptotic-bounds}
implies $\rc_{\{\eqref{eq:FibonacciAccelerated}\}}(\size{x \sigma_n}) =
\rc_{\{\eqref{eq:FibonacciAccelerated}\}}(n) \in
\Omega(\cost(\eqref{eq:FibonacciAccelerated})\,\sigma_n) = \Omega(2^{\frac{1}{2} n - 1} - 1)
= \Omega(2^{\frac{1}{2} n}) = \Omega(\sqrt{2}^n) \subset \Omega(1.4^n)$.  Due to the soundness of the
program
simplification in \Cref{sec:simplification} and \Cref{sec:non-linear}, this also implies that
the runtime complexity of the original Fibonacci program $\PP$ from \Cref{ex:its-fib}
is exponential, i.e.,  $\rc_\PP(n) \in \Omega(\sqrt{2}^n)$.
\end{example}
\setcounter{equation}{\value{auxctr}}

\begin{figure}[t]
  \fbox{
    \vspace*{-.4cm}
    
\hspace*{-.4cm}\begin{minipage}{13.5cm}
 \begin{enumerate}[series=leadsto,label=(\Alph*)]
  \item $L \cup \{ f(a_1,a_2)^\bullet \} \leadsto L \cup \{a_1^{\bullet_1},
    a_2^{\bullet_2} \}$ if $\bullet$ is $+$ (resp.\  $-,+_!,-_!$) and
    $(\bullet_1,\bullet_2)$ is an increasing (resp.\  decreasing, positive, negative)
    limit vector for \(f\)
  \item $L \cup \{m^{+_!}\} \leadsto L$ if $m \in \RR$ and $m > 0$, $\;L \cup \{m^{-_!}\}
    \leadsto L$ if $m \in \RR$ and $m < 0$
  \item $L \leadstox{\theta} L\theta$ if
 \(x \notin
  \VV(x\theta)\) for all $x \in \dom(\theta)$, and if $\sigma$ is an  integer substitution
  with $\VV(\mathrm{range}(\theta)) \subseteq \dom(\sigma)$, then  $\theta \circ \sigma$
  is also an integer
  substitution
 \item $L \cup \{(a \pm b)^{\bullet} \} \leadsto L  \cup \{ a^\bullet\}$ if ${\bullet}
    \in \{+,-\}$
    and $\mathrm{degree}(a) >\mathrm{degree}(b)$ for the univariate polynomials $a,b$
\item $L \cup \{(a^c \pm b)^+ \} \leadsto L \cup \{ (a-1)^\bullet, c^+ \}$ if ${\bullet}
  \in \{+,+_!\}$
  for the univariate polynomials $a,b,c$
 \end{enumerate}
 \end{minipage}}
  \caption{Definition of our Transformation for a Limit Problem $L$}
  \label{fig:leadsTo}
\end{figure}

Note that \ref{it:its-transform-limit-problem-E} can also be used to handle limit problems
like $(a^c)^+$ (by choosing $b = 0$). We summarize the full definition of our
transformation $\leadsto$ of limit problems in
\Cref{fig:leadsTo}.
\Cref{thm:its-leadsto-correct} states that our transformation \(\leadsto\) is indeed correct.
As already illustrated in \Cref{ex:leading-leadsto-derivation2,ex:CombiningSubstitutions}, when constructing the solution from the resulting
trivial limit problem, one has to take the substitutions into account which were used in
the derivation.

\begin{theorem}[Correctness of \(\leadsto\)]
  \label{thm:its-leadsto-correct}
  If $L \leadstox{\theta} L'$ and the family $\sigma_n$ is a solution of $L'$, then
  $\theta \circ \sigma_n$ is a solution of $L$.
\end{theorem}
\begin{proof}
  First assume that the step from $L$ to $L'$ was done by \Cref{def:its-leadstoarrow}
  \ref{it:its-transform-limit-problem-A}.
  Since $\sigma_n$ is a solution for $L'$, it is a solution for $a_1^{\bullet_1}$ and
  $a_2^{\bullet_2}$, where $(\bullet_1,\bullet_2)$ is an increasing (resp.\ decreasing,
  positive, or negative) limit vector for \(f\).
  As $\sigma_n$ is a solution for both $a_i^{\bullet_i}$, the function $\lambda n .\  a_i\sigma_n$ satisfies $\bullet_i$.
  By the definition of limit vectors, this implies that $\lambda n .\  f(a_1\sigma_n,
  a_2\sigma_n) = \lambda n .\  f(a_1,a_2)\, \sigma_n$ satisfies $\bullet$.
  Thus, $\sigma_n$ is a solution for $f(a_1,a_2)^\bullet$.

  If the step from  $L$ to $L'$ was done by \Cref{def:its-leadstoarrow}
  \ref{it:its-transform-limit-problem-B}, then every solution $\sigma_n$ for $L'$ is also
  a solution for $L$, since $m\,\sigma_n = m$ holds for any $m \in \RR$.

  If the step from $L$ to $L'$ was done by \Cref{def:its-leadstoarrow-substitutions}
  \ref{it:its-transform-limit-problem-C}, then let $\sigma_n$ be a solution for $L' = L \theta$.
  Then for every $(a\,\theta)^\bullet \in L \theta$, $\lambda n .\  a\,\theta\sigma_n$
  satisfies $\bullet$ and hence $\theta \circ \sigma_n$ is a solution for $a^\bullet$.
  Thus, $\theta \circ \sigma_n$ is a solution for $L$.

  If the step from  $L$ to $L'$ was done by \Cref{def:its-leadstoarrow-continued}
  \ref{it:its-transform-limit-problem-D}, then let  $\sigma_n$ be a solution for
  $a^\bullet$.
  Since the polynomial $a$ only contains a single variable (say, $x$), we must have
  $\lim_{n \mapsto \infty} x\sigma_n = \infty$ or $\lim_{n \mapsto \infty} x\sigma_n =
  -\infty$.
  W.l.o.g, let  $\lim_{n \mapsto \infty} x\sigma_n = \infty$ and $\bullet = +$ (the other
  cases work analogously).
  Then $\lim_{n \mapsto \infty} a\,\sigma_n = \infty$ implies $\lim_{x \mapsto \infty} a = \infty$.
  Since the degree of $a$ is greater than the degree of $b$, this means $\lim_{x \mapsto
    \infty} (a \pm b) = \infty$ and hence $\lim_{n \mapsto \infty} (a \pm b)\,\sigma_n =
  \infty$.

  For \Cref{def:its-leadstoarrow-continued} \ref{it:its-transform-limit-problem-E}, the
  proof is analogous.
  Here for large enough $n$, $a^c\sigma_n$ is an exponential function with a base $> 1$.
  Since $\sigma_n$ is a solution for $c^+$, we again have $\lim_{n \mapsto \infty}
  x\sigma_n = \infty$ or $\lim_{n \mapsto \infty} x\sigma_n = -\infty$.
  Thus  $a^c\sigma_n$ is an exponential function which grows faster than $b\,\sigma_n$ for $n \mapsto \infty$.
  Hence, we obtain $\lim_{n \mapsto \infty} (a^c \pm b)\,\sigma_n = \infty$. \qedhere
\end{proof}

So to find an asymptotic lower bound for the runtime of a simplified program with a rule
$\alpha$,
we start with an initial limit problem $L=\{ a_1^{\bullet_1}, \ldots, a_k^{\bullet_k},
\cost(\alpha)^+\}$ that represents $\guard(\alpha)$ and requires non-constant cost, and
transform $L$ with $\leadsto$ into a trivial limit problem $L'$, i.e., \(L
\leadstox{{}_{\theta_1}} \ldots \leadstox{{}_{\theta_m}} L'\).
As mentioned before,
for automation one should leave the $\bullet_i$  in the initial problem $L$ open, and
only instantiate them by a value from $\{+, +_!\}$ when this is needed to apply a
particular rule of the transformation $\leadsto$.
Then by \Cref{Solving Trivial Limit Problems} and \Cref{thm:its-leadsto-correct},
the resulting family  $\sigma^{L'}_n$  of substitutions gives rise to a solution $\sigma_n
= \theta_1 \circ  \ldots \circ \theta_m \circ \sigma^{L'}_n$ of $L$.
Thus by \Cref{thm:its-asymptotic-bounds}, we have $\rc_{\PP}(\size{\vect{x}\sigma_n}) \in
\Omega(\cost(\alpha)\, \sigma_n)$, which leads to a lower bound for \(\rc_{\PP}(n)\) with \Cref{lem:its-to-irc}.

Our implementation uses the following strategy to apply the rules from
\Cref{def:its-leadstoarrow}, \ref{def:its-leadstoarrow-substitutions}, and \ref{def:its-leadstoarrow-continued}
for the transformation $\leadsto$.
Initially, we reduce the number of variables by propagating bounds implied by the guard of
the rule $\alpha$. For example,
if
an arithmetic expression  \(a\) with $x \notin \VV(a)$ is a minimal upper or a maximal
lower bound on $x$ (i.e., $\guard(\alpha)$ implies $x \leq a$ but not $x \leq a-1$, or
$\guard(\alpha)$ implies $x \geq a$ but not $x \geq a+1$),
then we may apply the substitution \(\{x / a\}\) to the initial limit problem by the
rule \ref{it:its-transform-limit-problem-C}.
Thus, we can, e.g., simplify the limit problem from \Cref{ex:its-sqrt-cont} by instantiating \(x\)
with \(y^2 + 1\), see \Cref{ex:leading-leadsto-derivation2}.
In the same way, the substitution which is applied in
\Cref{ex:CombiningSubstitutions} can
be found automatically.
Afterwards, we use \ref{it:its-transform-limit-problem-B} and
\ref{it:its-transform-limit-problem-D} with highest and
\ref{it:its-transform-limit-problem-E} with second highest priority.
The third priority is trying to apply \ref{it:its-transform-limit-problem-A} to univariate
expressions (since processing univariate expressions helps to guide the search).
As fourth priority, we  apply \ref{it:its-transform-limit-problem-C} with a suitable substitution
$\{x/m\}$ if $x^{+_!}$ or $x^{-_!}$ occurs in the current limit problem.
Otherwise, we apply \ref{it:its-transform-limit-problem-A} to multivariate expressions.
Since $\leadsto$ is well founded and, except for \ref{it:its-transform-limit-problem-C},
finitely branching, one may also backtrack and explore alternative applications of
$\leadsto$.
In particular, we backtrack if we obtain a contradictory limit problem.
Moreover, if we obtain a trivial limit problem $L'$ where \(\cost(\alpha)\,\sigma_n^{L'}\) is a polynomial, but $\cost(\alpha)$ is a polynomial of higher degree or an exponential function, then we backtrack to search for other solutions which might lead to a higher lower bound.
However, our implementation can of course fail, since solvability of limit problems is undecidable (due to Hilbert's Tenth Problem).

\begin{example}[Solving the Limit Problem of the Program \eqref{eq:finalFirstLeadingEx}]
  For the program $\PP = \{ \eqref{eq:finalFirstLeadingEx} \}$ that results from the simplification of
  the program $\widetilde{\PP}$ in \Cref{fig:its-leading-ex}, we obtain the initial limit
  problem
  $\{ (\frac{1}{2}  x^2+\frac{1}{2} x - 1)^{+}, (\frac{1}{8}  x^4 + \frac{1}{4}  x^3 +
  \frac{7}{8}x^2 + \frac{7}{4} x)^+ \}$. Here,  we have:
  \[
    \begin{array}{lll@{\hspace*{.1cm}}l}
&&\multicolumn{2}{l}{\{ (\frac{1}{2}  x^2+\frac{1}{2} x - 1)^{+}, (\frac{1}{8}  x^4 + \frac{1}{4}  x^3 +
  \frac{7}{8}x^2 + \frac{7}{4} x)^+ \}}\\[0.3cm]
    &\leadstox{\ref{it:its-transform-limit-problem-D}}^2\hspace*{-.2cm}& \{ (\frac{1}{2}  x^2)^{+}, (\frac{1}{8}  x^4)^{+}  \} &
    \text{as the degree of } \frac{1}{2}  x^2 \text{ is greater than the degree of }
    \frac{1}{2} x - 1 \text{ and}\\
    &&& \text{the degree of } \frac{1}{8}  x^4 \text{ is greater than the degree of }
     \frac{1}{4}  x^3 +
  \frac{7}{8}x^2 + \frac{7}{4} x\\[0.2cm]
    &\leadstox{\ref{it:its-transform-limit-problem-A}}^2\hspace*{-.2cm}& \{ \frac{1}{2}^{+_!},
  (x^2)^{+}, \frac{1}{8}^{+_!},  (x^4)^{+} \} & \text{using the increasing limit vector }
  ({+}_!, {+}) \text{ for multiplication} \\[0.2cm]
    &\leadstox{\ref{it:its-transform-limit-problem-B}}^2\hspace*{-.2cm}& \{ (x^2)^{+},  (x^4)^{+} \} \\[0.2cm]
    &\leadstox{\ref{it:its-transform-limit-problem-A}}^2\hspace*{-.2cm}& \{ x^{+} \} & \text{using the increasing limit vector } ({+}, {+}) \text{ for multiplication}
    \end{array}
  \]
  The resulting trivial limit problem $\{ x^{+} \}$ gives rise to the solution
  \eqref{eq:sigmanFirstLeadingEx} and hence proves $\rc_\PP(|\vect{x}\sigma_n|) \in
  \Omega(n^4)$ for $\vect{x} =(x,y,z,u)$, see \Cref{ex:its-leading-ex-asymptotic}.
  As $|\vect{x}\sigma_n| \in \OO(n)$, as in \Cref{ex:its-leading-asymptotic-bound} we
  get $\rc_{\PP}(n) \in \Omega(n^4)$ and thus also
  $\rc_{\widetilde{\PP}}(n) \in \Omega(n^4)$.
\end{example}


\subsection{Solving Limit Problems via SMT}
\label{sec:asymptotic-smt}

While the calculus presented in \Cref{subsec:its-transforming-limit-problems} permits a
precise analysis of simplified programs, it can also be quite expensive in practice.
The reason is that the next $\leadsto$-step is rarely unique and thus backtracking is
often unavoidable in order to find a good lower bound.
We now show how limit problems can be encoded as conjunctions of polynomial inequations.
In many cases, this allows us to use SMT solvers to solve limit problems more efficiently.

Essentially, the idea is to search for a solution $\sigma_n$ (i.e., a suitable family of
substitutions) that instantiates each variable $x$ in the limit problem by a linear
polynomial $x \sigma_n = m_x \cdot n + k_x$. Here, we leave the integers $m_x$ and $k_x$
open (i.e., they are \emph{abstract coefficients}) and we use SMT solving to find an
instantiation of the abstract coefficients by integer numbers such that $\sigma_n$ becomes a
solution for the limit problem.

Thus, if $a$ is a polynomial arithmetic expression,
then $a  \sigma_n = a \; \{x / (m_x \cdot n + k_x) \mid x \in \VV(a)\}$ is a univariate
polynomial over $n$ with abstract coefficients.
If $a$ is of degree $d$, then $a \sigma_n$ can be rearranged to the form $a_{d}
\cdot n^d + \ldots + a_{0} \cdot n^0$
where the $a_i$ are arithmetic expressions over the abstract coefficients $\{m_x,k_x \mid
x \in \VV(a)\}$ that do not contain $n$.

\begin{example}[Encoding the Initial Limit Problem for Program  \eqref{eq:finalFirstLeadingEx}]
  \label{ex:initial-encoding1}
  Consider the initial limit problem
  \(
    \{ \left(\frac{1}{2} x^2 + \frac{1}{2} x - 1\right)^+\}
  \)
  for the program with the rule \eqref{eq:finalFirstLeadingEx}, see
 \Cref{ex:its-leading-ex-solution}.
  We use $\sigma_n$ with $x\sigma_n = m_x \cdot n + k_x$.
  Therefore, we obtain
\[  \begin{array}{rcl}
    \left(\frac{1}{2}  x^2+\frac{1}{2} x - 1\right) \, \sigma_n &=& \frac{1}{2} \cdot (m_x \cdot n + k_x)^2+\frac{1}{2} \cdot (m_x \cdot n + k_x) - 1\\
    &=& a_2 \cdot n^2 + a_1 \cdot n + a_0
\end{array}
\]
  where
  $a_2 = \frac{1}{2}  m_x^2, \,$
  $a_1 = m_x \cdot k_x + \frac{1}{2}  m_x, \,$ and
  $a_0 = \frac{1}{2} k_x^2 + \frac{1}{2} k_x - 1$.
\end{example}

Clearly, we have $\lim_{n \mapsto \infty} a \sigma_n = \infty$ (resp.\ $-\infty$) if and
only if $a_{i} > 0$ (resp.\ $a_{i} < 0$) for some $i > 0$ and $a_{j} = 0$ for all $i+1 \leq j
\leq d$.
Similarly, $\lim_{n \mapsto \infty} a \sigma_n$ is a positive (resp.\ negative) constant
if and only if $a_i = 0$ for all $1 \leq i \leq d$ and $a_0 > 0$ (resp.\ $a_0 < 0$).
This allows us to translate the solvability of a limit problem into the satisfiability
of an arithmetic formula.

\begin{definition}[SMT Encoding of Limit Problems]
  \label{def:its-limit-smt}
  Let $a$ be a polynomial arithmetic expression of degree $d$ and
let $\sigma_n$ instantiate each occurring variable $x$ by $m_x \cdot n + k_x$ where $m_x,
k_x$ are abstract coefficients. Let $a \sigma_n =
 a_{d} \cdot n^d + \ldots + a_{0} \cdot n^0$ where $a_0,\ldots,a_d$ do not contain $n$.
  We define
  \[
    \smt(a^\bullet) =
    \begin{cases}
      \bigvee_{i=1}^d \left(a_i > 0 \land \bigwedge_{j=i+1}^d a_j = 0\right) & \text{if } \bullet = + \\
      \bigvee_{i=1}^d \left(a_i < 0 \land \bigwedge_{j=i+1}^d a_j = 0\right) & \text{if } \bullet = - \\
      \bigwedge_{j=1}^d a_j = 0 \land a_0 > 0 & \text{if } \bullet = +_! \\
      \bigwedge_{j=1}^d a_j = 0 \land a_0 < 0 & \text{if } \bullet = -_! \\
    \end{cases}
  \]
  We lift $\smt$ to limit problems $L$ where $a$ is a polynomial for each $a^\bullet \in
  L$ by defining $\smt(L) = \bigwedge_{a^\bullet \in L} \smt(a^\bullet)$.

  Furthermore, given a polynomial cost $\cc$ of degree $d$ with
  $\cc \sigma_n =
  c_{d} \cdot n^d + \ldots + c_{0} \cdot n^0$ where $c_0,\ldots,c_d$ do not contain $n$,
   we define $\smt_{\cc,i}(L) = \smt(L) \land c_i > 0$ for each $1 \leq i \leq d$.
  Finally, for the program variables $\vect{x}$, we define $\smt_{\infty}(L) = \smt(L)
  \land \bigwedge_{x \in \vect{x}} m_x = 0$.
\end{definition}

To solve a limit problem $L$, it suffices to find a solution for $\smt(L)$, because then the
substitution that results from instantiating the abstract coefficients of $\sigma_n$
accordingly is a solution for the limit problem $L$.
However, to maximize the cost $\cc$, one should try to find a solution for
$\smt_{\infty}(L)$ or $\smt_{\cc,i}(L)$ where $i$ is as large as possible.
The reason is that a solution for $\smt_{\infty}(L)$ allows us to deduce unbounded costs
(provided that the initial limit problem contained $\cc^+$, i.e., the cost is
non-constant for each solution of $L$), as the corresponding substitution $\sigma_n$
maps all program variables $x$ to constants
$k_x$ that do not depend on $n$.
A solution for $\smt_{\cc,i}(S)$ allows us to prove a polynomial lower bound whose degree
is at least $i$ via \Cref{thm:its-asymptotic-bounds} and \Cref{lem:its-to-irc} (since
$|\vect{x} \sigma_n | \in \OO(n)$).

\begin{example}[Encoding the Initial Limit Problem for  Program
    \eqref{eq:finalFirstLeadingEx} Continued]
  \label{ex:initial-encoding2}
  We now show how to encode the initial limit problem
  \(
    \{ \left(\frac{1}{2}  x^2+\frac{1}{2}  x-1\right)^{+} \}
  \)
  from \Cref{ex:its-leading-ex-solution}.\footnote{
  For reasons of simplicity, we do not include the cost of the rule in the initial limit
  problem.
}
  Since $x \, \sigma_n = m_{x} \cdot n + k_{x}$, we have $(\frac{1}{2}  x^2+\frac{1}{2}  x-1) \, \sigma_n =
  a_2 \cdot n^2 + a_1 \cdot n + a_0$ with $a_2,a_1,a_0$ as in \Cref{ex:initial-encoding1}.
Thus,
  \[ \begin{array}{rcl}
    \smt\left(\left\{\left(\frac{1}{2}  x^2+\frac{1}{2} x-1\right)^{+}\right\}\right) &=&
    (a_2 > 0 \lor (a_1 > 0 \land  a_2 = 0)).
    \end{array}
  \]
  Now SMT solvers can easily find a solution like, e.g., $\{m_x / 1, k_x / 0\}$.
\end{example}

\begin{example}[Encoding the Initial Limit Problem for \Cref{ex:its-sqrt}]
  \label{ex:initial-encoding-sqrt}
  Next we encode the initial limit problem
  \(
  \{(x-y^2)^{+_!}, y^+\}
  \)
  for \Cref{ex:its-sqrt}.
  Since $x \, \sigma_n = m_{x} \cdot n + k_{x}$ and $y \, \sigma_n = m_y \cdot n + k_y$, we have $(x-y^2) \, \sigma_n =
  m_{x} \cdot n + k_{x} - (m_y \cdot n + k_y)^2 = -m_y^2 \cdot n^2 + (m_x - 2 \cdot m_y \cdot k_y) \cdot n + k_x - k_y^2$.
Thus,
\[
  \smt\left(\{(x-y^2)^{+_!}, y^+\}\right) = (-m_y^2 = 0 \land  m_x - 2 \cdot m_y \cdot k_y = 0 \land k_x - k_y^2 > 0 \land m_y > 0).
  \]
  Here, the first three (in)equations
  are the encoding of $(x-y^2)^{+_!}$ and the last inequation is the encoding of $y^+$.
  As $-m_y^2 = 0$ implies $m_y = 0$, the overall formula is unsatisfiable.
  State-of-the-art SMT solvers can prove unsatisfiability of $\smt\left(\{(x-y^2)^{+_!}, y^+\}\right)$ within milliseconds.
  This is not surprising, since we instantiated $x$ with a non-linear expression in
  \Cref{ex:leading-leadsto-derivation2} in order to find a solution, but
  \Cref{def:its-limit-smt} instantiates $x$ with a \emph{linear} template polynomial.
\end{example}

Thus, \Cref{ex:initial-encoding-sqrt} shows that even if all arithmetic expressions in the
analyzed limit problem are polynomials, $\leadsto$ is still required, i.e., our SMT-based
technique does not subsume the calculus of
\Cref{subsec:its-transforming-limit-problems}.\footnote{
  We could also use non-linear template polynomials for $\sigma_n$, but in any case the
  degree of the template polynomials has
  to be fixed in advance and thus, it may be insufficient for the problem at hand.
}
Note that the new SMT-based technique can be integrated into the calculus from
\Cref{subsec:its-transforming-limit-problems} seamlessly. In other words, one can first
simplify a limit problem with the transformation $\leadsto$ for a few steps and then apply
the SMT-based technique to find a solution for the obtained limit problem. For
instance, the initial limit problem
$\{(x-y^2)^{+_!}, y^+\}$ in
\Cref{ex:initial-encoding-sqrt}
can easily be solved via the SMT encoding from \Cref{def:its-limit-smt} after applying the
substitution $\{x / y^2 + 1\}$ as in \Cref{ex:leading-leadsto-derivation2}.

The following theorem shows how a solution for the SMT problem $\smt(L)$ can be used to
obtain a solution for the limit problem $L$.

\begin{theorem}[Solving Limit Problems via SMT]
  \label{thm:its-smt-limit}
  Let $L$ be a limit problem such that each expression in $L$ is a polynomial and let
  $\sigma$ be an integer substitution such that $\sigma \models \smt(L)$.
  Then
  \[
    \sigma_n \circ \sigma = \{x / (m_x\sigma \cdot n + k_x\sigma) \mid x \in \VV(L)\}
  \]
  is a solution for $L$.
\end{theorem}
\begin{proof}
  First note that $\sigma_n\circ \sigma$ is clearly an integer substitution for each $n \in \NN$.
  We have to show that $\lambda n .\  a\sigma_n\sigma$ satisfies $\bullet$ for any  $a^\bullet \in L$.
  Let $d$ be the degree of $a$.
  Then we have $a \sigma_n = a_{d} \cdot n^d + \ldots + a_{0} \cdot n^0$ for suitable
  expressions $a_0, \ldots, a_d$ over $\{ m_x, k_x \mid x \in \VV(a) \}$ that do not
  contain $n$.

  We first consider the case $\bullet = {+}$ (the case $\bullet = {-}$ works analogously).
  Then  $\sigma \models \smt(L)$ implies $\sigma \models \smt(a^+)$, i.e.,
  $\sigma \models \bigvee_{i=1}^d (a_i > 0 \land \bigwedge_{j=i+1}^d a_j = 0)$.
  Hence, there exists an $1 \leq i \leq d$ with $a_i\sigma > 0$ and
  we have
  \[
    a \sigma_n\sigma = a_{i}\sigma \cdot n^i + \ldots + a_{0}\sigma \cdot n^0.
  \]
  Thus, we obtain $\lim_{n \mapsto \infty} a \sigma_n\sigma = \infty$,
  i.e., $\lambda n .\  a\sigma_n\sigma$ satisfies $+$.

  Now we consider the case $\bullet = {+_!}$ (the case $\bullet = {-_!}$ works analogously).
  Then $\sigma \models \smt(L)$ implies $\sigma \models \smt(a^{+_!})$, i.e.,
  $\sigma \models \bigwedge_{j=1}^d a_j = 0 \land a_0 > 0$.
  Hence, we have
  \[
    \lim_{n \mapsto \infty} a_n\sigma_n\sigma = \lim_{n \mapsto \infty}a_0\sigma = a_{0}\sigma > 0,
  \]
  i.e., $\lambda n .\  a\sigma_n\sigma$ satisfies $+_!$.
\end{proof}

\begin{example}[Solving  the Initial Limit Problem for  Program
    \eqref{eq:finalFirstLeadingEx}]\label{ex:smt-finalFirstLeadingEx}
  In \Cref{ex:initial-encoding2}, we saw that
  $\sigma = \{m_x / 1, k_x / 0\} \models \smt(\{(\frac{1}{2}  x^2+\frac{1}{2}  x - 1)^+\})$.
  Hence, according to \Cref{thm:its-smt-limit}, $\sigma_n \circ \sigma = \{ x / (m_x \cdot
  n + k_x) \} \circ \{ m_x / 1, k_x / 0\}$ with $x \sigma_n \sigma = n$ solves the limit problem
  $\{(\frac{1}{2}  x^2+\frac{1}{2}  x - 1)^+\}$.
  This corresponds to the solution in \eqref{eq:sigmanFirstLeadingEx}.
  As explained in \Cref{ex:its-leading-asymptotic-bound}, this proves \(\rc_{\PP}(n) \in
  \Omega(n^4)\) for the program $\PP = \{ \eqref{eq:finalFirstLeadingEx} \}$ and hence,
  also for the program in \Cref{fig:its-leading-ex}.
\end{example}

To integrate \Cref{thm:its-smt-limit} into the calculus of
\Cref{subsec:its-transforming-limit-problems}, we proceed as follows.
Whenever the current limit problem $L$ for a rule $\alpha$ only contains polynomial
arithmetic expressions, one tries to find a solution for $\smt_\infty(L)$ or
$\smt_{\cc,i}(L)$ where $i$ is initially set to the degree of the cost $\cc$ of the rule
$\alpha$ and decremented until the SMT solver finds a solution.
As soon as a solution is found, one can either return the resulting family of
substitutions that solves $L$ or keep searching for a better solution.
To this end, one can either backtrack or continue simplifying $L$ via $\leadsto$.
Similarly, if the SMT solver does not find a solution one can either backtrack or keep
simplifying $L$ via $\leadsto$.

Note that the intention of \Cref{thm:its-smt-limit} and its integration into $\leadsto$ is
not to add more power to $\leadsto$.
Instead,
as often as possible one should
delegate the search for a solution to SMT solvers (which are very efficient in solving
search problems) instead of relying on heuristics.
This may, of course, also lead to better results in cases where the heuristics
used for $\leadsto$ are not ideal. For example, consider a simplified
rule with cost $x \cdot (1 - \tv^2)$ and guard $-2 < \tv < 2$. The heuristic discussed at
the end of
\Cref{subsec:its-transforming-limit-problems} would instantiate $\tv$ with the bounds
implied by the guard, i.e., it would apply the substitution $\theta = \{\tv /
1\}$ or $\theta = \{\tv / -1\}$, resulting in the unsolvable limit problem $\{(x
\cdot (1 - \tv^2))^+\}\,\theta = \{0^+\}$. In contrast, our SMT encoding allows us
to find the solution $\{\tv / 0, x / n\}$ which results in a linear lower bound.

\begin{example}[\Cref{ex:its-non-linear-chaining-facsum} Continued]
  For the simplified $\fs{facSum}$ program
with the rule
  \eqref{eq:FacSumSimplified}, we obtain the initial limit problem
  \[ \begin{array}{l}
   \{ (x-1)^{+}, \left(\frac{1}{2}x^2 + \frac{3}{2}x - 2\right)^+\}.
    \end{array}
  \]
  Using $\sigma_n = \{ x/(m_x \cdot n + k_x) \}$,
its SMT encoding is
  \[
  m_x > 0 \land
\left( \left(m_x \cdot k_x + \tfrac{3}{2} m_x > 0 \land \tfrac{1}{2} m_x^2 = 0 \right)
\lor \tfrac{1}{2} m_x^2 > 0 \right).
  \]
 An SMT solver can find a model like $\sigma = \{m_x / 1, \, k_x / 0\}$, for example.
  This results in the solution
  \[
    \sigma_n \circ \sigma = \{x / (m_x\sigma \cdot n + k_x\sigma)\} = \{ x / n\}.
  \]
  Applying it to the cost $\frac{1}{2}x^2 + \frac{3}{2}x - 2$ yields
  $\frac{1}{2}n^2 + \frac{3}{2}n - 2 \in \Omega(n^2)$.
By \Cref{thm:its-asymptotic-bounds}, this proves $\rc_{\{\eqref{eq:FacSumSimplified}\}} (\size{x
\sigma_n \sigma}) = \rc_{\{\eqref{eq:FacSumSimplified}\}} (n) \in \Omega(n^2)$. By the
soundness of the program simplification in Sections \ref{sec:simplification} and \ref{sec:non-linear}, we obtain
   $\rc_{\PP}(n) \in \Omega(n^2)$ for the original integer program
$\PP$ from \Cref{ex:its-fac}.
\end{example}


\section{Experiments}
\label{sec:experiments}

To evaluate the performance of our approach, we  implemented it in the tool \loat
(``\underline{Lo}op \underline{A}c\-cel\-er\-a\-tion \underline{T}ool'')\footnote{Initially,
  \loat stood for ``Lower Bounds Analysis Tool'', but we renamed it to reflect that
  \loat's loop acceleration techniques can be used for several purposes, see \cite{nonterm-acceleration}.}
using the recurrence solver \tool{PURRS}~\cite{purrs} and the SMT solver
\tool{Z3}~\cite{z3}.
We evaluated \loat\ on the 689 benchmark integer programs from the evaluation
\cite{koat-benchmarks} of the tool
\koat~\cite{koat} which infers \emph{upper} runtime bounds for integer programs.
On average, the ITSs in the collection \cite{koat-benchmarks}  have a size of 
23.4 rules.
In addition, the results of running \tool{LoAT} on the examples from this paper can be found at \cite{website}.

\paragraph{Comparison with Upper Bound Provers}
As we are not aware of any other tool to compute worst-case lower bounds for integer
programs, we compared our results with the asymptotically smallest results of the
tools \koat, \tool{CoFloCo}~\cite{cofloco1,cofloco2},
\tool{Loopus}~\cite{loopus-jar-17}, and \tool{RanK}~\cite{rank},
that compute upper runtime bounds for integer programs.

\newcommand{\Om}{\Omega}
\begin{table}
  \begin{adjustbox}{width=\textwidth,center}
  \begin{tabular}{|l||l|c|c|c|c|c|c|c|c|}
    \cline{2-10}
    \multicolumn{1}{c}{} & \multicolumn{9}{|c|}{\loat}\\
    \hhline{-=========}
    \multirow{9}{*}{\rotatebox[origin=c]{90}{\pbox[c]{100em}{\centering Best Upper Bound}}}
 & $\rc_{\PP}(n)$  & $\Om(1)$ & $\Om(n)$ & $\Om(n^2)$ & $\Om(n^3)$ & $\Om(n^4)$ & $\Om(n^5)$ & $EXP$ & $\Om(\omega)$ \\\cline{2-10}
 & $\OO(1)$        &    $135$ &      $-$ &        $-$ &        $-$ &        $-$ &        $-$ &   $-$ &           $-$ \\\cline{2-10}
 & $\OO(n)$        &     $41$ &    $146$ &        $-$ &        $-$ &        $-$ &        $-$ &   $-$ &           $-$ \\\cline{2-10}
 & $\OO(n^2)$      &      $7$ &     $14$ &       $54$ &        $-$ &        $-$ &        $-$ &   $-$ &           $-$ \\\cline{2-10}
 & $\OO(n^3)$      &      $1$ &      $1$ &        $-$ &        $9$ &        $-$ &        $-$ &   $-$ &           $-$ \\\cline{2-10}
 & $\OO(n^4)$      &      $-$ &      $-$ &        $-$ &        $-$ &        $2$ &        $-$ &   $-$ &           $-$ \\\cline{2-10}
 & $\OO(n^5)$      &      $-$ &      $-$ &        $-$ &        $1$ &        $-$ &        $-$ &   $-$ &           $-$ \\\cline{2-10}
 & $EXP$           &      $-$ &      $-$ &        $-$ &        $-$ &        $-$ &        $-$ &  $13$ &           $-$ \\\cline{2-10}
 & $INF$           &     $32$ &     $18$ &        $2$ &        $-$ &        $-$ &        $-$ &   $-$ &         $213$ \\\hline
  \end{tabular}
  \end{adjustbox}
 \caption{\label{tab:eval-results}Best Upper Bound vs.~\loat}
\end{table}

The results of running all tools on an \tool{Azure F4s v2}
instance with a timeout of 60 seconds per example
are shown in \Cref{tab:eval-results}.
\loat inferred
non-constant lower bounds for 473
examples.
For
  135 additional examples, the upper bound
$\rc_{\PP}(n) \in \OO(1)$ was proved and thus, the lower bound $\Om(1)$ inferred by \tool{LoAT} is
optimal. Thus,
\tool{LoAT} finds non-trivial or optimal bounds on 473 + 135 = 608 (88.2~\%)
  of all examples.
For 572 examples (83.0~\%), the inferred bounds
are asymptotically tight (e.g., lower and upper bounds
coincide). Whenever
an exponential upper bound was proved,
\tool{LoAT} also proved an exponential lower
bound (i.e., $\rc_{\PP}(n) \in \Om(k^n)$ for some $k > 1$).
It proved the existence of executions with unbounded length in 213 cases (this includes both
non-terminating
examples
and
examples whose runtime depends on temporary
variables).
On average, \loat required 2 seconds per example.
For 10 of the 689 examples, \loat could not
finish its analysis due to the timeout. The reason for most
timeouts is that the number of rules gets too large during the
program simplification and hence, \loat fails to compute a simplified program in time.

One reason why the bounds inferred by \loat do not always coincide with the upper bounds obtained by other tools may of
course be that these upper bounds are not necessarily tight. If the lower bound is too small, then in our experience the
most common reasons for imprecision are the following: For some loops, \loat fails to find (non-trivial)
metering functions. But if \loat finds a metering function, then the precision of the metering function is usually
very good. Another reason for imprecision are unsuitable instantiations of temporary variables. Finally, \loat
heuristically deletes rules from the analyzed program if the number of rules becomes too large, which is another common
source of imprecision. Nevertheless, our experiments show that optimal lower bounds are inferred for the large majority
of examples.

\paragraph{Evaluating Individual Contributions and Comparison with Best-Case Lower Bounds}
In a second experiment, we performed an individual evaluation of the main contributions of
this paper that are new compared to our earlier paper \cite{ijcar16}.
As baseline, we took \tool{LoAT-Basic}, the version of \loat implementing the techniques
in \cite{ijcar16}.\footnote{As we refactored large parts of the code and improved some
heuristics
 since publishing that paper, \tool{LoAT-Basic} is already more powerful than the version
 of \loat from 2016 that we used in
 \cite{ijcar16}.}
Using \tool{LoAT-Basic} as a starting point, we considered the following four variants:
\begin{itemize}
 \item \tool{LoAT-Cond} adds support for conditional metering functions, as introduced in \Cref{lem:irrelevant-constraints}.
 \item \tool{LoAT-Rec} is \tool{LoAT-Basic} extended by the handling of non-tail-recursive
   programs (described in \Cref{sec:non-linear}).
 \item \tool{LoAT-SMT} is like \tool{LoAT-Basic}, but it
   applies the SMT encoding of limit problems from \Cref{sec:asymptotic-smt} in addition
   to
     the calculus for the transformation of limit problems that is used in
     \tool{LoAT-Basic}. For this combination, we use a strategy which first simplifies
     limit problems with the calculus if the guard or the cost of the analyzed simplified rule contains
     non-polynomial arithmetic. However, the SMT encoding is
     used whenever it is applicable.
Whenever there was a choice during the application of the calculus, we backtrack
afterwards and apply the calculus again in order to examine the remaining possibilities.
    The analysis of the rule terminates as soon as \tool{LoAT} proves a bound which is asymptotically
    equal to its cost function, when the timeout specified by the user expires, or when there are no further
    possibilities to backtrack (where we use suitable heuristics to ensure that case
    \ref{it:its-transform-limit-problem-C} of our calculus is only applied with finitely many substitutions).
    Then the largest bound found so far is returned as the result.
   \item \tool{LoAT-JUST-SMT} is like \tool{LoAT-SMT}, but
     in contrast to \tool{LoAT-SMT}, \tool{LoAT-JUST-SMT} never
   applies the calculus again once the SMT encoding is applicable.
 \end{itemize}
Note that each of these variants only
adds one single new contribution to \tool{LoAT-Basic}, whereas the other new
contributions are disabled.
The intention of the last variant is to compare the power and performance of the
calculus from \Cref{subsec:its-transforming-limit-problems} with the SMT encoding
of \Cref{sec:asymptotic-smt} (whereas \tool{LoAT-SMT} represents the
\emph{combination} of both techniques). However, as the novel SMT encoding only
applies to polynomial limit problems, \tool{LoAT-JUST-SMT} still uses the
calculus from \Cref{subsec:its-transforming-limit-problems} for non-polynomial
limit problems.

 \begin{table}
  \begin{adjustbox}{width=\textwidth,center}
  \begin{tabular}{|l||r|r|r|r|r|r|r|r|r|r|}
    \hline
 Tool                 & $\Om(1)$ &  $\Om(n)$ & $\Om(n^2)$ & $\Om(n^3)$ & $\Om(n^4)$ & $EXP$ & $\Om(\omega)$ & time (s) \\\hhline{=========}
 \tool{LoAT-Basic}    &    $269$ &     $159$ &       $42$ &        $2$ &        $2$ &   $5$ &         $210$ & $2.29$   \\\hline
 \tool{LoAT-Cond}     &    $268$ &     $159$ &       $43$ &        $3$ &        $2$ &   $5$ &         $209$ & $2.32$   \\\hline
 \tool{LoAT-Rec}      &    $227$ &     $175$ &       $55$ &        $6$ &        $2$ &  $13$ &         $211$ & $2.44$   \\\hline
 \tool{LoAT-SMT}      &    $263$ &     $161$ &       $43$ &        $2$ &        $2$ &   $5$ &         $213$ & $1.68$   \\\hline
 \tool{LoAT-Just-SMT} &    $263$ &     $161$ &       $43$ &        $2$ &        $2$ &   $5$ &         $213$ & $1.64$   \\\hline
 \tool{LoAT}          &    $216$ &     $179$ &       $56$ &       $10$ &        $2$ &
 $13$ &         $213$ & $1.91$   \\\hhline{=========}
 \tool{CoFloCo}          &    $454$ &     $189$ &       $40$ &       $6$ &        $0$ &
 $0$ &         $0$ & $2.00$   \\\hline
  \end{tabular}
  \end{adjustbox}
 \caption{\label{tab:eval-ablation-results}Evaluating Individual Contributions in
   \tool{LoAT} and Comparison with Best-Case Lower Bounds}
\end{table}

The results of our experiments are summarized in
\Cref{tab:eval-ablation-results}: \tool{LoAT-Basic} already uses an
 optimization from \cite{ijcar16} which
  is similar to (but weaker than) the conditional metering functions
 of \Cref{sec:metering}. Conditional metering functions do not only improve the formalization and presentation of our approach
  (by integrating the optimization of \cite{ijcar16}  into our concept of metering functions), but they
 also
lead to a minimal change in power: The detailed experimental results on our website
\cite{website} show that \tool{LoAT-Cond} deduces
 better
  (i.e., larger) asymptotic bounds in
 eight cases, whereas \tool{LoAT-Basic} deduces better asymptotic bounds in four cases.
Adding support for arbitrary recursion allows \loat to infer
non-constant bounds for 41 of the
 50 non-tail-recursive examples in the collection (where a constant upper bound was proved
 for three of these examples). The
 full version of \loat even
obtains non-constant bounds for 44 of these
examples.
The SMT encoding of limit problems improves the performance compared to \tool{LoAT-Basic}
by 26\%. Moreover, \tool{LoAT-SMT} deduces a better asymptotic bound than \tool{LoAT-Basic} in eight
cases, whereas \tool{LoAT-Basic} infers a better asymptotic bound in one case.
  The configurations
\tool{LoAT-SMT} and \tool{LoAT-Just-SMT} yield identical results. So when disregarding limit problems with non-polynomial arithmetic,
the calculus
from \Cref{subsec:its-transforming-limit-problems} is outperformed by the novel SMT
encoding from
 \Cref{sec:asymptotic-smt}
in our experiments on the
examples from \cite{koat-benchmarks}. 
However, this is not true in general, as shown by \Cref{ex:initial-encoding-sqrt}.
  Moreover, the calculus from \Cref{subsec:its-transforming-limit-problems} is still required for the analysis of limit problems
with non-polynomial arithmetic, i.e., for all examples where \tool{LoAT} proves exponential
lower bounds.
In \tool{LoAT}(i.e., in the last line of \Cref{tab:eval-ablation-results}), we extended
\tool{LoAT-Basic} by all
new contributions. This results in a significant improvement in both
power and runtime.

Finally, in the last row of \Cref{tab:eval-ablation-results}
we used the tool \tool{CoFloCo} \cite{cofloco1,cofloco2} to compute \emph{best-case} lower
bounds.
While the asymptotic bounds obtained from \tool{LoAT} and \tool{CoFloCo} coincide for 335 
of the 689 examples,
 the results
of the two tools are of course not directly comparable,
since \tool{LoAT} infers \emph{worst-case} lower bounds,
but in general, a worst-case lower bound is not a valid best-case lower bound.
Moreover, \tool{CoFloCo} analyzes different
program paths (so-called \emph{chains})
separately and infers individual lower bounds for them.
However, similar to the experimental evaluation of \tool{CoFloCo} in \cite{cofloco2}, 
the results in the last row of  \Cref{tab:eval-ablation-results} do not take the preconditions of the 
different chains into account, but 
they are simply the maximum
lower bound of all chains.
Thus, they are not always asymptotic best-case lower bounds
for the whole program. So the purpose of the last row is only to indicate that
\tool{LoAT}'s performance
is also convincing when comparing it to the performance of other tools for the inference
of (other forms of) lower bounds.

For a
detailed experimental evaluation of our implementation as well as a pre-compiled binary of
\tool{LoAT} we refer to \cite{website}. The source code of
\tool{LoAT} is freely available \mbox{at \cite{loat}}.

\section{Related Work}
\label{sec:related}

While there are many techniques to infer \emph{upper bounds} on the worst-case
complexity of integer programs (e.g.,
\cite{%
  costa-asymptotic,costa-complexity,
  MaxCore,
  abc,%
  speed-popl-09,
  JLAMP18,
  campy,%
  rank,%
  purrs,
  pubs,
  koat,%
  cofloco1,cofloco2,%
  loopus-jar-17,%
  c4b,%
  pastis,%
  alonso12,%
  pubs-upper-lower,%
  raml,ramlnat,ramlpopl17}),
there is little work on \emph{lower bounds}.  In \cite{alonso12}, it is briefly
mentioned that their technique could also be adapted to infer lower instead of
upper bounds for \emph{abstract cost rules}, i.e., integer procedures with
(possibly multiple) outputs.  However, this only considers \emph{best-case}
lower bounds instead of worst-case lower bounds as in our technique.  Upper and
lower bounds for \emph{cost relations} are inferred in
\cite{pubs-upper-lower,cofloco2}. Cost relations extend recurrence equations such that,
e.g., non-determinism can be modeled. However, this technique also considers
best-case lower bounds only.  A method for best-case lower bounds for logic
programs is presented in \cite{prolog-lower}.

Note that techniques to infer lower bounds on the best-case complexity differ
fundamentally from techniques for the inference of worst-case lower bounds.
To deduce best-case lower bounds, one has to prove that a certain bound holds for
\emph{every} program run.
Thus, as in the case of worst-case upper bounds, \emph{over-approximating} techniques are
used to ensure that the proven bound covers all program runs, i.e., even though such
techniques under-approximate the runtime of the program, they over-approximate the set of
all program runs.

In contrast, techniques to infer lower bounds on the worst-case complexity have to
identify families of inputs (i.e., witnesses) that result in expensive program runs.
Thus, for the inference of worst-case lower bounds, over-approximations are usually
unsound, since one has to ensure that the witness of the proven lower bound corresponds to
``real'' program runs.
Thus, \emph{under-approximating} techniques have to be used in order to infer lower bounds
on the worst-case complexity.

Nevertheless, our approach has certain aspects in common with the
technique in \cite{pubs-upper-lower}, since \cite{pubs-upper-lower} also uses recurrence
solving to compute a closed 
form for the costs of several consecutive applications of a cost equation with direct
recursion, which corresponds to a simple loop or simple recursion in our setting.
However, as mentioned above, the analyses for best-case lower bounds
from \cite{pubs-upper-lower,cofloco2} have to reason about
all program runs.
Thus, there the handling of non-determinism is challenging as all possible
non-deterministic choices have to be taken into account.
In contrast, we can treat temporary variables (which we use to model non-determinism)
as constants when computing the iterated update and cost.
Thus, our iterated update and cost only represent evaluations where temporary variables
are instantiated with the same values in each iteration.
This restriction is sound in our setting, as we only need to prove the existence of a certain family of program runs, i.e., we do not have to reason about all program runs.
To reason about evaluations where the valuation of the temporary variables changes, we can instantiate them with expressions containing program variables via \emph{Instantiation} (\Cref{thm:its-instantiation}).

Since our computation of the iterated update relies on the existence of a single
deterministic update, it is not applicable to simple recursions, which also prevents us
from computing iterated costs when accelerating non-tail-recursive rules in
\Cref{thm:its-non-linear-acceleration}. Thus, our handling of simple recursions may be
improved by incorporating ideas from \cite{pubs-upper-lower,cofloco2} for the inference of bounds
of cost  equations with multiple recursive calls.

In \cite{jar17}, we introduced two techniques to infer worst-case lower
bounds for term rewrite systems (TRSs).  However, TRSs differ substantially from
the programs considered here, since they do not allow integers and have no
notion of a ``program start''.  Thus, the techniques from \cite{jar17} are
very different to the present paper.

In contrast to the techniques for the computation of symbolic runtime bounds, \cite{wise}
presents a technique to generate test-cases that trigger the worst-case execution time of
programs.
The idea is to execute the program for small inputs, observe the required runtime, and
then generalize those inputs that lead to expensive runs.
In this way, one obtains \emph{generators} which can be used to construct larger inputs
that presumably result in expensive runs as well.
In contrast to the technique presented in the current paper, \cite{wise} operates on
\pl{Java}, i.e., it also supports data structures.
However, \cite{wise} does not try to infer symbolic bounds, which is the main purpose of
our technique.
Nevertheless, ideas from \cite{wise} could be integrated into our framework.
For example, a similar approach could be used in order to apply \emph{Instantiation}
(\Cref{thm:its-instantiation}) in a way that leads to expensive runs.

The approach from \cite{raml-worst-case} can synthesize worst-case inputs, but the size of
the input needs to be fixed a priori. This approach is 
fundamentally different from
our technique to deduce \emph{symbolic} worst-case bounds. However, the inputs that are
synthesized by the technique from \cite{raml-worst-case} are provably optimal, whereas our
worst-case lower bounds are correct, but not necessarily tight.

Inferring bounds on the runtime of programs has also been investigated for probabilistic
programs. While there exist several approaches to find upper bounds on the expected
runtime of such programs, again there are only very few works that
consider the inference of lower bounds on the expected
runtime of probabilistic programs \cite{McIverMorgan05,FuChatterjee19,AimingLow,CADE19}.

To simplify programs, we use \emph{Loop Acceleration} to summarize the effect
of applying a simple loop (or a simple recursion) repeatedly.
Acceleration is mostly used in over-approximating settings (e.g., \cite{kincaid15,
  gonnord06, jeannet14, madhukar15, strejcek12}), where handling non-determinism is challenging, as
loop summaries have to cover \emph{all} possible non-deterministic choices.
However, our technique is under-approximating, i.e., we can instantiate non-deterministic
values arbitrarily.

The under-approximating acceleration technique in \cite{underapprox15} uses
quantifier elimination, whereas our acceleration technique relies on metering
functions.
In \cite{nonterm-acceleration}, we generalized the loop acceleration technique from \cite{underapprox15}
in order to prove non-termination of integer programs.
In future work, we will examine whether ideas from \cite{underapprox15,nonterm-acceleration} can
also be
incorporated
into our framework for the inference of lower runtime bounds.

Another related approach
(see, e.g., \cite{bozga09a,bozga10,iosif17})
accelerates loops whose transitive closure can be expressed in Presburger Arithmetic.
In particular, this is the case for loops whose transition relation can be described by \emph{octagons}, i.e.,
conjunctions of inequations of the form $\pm x \pm y \leq c$ where $x,y$ are variables and $c \in \ZZ$, and for loops with
Presburger-definable guards and affine updates $\mu(\vect{x}) = A\cdot \vect{x} + \vect{c}$ with the \emph{finite monoid
  property}, i.e., where the set $\{A^i \mid i \in \NN\}$ is finite.
In contrast, our acceleration technique does not necessarily compute the transitive closure of loops exactly.
The reason is that our metering functions may be imprecise and that we approximate non-determinism by assuming that
the values of temporary variables remain unchanged across loop iterations.
On the other hand, our approach can also handle loops where the transitive closure cannot be expressed in Presburger
Arithmetic.

The paper \cite{costa-asymptotic} presents a technique to infer asymptotic bounds from concrete bounds with a so-called context constraint $\phi$, i.e., bounds of the form $\phi \implies \mathit{rt} \leq e$ or $\phi \implies \mathit{rt} \geq e$.
Here, $\mathit{rt}$ is the runtime of the program and $e$ is a \emph{cost expression}.
These expressions are orthogonal to the expressions that are supported by our
technique from \Cref{sec:asymptotic} (e.g., $e$ may contain maximum and
logarithm, but negative numbers are only allowed in sub-expressions of the form $\max(0,\ldots)$).
The approach in \cite{costa-asymptotic} infers multi-variate asymptotic bounds, whereas
our technique infers univariate bounds which are only parameterized in the size of the
input.
Moreover, \cite{costa-asymptotic} does not aim to eliminate the context constraint, i.e.,
the resulting asymptotic bounds are of the form $\phi \implies \mathit{rt} \in \OO(b)$ or
$\phi \implies \mathit{rt} \in \Omega(b)$.
In contrast, eliminating the context constraints $\phi$ is one of the main motivations for our
technique to deduce asymptotic bounds from concrete bounds.

Our SMT encoding for limit problems from \Cref{sec:asymptotic-smt} inspired parts of the work from \cite{cav19},
where we used a similar encoding to prove that termination is decidable for a certain class of integer loops.
As in \Cref{sec:asymptotic-smt}, the underlying idea of \cite{cav19} is to abstract the loop guard by focusing on its behavior
``for large enough values of $n$''. In the present work, the variable $n$ is the parameter
of the family of substitutions that solves the
limit problem. In \cite{cav19}, $n$ represents the number of loop iterations. Due to the restricted form of the
loops in \cite{cav19}, their SMT encoding only requires \emph{linear} integer arithmetic, such that we obtain a
decision procedure for termination.

Finally, in \cite{metering-invariants}, our concept of metering functions has been adapted in order to
synthesize invariants. Note that invariant inference and complexity analysis are closely related:
Invariant inference techniques can be used to compute complexity bounds by introducing an additional
counter that is incremented in each step and deducing an invariant that bounds its value (see, e.g.,
\cite{loopus-jar-17}). Conversely, complexity analysis techniques can be used to bound the value of any
arithmetic expression $b$ (i.e., to compute invariants) by choosing the costs of transitions in a way
that reflects changes of the value of $b$ (see, e.g., \cite{frocos17, loopus-jar-17}).


\section{Conclusion and Future Work}
\label{sec:conclusion}

We introduced the first technique to infer lower bounds on the worst-case runtime
complexity of integer programs, based on a modular program simplification framework.
The main simplification techniques are \emph{Loop Acceleration} and \emph{Recursion Acceleration},
which rely on \emph{recurrence solving} and \emph{metering functions}, an adaptation of
classical ranking functions.
By eliminating loops and function symbols via \emph{Chaining} and \emph{Partial Deletion},
we eventually obtain \emph{simplified programs}.
We presented a technique to infer \emph{asymptotic lower bounds} from simplified programs,
which can also be used to find program vulnerabilities.
An experimental evaluation with our tool \tool{LoAT} demonstrates the applicability of
our technique in practice, see \cite{loat,website}.

In comparison to the preliminary version of our paper from \cite{ijcar16}, we
showed how to deduce \emph{conditional} metering functions,
we improved our program simplification by eliminating variables from the program,
we extended
our approach to non-tail-recursive programs, and we
improved our technique to infer asymptotic lower bounds for simplified programs
by an SMT encoding. See \Cref{sec:intro} for a
full list of the contributions of the current paper compared to
\cite{ijcar16}.

There are several interesting directions for future work.
First of all,  one could couple \loat with
invariant inference techniques to improve its power.
Furthermore, \loat's heuristics to apply \emph{Instantiation} are relatively simple and
should be improved, e.g., by incorporating ideas from \cite{wise}.
Another interesting question is to what extent \loat can benefit from more sophisticated
techniques to infer metering functions.
Possibilities include the inference of logarithmic or super-linear polynomial metering
functions, but one could also adapt the \emph{quasi-ranking functions} from
\cite{larraz13} to our setting.
Apart from that, we plan to investigate if our approach can benefit from
alternative loop acceleration techniques
\cite{underapprox15,nonterm-acceleration,bozga10}.
Moreover, as mentioned in \Cref{sec:related}, ideas from \cite{pubs-upper-lower,cofloco2}
could be adapted to under-approximate the costs of repeatedly applying simple
recursions more precisely when accelerating them.
Finally, one could generalize our program model to improve its expressiveness.
In particular, one could consider the return values of auxiliary function calls (by
allowing terms with nested occurrences of functions). Moreover, one could
combine our technique with ideas from \cite{jar17} for the inference of lower
bounds for term rewrite systems (i.e., programs operating on tree-shaped data
structures) to analyze programs whose complexity depends on both integers and
data structures.

\begin{acks}
  We thank Samir Genaim, Jan B\"oker, and Jera Hensel for
  discussions and comments. We are also very grateful to the anonymous reviewers for many
  helpful suggestions.

  This work is funded by the 
  \grantsponsor{DFG}{Deutsche Forschungsgemeinschaft (DFG, German Research
    Foundation)}{https://www.dfg.de/} -
  \grantnum{DFG}{389792660}
   as part of
   \grantnum[https://perspicuous-computing.science]{DFG}{TRR~248} and
 by the 
  \grantsponsor{DFG}{Deutsche Forschungsgemeinschaft (DFG, German Research
    Foundation)}{https://www.dfg.de/} - \grantnum{DFG}{235950644} (Project 
   GI 274/6-2).   
\end{acks}

\section*{A \quad Proof of Lemma \ref{lem:its-to-irc}}
\label{appendix}

\Cref{lem:its-to-irc} is based on Lemma
 24 from our paper \cite{jar17}. However, since the lemmas are slightly different and
 since the proof for part (b) was omitted from \cite{jar17}, we provide the proof of
 \Cref{lem:its-to-irc} in this appendix. Moreover, part (c) of the lemma was not present
 in \cite{jar17}.

 \setcounter{section}{\value{sectionctr}}
 \setcounter{theorem}{\value{lemmactr}}

\begin{lemma}[Bounds for Function Composition]
  Let
    $f: \NN \to \RR_{\geq 0} \cup \{ \omega \}$
    and $g: \NN \to \NN$
      where $g(n) \in \OO(n^d)$ for some $d \in \NN$ with $d > 0$.
  Moreover, let $f(n)$ be weakly monotonically increasing for large enough $n$.
  \begin{itemize}
  \item[(a)] If  $g(n)$ is strictly monotonically increasing for large enough $n$ and
    $f(g(n)) \in \Omega(n^k)$ with $k \in \NN$, then  $f(n) \in \Omega(n^{\frac{k}{d}})$.
  \item[(b)] If  $g(n)$ is strictly monotonically increasing for large enough $n$ and $f(g(n)) \in \Omega(k^n)$ with $k > 1$, then $f(n) \in \Omega(e^{\sqrt[d]{n}})$
    for some number $e \in \RR$ with $e > 1$.
    \item[(c)] If $g(n) \in \OO(1)$ and $f(g(n)) \notin \OO(1)$, then
    $f(n) \in \Omega(\omega)$.
  \end{itemize}
\end{lemma}

  \begin{proof}
For any $n_0 \in \NN$, let  $\NN_{\geq n_0} = \{n \in
\NN \mid n \geq n_0\}$.
  For any (total) function $h: M \to \NN_{\geq n_0}$ with $M \subseteq \NN$ where $M$ is
  infinite, we define $\lfloor h \rfloor(n): \NN_{\geq \min(M)} \to \NN_{\geq n_0}$  and
  $\lceil h \rceil(n): \NN \to \NN_{\geq n_0}$ by:
  \begin{eqnarray*}
    \lfloor h \rfloor(n) &=& h(\max\{x \in M \mid x \leq n\})\\
    \lceil h \rceil(n) &=& h(\min\{x \in M \mid x \geq n\})
  \end{eqnarray*}
  Note that infinity of $h$'s domain $M$ ensures that there is always an $x \in M$ with $x
  \geq n$.
Since  $\lfloor h \rfloor$ is only defined on $\NN_{\geq \min(M)}$, there is always
an $x \in M$ with $x \leq n$ for any $n \in  \NN_{\geq \min(M)}$.

To prove the lemma, as in the proof of \cite[Lemma 24]{jar17}
we first show that if $h: M \to \NN_{\geq n_0}$ is strictly monotonically increasing and
surjective, then
  \begin{equation}
    \label{claim}
    \lfloor h \rfloor(n) \in \{\lceil h \rceil(n), \lceil h \rceil(n) - 1\} \quad
    \text{for all } n \in \NN_{\geq \min(M)}
  \end{equation}
 Afterwards, we prove (a) -- (c) separately.
  \begin{claims}
    \claim{\eqref{claim} holds, i.e.,
      $\lfloor h \rfloor(n) \in \{\lceil h \rceil(n), \lceil h \rceil(n) - 1\}$}
    To prove \eqref{claim}, let $n \in \NN_{\geq \min(M)}$.
    If $n \in M$, then clearly $\lfloor h\rfloor(n) = \lceil h \rceil(n)$.
    If $n \notin M$, then let $\check{n} = \max\{x \in M \mid x < n\}$ and $\hat{n} = \min\{x \in M \mid x > n\}$.
    Thus, $\check{n} < n < \hat{n}$.
    Strict monotonicity of $h$ implies $h(\check{n}) < h(\hat{n})$.
    Assume that $h(\hat{n}) - h (\check{n}) > 1$.
    Then by surjectivity of $h$, there is an $\overline{n} \in M$ with $h(\overline{n}) = h(\check{n}) + 1$ and thus $h(\check{n}) < h(\overline{n}) < h(\hat{n})$.
    By strict monotonicity of $h$, we obtain $\check{n} < \overline{n} < \hat{n}$.
    Since \(n \notin M\) and \(\overline{n} \in M\) implies \(n \neq \overline{n}\), we either have $\overline{n} < n$ which contradicts $\check{n} = \max\{\check{n} \in M \mid \check{n} < n\}$ or $\overline{n} > n$ which contradicts  $\hat{n} = \min\{\hat{n} \in M \mid \hat{n} > n\}$.
    Hence, $\lfloor h \rfloor(n) = h(\check{n}) = h(\hat{n}) - 1 = \lceil h \rceil(n) -
    1$, which proves \eqref{claim}. \medskip

    \claim{\Cref{lem:its-to-irc}(a) holds, i.e.,
      $f(g(n)) \in \Omega(n^k)$ implies $f(n) \in \Omega(n^{\frac{k}{d}})$}
    For the proof of this claim, we slightly adapt the corresponding proof of
 \cite[Lemma 24]{jar17} to arbitrary functions $f$ and $g$.
    Note that $g(n) \in \OO(n^d)$ and $f(g(n)) \in \Omega(n^k)$ imply
    \[
      \exists n_0, m,m' > 0 .\; \forall n \in \NN_{\geq n_0} .\;\; g(n) \leq m \cdot n^d \land m' \cdot n^k \leq f(g(n)).
    \]
    We can choose $n_0$ large enough such that $f|_{\NN_{\geq n_0}}$ is weakly and
    $g|_{\NN_{\geq n_0}}$ is strictly monotonically increasing, where for any function $h:
    \NN \to \NN$ and any $M \subseteq \NN$, $h|_M$ denotes the restriction of $h$ to $M$.
    Let $M = \{ g(n) \mid n \geq n_0 \}$ and let $g^{-1}: M \to \NN_{\geq n_0}$ be the function such that $g(g^{-1}(n)) = n$.
    Note that $g^{-1}$ exists, since strict monotonicity of $g$ implies injectivity of $g$.
    By instantiating $n$ with $g^{-1}(n)$, we obtain
    \[
      \exists n_0, m,m' > 0 .\; \forall n \in M .\;\;
      g(g^{-1}(n)) \leq m \cdot (g^{-1}(n))^d \land m' \cdot (g^{-1}(n))^k \leq f(g(g^{-1}(n)))
    \]
    which simplifies to
    \[
      \exists n_0, m,m' > 0 .\; \forall n \in M .\;\; n \leq m \cdot (g^{-1}(n))^d \land m' \cdot (g^{-1}(n))^k \leq f(n).
    \]
    When dividing by $m$ and taking the $d$-th
    root on both sides of the first inequation, we get
    \[\begin{array}{l}
      \exists n_0, m, m' > 0 .\; \forall n \in M .\;\; \sqrt[d]{\frac{n}{m}} \leq g^{-1}(n)
      \land m' \cdot (g^{-1}(n))^k \leq f(n).
      \end{array}
    \]
    By monotonicity of $\sqrt[d]{\frac{n}{m}}$ and $f(n)$ in $n$, this implies
    \[\begin{array}{l}
      \exists n_0, m, m' > 0 .\; \forall n \in \NN_{\geq g(n_0)} .\;\; \sqrt[d]{\frac{n}{m}} \leq \lceil g^{-1}\rceil(n) \land m' \cdot (\lfloor g^{-1} \rfloor(n))^k \leq f(n).
   \end{array}   \]
    Note that $g|_{\NN_{\geq n_0}}$ is total and hence, $g^{-1}: M \to \NN_{\geq n_0}$ is surjective.
    Moreover, by strict monotonicity of $g|_{\NN_{\geq n_0}}$, $M$ is infinite and $g^{-1}$ is also strictly monotonically increasing.
    Hence, by \eqref{claim} we get $ \lceil g^{-1}\rceil(n) \leq \lfloor g^{-1} \rfloor(n) + 1$ for all $n \in \NN_{\geq g(n_0)}$.
    Thus,
    \[\begin{array}{l}
      \exists n_0, m, m' > 0 .\; \forall n \in \NN_{\geq g(n_0)} .\;\; \sqrt[d]{\frac{n}{m}} - 1 \leq \lfloor g^{-1}\rfloor(n) \land m' \cdot (\lfloor g^{-1} \rfloor(n))^k \leq f(n)
     \end{array} \]
    which implies
    \[\begin{array}{l}
      \exists n_0, m, m' >0 .\; \forall n \in \NN_{\geq g(n_0)} .\;\; m' \cdot \left( \sqrt[d]{\frac{n}{m}} -1 \right)^k \leq f(n).
      \end{array} \]
    Therefore, \(\exists m>0 .\;\; f(n) \in \Omega\left(\left(\sqrt[d]{\frac{n}{m}} -
    1\right)^k\right)\) and thus, \(f(n) \in \Omega\left(n^{\frac{k}{d}}\right)\). \medskip

    \claim{\Cref{lem:its-to-irc}(b) holds, i.e.,
      $f(g(n)) \in \Omega(k^n)$ implies $f(n) \in \Omega(e^{\sqrt[d]{n}})$
for some $e > 1$}
    The proof is analogous to the proof of the case $f(g(n)) \in \Omega(n^k)$, but it was
    not given in \cite{jar17}.
    Here, $g(n) \in \OO(n^d)$ and $f(g(n)) \in \Omega(k^n)$ imply
    \[
    \exists n_0, m,m' > 0 .\; \forall n \in \NN_{\geq n_0} .\;\;
    g(n) \leq m \cdot n^d \land m' \cdot k^n \leq f(g(n)).
    \]
    Again, we can choose $n_0$ large enough such that $f|_{\NN_{\geq n_0}}$ is weakly and $g|_{\NN_{\geq n_0}}$ is strictly monotonically increasing.
    As in the proof of the previous claim, let $M = \{ g(n) \mid n \geq n_0 \}$ and let $g^{-1}: M \to \NN_{\geq n_0}$ be the function such that $g(g^{-1}(n)) = n$.
    By instantiating $n$ with $g^{-1}(n)$, we obtain
    \[
      \exists n_0, m,m' > 0 .\; \forall n \in M .\;\;
      g(g^{-1}(n)) \leq m \cdot (g^{-1}(n))^d \land m' \cdot k^{g^{-1}(n)} \leq f(g(g^{-1}(n)))
    \]
    which simplifies to
    \[
      \exists n_0, m,m' > 0 .\; \forall n \in M .\;\; n \leq m \cdot (g^{-1}(n))^d \land
      m' \cdot k^{g^{-1}(n)} \leq f(n).
    \]
    When dividing by $m$ and taking the $d$-th root on both sides of the first inequation, we get
    \[ \begin{array}{l}
      \exists n_0, m, m' > 0 .\; \forall n \in M .\;\; \sqrt[d]{\frac{n}{m}} \leq
      g^{-1}(n) \land m' \cdot k^{g^{-1}(n)} \leq f(n).
      \end{array}
    \]
    By monotonicity of $\sqrt[d]{\frac{n}{m}}$ and $f(n)$ in $n$, this implies
    \[ \begin{array}{l}
    \exists n_0, m, m' > 0 .\; \forall n \in \NN_{\geq g(n_0)} .\;\;
    \sqrt[d]{\frac{n}{m}} \leq \lceil g^{-1}\rceil(n) \land m' \cdot k^{\lfloor g^{-1}
      \rfloor(n)} \leq f(n).    \end{array}
    \]
    As in the proof of Claim 2,
  $g|_{\NN_{\geq n_0}}$ is total and hence, $g^{-1}: M \to \NN_{\geq n_0}$ is surjective.
    Moreover, by strict monotonicity of $g|_{\NN_{\geq n_0}}$, $M$ is infinite and $g^{-1}$ is also strictly monotonically increasing.
    Hence, by \eqref{claim} we get $ \lceil g^{-1}\rceil(n) \leq \lfloor g^{-1} \rfloor(n) + 1$ for all $n \in \NN_{\geq g(n_0)}$.
    Thus,
    \[\begin{array}{l}
    \exists n_0, m, m' > 0 .\; \forall n \in \NN_{\geq g(n_0)} .\;\;
    \sqrt[d]{\frac{n}{m}} - 1 \leq \lfloor g^{-1}\rfloor(n) \land m' \cdot k^{\lfloor g^{-1} \rfloor(n)} \leq f(n).  \end{array}
    \]
   Since $k > 1$, this implies
  \[      \exists n_0, m, m' >0 .\; \forall n \in \NN_{\geq g(n_0)} .\;\; m' \cdot
  k^{\sqrt[d]{\frac{n}{m}} -1 } \leq f(n)
  \]
  which is equivalent to
  \[\begin{array}{l} \exists n_0, m, m' >0 .\; \forall n \in \NN_{\geq g(n_0)} .\;\; \frac{m'}{k} \cdot
  k^{\sqrt[d]{\frac{n}{m}}} \leq f(n)  \end{array}\]
      Since $m' > 0$ and $k > 1$, we have $m'' = \frac{m'}{k} > 0$ and thus we get
    \[
       \exists n_0, m, m'' >0 .\; \forall n \in \NN_{\geq g(n_0)} .\;\; m'' \cdot
    k^{\sqrt[d]{\frac{n}{m}}} \leq f(n)\]
    which is equivalent to
    \[ \exists n_0, m, m'' >0 .\; \forall n \in \NN_{\geq g(n_0)} .\;\; m'' \cdot
    k^{\frac{\sqrt[d]{n}}{\sqrt[d]{m}}} \leq f(n)\]
      Since $m > 0$, we have $r = \sqrt[d]{m} > 0$ and hence:
    \[
       \exists n_0, r, m'' >0 .\; \forall n \in \NN_{\geq g(n_0)} .\;\; m'' \cdot
       k^{\frac{\sqrt[d]{n}}{r}} \leq f(n)
       \]
       which is equivalent to
      \[ \exists n_0, r, m'' >0 .\; \forall n \in \NN_{\geq g(n_0)} .\;\; m'' \cdot
      \sqrt[r]{k}^{\sqrt[d]{n}} \leq f(n) \]
          Finally, $k > 1$ implies $e = \sqrt[r]{k} > 1$ and we obtain
 \[ \exists n_0, m'' >0, e > 1 .\; \forall n \in \NN_{\geq g(n_0)} .\;\; m'' \cdot
 e^{\sqrt[d]{n}} \leq f(n)\]
   which implies
   \[ \exists e > 1 .\;\; f(n) \in \Omega(e^{\sqrt[d]{n}})\] \medskip

   \claim{\Cref{lem:its-to-irc}(c) holds, i.e., $g(n) \in \OO(1)$ and $f(g(n)) \notin
     \OO(1)$ implies
    $f(n) \in \Omega(\omega)$}
   Note that  $f(g(n)) \notin \OO(1)$ means that
   \[ \forall m \in \NN. \; \exists n \in \NN. \;\; f(g(n)) > m \]
   By $g(n) \in \OO(1)$ we have
   \[ \exists m' \in \NN. \; \forall n \in \NN. \; \; g(n) \leq m'\]
   As $f$ is weakly monotonic, we have $f(g(n)) \leq f(m')$ for all $n$. Hence, we get
   \[ \exists m' \in \NN. \;  \forall m \in \NN. \;\; f(m') > m \]
   This implies that $f(m') = \omega$. By weak monotonicity of $f$, we obtain
   \[  \exists m' \in \NN. \; \forall n \geq m'. \; \; f(n) = \omega\]
   which means $f(n) \in \Omega(\omega)$.
       \qedhere
  \end{claims}
\end{proof}

\bibliographystyle{ACM-Reference-Format}
\bibliography{refs,strings,crossrefs}

\end{document}